\documentclass[11pt]{article}

\usepackage[utf8]{inputenc}
\usepackage[T1]{fontenc}
\usepackage{amsmath,amssymb,amsthm,mathtools}
\numberwithin{equation}{section}
\usepackage{enumerate}
\usepackage[colorlinks=true,linkcolor=blue,citecolor=blue]{hyperref}
\usepackage{cleveref}
\usepackage{svg}
\usepackage[margin=1.2in]{geometry}
\usepackage{authblk}
\usepackage{tikz}
\usetikzlibrary{arrows.meta, positioning}
\usepackage{caption}
\usepackage{bbm}

\theoremstyle{plain}
\newtheorem{theorem}{Theorem}[section]
\newtheorem{proposition}[theorem]{Proposition}
\newtheorem{lemma}[theorem]{Lemma}
\newtheorem{corollary}[theorem]{Corollary}

\theoremstyle{definition}
\newtheorem{definition}[theorem]{Definition}

\theoremstyle{remark}

\newcommand{\R}{\mathbb{R}}

\begin{document}

\title{Reconstruction of the Riemann tensor\\from causal order and element counts}
\author[1,2]{Nathan Madsen\,}%\orcidlink{0000-0001-7791-0481}}
\affil[1]{Department of Physics, University of California, Santa Barbara}
\affil[2]{Department of Mathematics, University of California, Santa Barbara}
\date{October 13, 2026}
\maketitle

\begin{abstract}
A causal set is a locally finite partially ordered set. Causal set theory proposes such discrete structures as a basis for spacetime physics. Connecting this description with continuum physics requires recovering the geometric quantities used in general relativity and controlling the errors introduced by discreteness. Here, with dimension and sampling density specified, we construct approximate normal coordinates, estimate parallel transport, and recover the full Riemann tensor from causal order and element counts. For samples embedded in a fixed smooth spacetime and satisfying stated uniform counting assumptions, we prove a uniform curvature error bound that tends to zero as the sampling density increases.
\end{abstract}

% ============================================================
\section{Introduction}\label{sec:intro}
% ============================================================

The causal set program~\cite{Bombelli1987,Surya2019} rests on the idea that order and number together determine geometry. Its continuum motivation comes from reconstruction results of Hawking, King, and McCarthy~\cite{Hawking1976} and Malament~\cite{Malament1977}. For past- and future-distinguishing spacetimes of dimension greater than two, chronological order---the relation defined by future-directed timelike curves---determines the topology, smooth structure, and conformal metric; see also~\cite[Sec.~4.3.4]{Minguzzi2019Causality}. The spacetime volume measure fixes the remaining conformal factor. In a causal set, element counts approximate this volume measure.

Bombelli~\cite{Bombelli2000Statistical} introduced statistical measures of closeness between finite-volume Lorentzian geometries by comparing the probabilities of obtaining finite causal orders from random samples. Statistical reconstruction has subsequently been studied through rigorous uniqueness theorems. Braun~\cite{Braun2025OrderNumber} samples $k$ points independently from the normalized spacetime volume measure and records a matrix whose $(i,j)$ entry is one precisely when the $i$th point chronologically precedes the $j$th. This sampling law is equivalent to a Poisson sprinkling conditioned to contain exactly $k$ points. For spacetimes of the same finite total volume, in dimension at least three and under stated causal and completeness assumptions, equality of the probabilities of all such matrices for every $k$ implies that the spacetimes are smoothly isometric. The present paper addresses local reconstruction from a finite sample, with quantitative bounds on coordinates, transport, and curvature.

Finite Minkowski embedding procedures have been developed by Henson~\cite{Henson2006}, Clough and Evans~\cite{CloughEvans2017}, Liu and Reid~\cite{LiuReid2018}, and Johnston~\cite{Johnston2022,Johnston2025}. Johnston's constructions recover coordinates from volume-derived inner products or spatial distances. Related constructions of intrinsic spatial distance include those of Rideout and Wallden~\cite{RideoutWallden2009}, Eichhorn, Surya, and Versteegen~\cite{EichhornSuryaVersteegen2019}, and Bogu\~n\'a and Krioukov~\cite{BogunaKrioukov2024}. Glaser and Surya~\cite{GlaserSurya2013Locality} address intrinsic locality by comparing the abundances of order intervals in a region with those expected for sprinklings into flat Alexandrov intervals. Their work uses these profiles to identify candidate approximately flat regions and test manifoldlikeness. Rovelli~\cite{Rovelli2002GPS} defines physical coordinates in four-dimensional general relativity using clock readings transmitted by light signals from four reference worldlines, so that metric components in these coordinates are observable quantities. Here we obtain approximate normal coordinates from causal order and interval counts, then use coordinate comparisons to estimate transport and curvature.

Quantitative uniqueness of well-conditioned causal-set embeddings was studied in~\cite{Madsen_HV}, using timelike anchor data and Lorentzian trilateration to compare two embeddings. The present paper uses related anchor methods for local reconstruction, with causal profiles first localizing the sample and interval counts then supplying the data for coordinates, transport, and curvature.

One route to curvature uses discrete wave operators. Building on Sorkin's earlier construction~\cite{Sorkin2007}, Benincasa and Dowker~\cite{BenincasaDowker2010} introduced a four-dimensional causal-set d'Alembertian and used its action on a constant scalar field to propose a scalar-curvature estimator. Belenchia, Benincasa, and Dowker~\cite{BelenchiaBenincasaDowker2016} proved, under their assumptions on the geometry and scalar field, that the mean over Poisson sprinklings converges to $(\Box-\tfrac12 R)\phi$. Extensions to other dimensions and generalized operators were developed by Dowker and Glaser, Glaser, Aslanbeigi, Saravani, and Sorkin, and Belenchia~\cite{DowkerGlaser2013,Glaser2014,AslanbeigiSaravaniSorkin2014,Belenchia2016Universal}.

Roy, Sinha, and Surya~\cite{RoySinhaSurya2013} derived curvature corrections to the expected counts of chains, meaning totally ordered subsets, in small causal diamonds. These yield estimates of the scalar curvature and the Ricci component along the diamond's timelike axis. They also analyzed sampling fluctuations as the sprinkling density increases.

Shuman~\cite{Shuman2024Derivatives,Shuman2024TangentSpaces} combines continuum differential-geometric formulas with discrete estimates of metric components and their derivatives to construct connection and Riemann-curvature estimates. He also fits curvature components from parallel transport around triangular loops. Numerical tests in two-dimensional Minkowski spacetime provide evidence for these constructions; the paper does not establish a uniform curvature error bound in curved spacetime. Here we estimate short-edge transport directly by affine fits between reconstructed normal charts and prove a quantitative consistency bound under stated sampling and geometric hypotheses. The proof controls the propagation of coordinate and transport errors around loops, with distinct scales for the short edges and the loops themselves.

Recent work has approached Ricci curvature through Lorentzian optimal transport. Barton, Borza, and R\"ohrig~\cite{BartonBorzaRohrig2026} define a causal-set curvature using distributions on causal diamonds, with a continuum expansion recovering timelike Ricci curvature and numerical tests on Poisson sprinklings. Braun and Li~\cite{BraunLi2026} give a related continuum construction using distributions on small spacelike slices.

With the dimension, density, and fixed construction thresholds specified, we reconstruct approximate normal coordinates on the sample points, estimate parallel transport along geodesic segments joining nearby sample points by fitting their coordinates on chart overlaps, and estimate curvature from loops at a larger scale. For samples embedded in a fixed smooth spacetime and satisfying the volume conditions (V1) and (V2), we control the errors through this construction and obtain a uniform, density-dependent estimate of the full Riemann tensor. The comparison uses the reconstructed Lorentz frames, with uniformly bounded changes of frame, at sample points satisfying the past- and future-count thresholds of Section~\ref{sec:eligible-bases}. Theorem~\ref{thm:volume-curvature} gives the quantitative bound, and Corollary~\ref{cor:poisson-curvature} gives its high-probability and eventual almost-sure Poisson realizations. Section~\ref{sec:alternative-estimators} describes the corresponding conditional construction for other interval estimators.

The proof separates the geometric reconstruction from its probabilistic realization. We work under the uniform counting conditions (V1) and (V2), and Section~\ref{sec:poisson-volume} shows how Poisson samples satisfy them. A first geometric issue is identifying nearby points: a small proper time can connect endpoints that are far apart along an almost null direction. Section~\ref{sec:geo} therefore begins the reconstruction by comparing which sample points lie to the past or future of each endpoint. These comparisons locate nearby points before interval counts are used to estimate proper time.

Section~\ref{sec:sampled-charts} uses the resulting interval estimates to reconstruct approximate normal coordinates. The finite tests select suitable anchor points and retain overlaps containing sample points spread across independent coordinate directions. This spread keeps the coordinate comparisons used to estimate transport stable.

Section~\ref{sec:transport-curvature} estimates transport along short edges at scale $r$ and combines these estimates around loops of the larger scale $H=\sqrt r$. An $O(r^2)$ error per edge accumulates to $O(Hr)$, while the leading curvature contribution has size $O(H^2)$. The two scales therefore make the accumulated error small compared with the curvature contribution. The chart frames may vary independently. Intermediate changes of frame cancel in the ordered loop product. At the base, changing frame leaves the identity matrix fixed, so the alignment error is multiplied by the small deviation of the loop transport from the identity. Together, these estimates give the curvature bound of Theorem~\ref{thm:volume-curvature}.

% ============================================================
\section{Spacetime geometry and counting functions}\label{sec:geo}
% ============================================================

Let $(M,g)$ be a smooth, time-oriented, globally hyperbolic spacetime of dimension $d \geq 2$, with signature $(-,+,\dots,+)$. Fix a regular compact causal diamond
\[ D = J^+(a) \cap J^-(b) = \overline{I^+(a) \cap I^-(b)}, \quad a \ll b \]
of finite positive volume. We write $\overline\mu=\operatorname{vol}_g$ for the spacetime volume measure on $M$ and $\mu=\overline\mu|_D$ for its restriction to the diamond $D$. In local coordinates $z$, this measure is given by $d\overline\mu=\sqrt{|\det(g_{\alpha\beta}(z))|}\,d^dz$. We call $\overline\mu$ ambient volume to distinguish it from the restricted measure $\mu$; for any measurable set $B\subset M$, $\mu(B)=\overline\mu(B\cap D)$. Fix a smooth future unit timelike field $T$ on a neighborhood of $D$ and define the positive Riemannian metric $h = g+2T^\flat \otimes T^\flat$. We call $h$ the observer metric associated with $T$. Write $d_h$ for its induced distance and $B_h(x,R)=\{y:d_h(x,y)\leq R\}$ for the closed ball of radius $R$ centered at $x$. Such a vector field is guaranteed to exist on a globally hyperbolic manifold by Bernal--S\'anchez~\cite{BernalSanchez2005}, where we take the Cauchy time function $t$ and write $T = -\nabla t / \sqrt{-g(\nabla t,\nabla t)}$. 

Let $C_\rho$ be a finite poset with an embedding $f: C_\rho \to D$ such that the map is injective and the partial order is exactly the induced causal order. For the embedded sample $P = f(C_\rho) \subset D$ we define the count-volume estimator as $\nu = \rho^{-1} \sum_{p\in P} \delta_p$.

A hat indicates an estimate constructed from the finite causal set. For example, the proper-time estimate introduced below is denoted $\hat\tau(p,q)$, while $\tau_g(p,q)$ denotes the continuum proper time between the corresponding sample points.

Then fix a relatively compact open neighborhood $\Omega$ of $D$, with $T$ defined on a neighborhood of $\overline \Omega$. This neighborhood provides room to move the endpoints of small diamonds and to thicken their boundaries. All sufficiently small constructions centered near $D$ remain in $\Omega$, and constants may depend on the fixed geometry on $\overline \Omega$. The sample points lie in $D$; the larger neighborhood $\Omega$ provides room for auxiliary geometric constructions beyond $\partial D$.

We take closures in the manifold topology. The distance induced by the smooth Riemannian metric $h$ gives the same topology, and we use this distance to measure separation from boundaries.

In estimates, $C$ denotes a positive constant depending on the fixed geometry and stated fixed parameters; its value may change between estimates. The notation $C_N$ allows additional dependence on $N$. Constants assigned specific roles, such as $K$, $A_{\mathrm{br}}$, and $A_{\mathrm{cand}}$, retain their chosen values.

We express lengths in a fixed unit, with $\rho$ denoting the corresponding dimensionless density, and set
\[ 
L_\rho=\log(e+\rho), \qquad s_{\min}=\rho^{-1/d}L_\rho^{2/d}. 
\]
The cutoff $s_{\min}$ will be the minimum height of the diamonds used to test density. It is chosen so that these regions shrink as $\rho$ increases while still containing many points on average. Indeed,
\[
 \rho s_{\min}^d=L_\rho^2. 
\]
The volume estimate below therefore implies that the smallest test diamonds have expected Poisson counts comparable to $L_\rho^2$. At this cutoff, the relative tolerance in (V1) is $K/\sqrt{L_\rho}$, which tends to zero.

For sufficiently small $s>0$, define the fixed-axis diamonds
\[ E_s(x) = J^+(\exp_x^g(-sT_x/2)) \cap J^-(\exp_x^g(sT_x/2)). \]
Fix an upper test height $s_0>0$ as supplied by Lemma~\ref{lem:small-diamond-stability}. The same formula defines the comparison diamonds used in its proof at heights up to $2s_0$.

\begin{definition}[Uniform volume conditions]
\label{def:deterministic-volume-event}
	For a fixed finite constant $K$, we say that the embedded sample satisfies the deterministic volume conditions if 
	\[ |\nu(E_s(x)) - \overline \mu (E_s(x))| \leq K \sqrt{\frac{L_\rho}{\rho s^d}}\overline \mu (E_s(x)), \quad s_{\min} \leq s \leq s_0, \quad E_s(x) \subset D, \tag{V1}\label{eq:V1} \]
	\[
	|\nu(D)-\mu(D)| \leq K\sqrt{L_\rho/\rho}.
	\tag{V2}\label{eq:V2}
	\]
	Condition~\eqref{eq:V1} is required simultaneously for every center and test height satisfying the displayed containment and height restrictions. 
\end{definition}

Whenever (V1) is applied, the test diamond must be contained in $D$ and its height must lie in $[s_{\min},s_0]$.

Conditions (V1) and (V2) are a quantitative form of the idea that counting causal-set elements should measure spacetime volume. In the original proposal of Bombelli, Lee, Meyer, and Sorkin~\cite{Bombelli1987}, uniform density was understood to allow the Poisson-type fluctuations expected from a random sprinkling, rather than exact equality between number and volume. 

Condition (V1) formalizes this fluctuation on the regular diamonds $E_s(x)$, whose height controls their extent in every direction. Because a Poisson count with mean $m$ has standard deviation $\sqrt{m}$ and its relative fluctuations decrease as $m^{-1/2}$, the tolerance in (V1) has this square-root counting dependence along with a logarithmic allowance for controlling all the test diamonds simultaneously. Condition (V2) controls the total count in $D$. Together with (V1), it also controls the sample near $\partial D$. 

\subsection{Small diamonds and boundary estimates}

\begin{figure}[htbp]
  \centering
  \includegraphics[width=.80\linewidth]{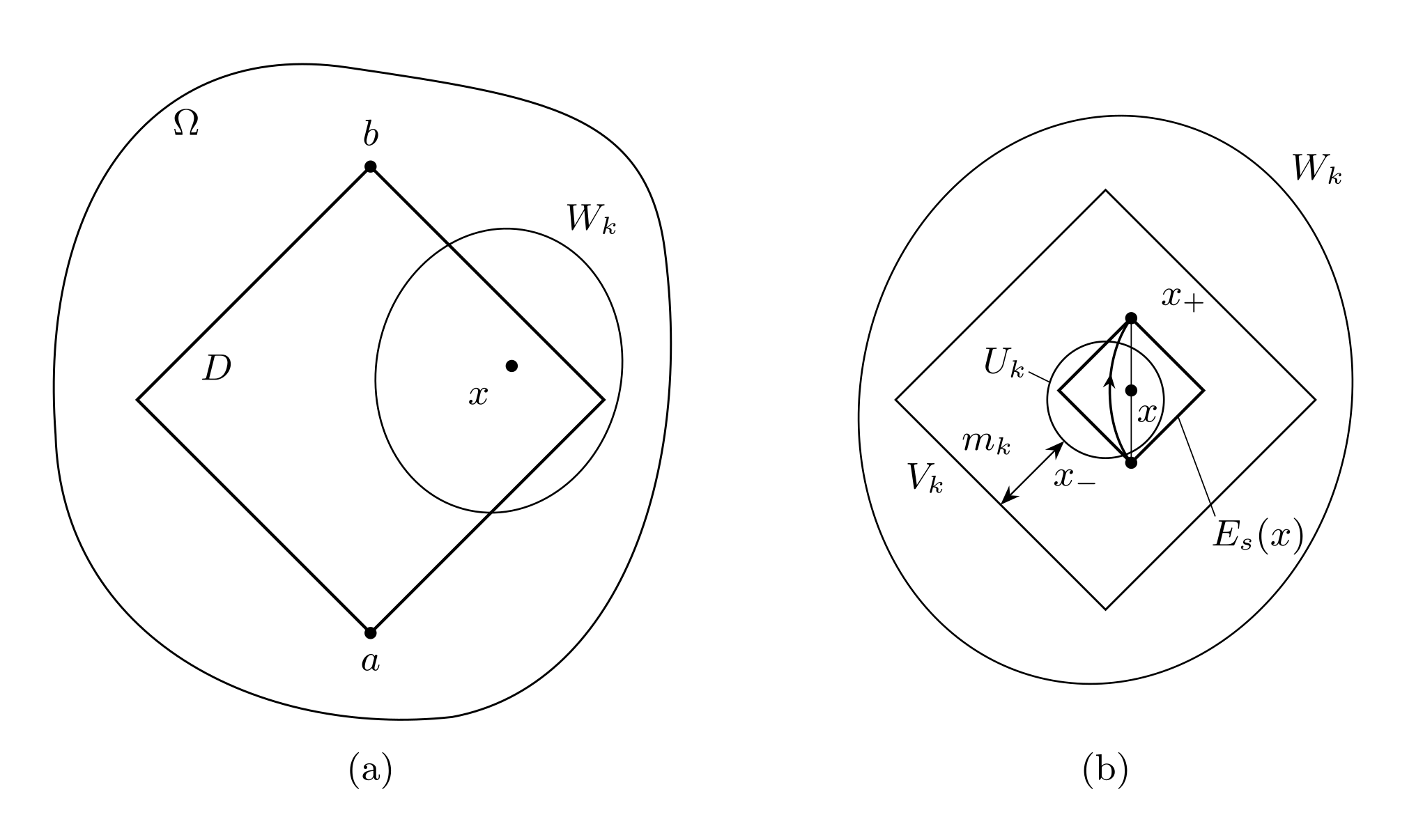}
  \caption{The ambient region and a local neighborhood.
  (a) The sampling diamond $D$ lies inside $\Omega$;
  a normal neighborhood $W_k$ may cross $\partial D$.
  (b) The $h$-distance between $\overline{U_k}$ and
  $\partial V_k$ is $m_k>0$. Once the tips $x_-$ and $x_+$
  of the diamond $E_s(x)$ lie in $V_k$, causal convexity
  of $V_k$ in $M$ ensures that the entire diamond remains
  in $V_k$, and hence in $W_k$. The curved arrow represents
  a causal curve between these tips.}
  \label{fig:local-neighborhoods}
\end{figure}

The next lemma records two useful features of the test diamonds. First, their height $s$ controls both diameter and volume. In flat spacetime, coordinates with time direction $T_x$ put the tips at $(\pm s/2,0)$, and a point $(t,z)$ belongs to the diamond precisely when $|z|\leq s/2-|t|$, where $|z|$ is the Euclidean spatial norm. Thus the temporal and spatial extents both shrink with $s$. For the fixed-axis family defined by $T$, the curved-spacetime estimate is $E_s(x)\subset B_h(x,Cs)$. An arbitrary causal interval whose axis approaches a null direction can instead have small proper-time height while retaining a large $h$-diameter.

The second feature concerns changing the center and height by a total amount at most $\epsilon$, measured using $d_h(x,x')+|s-s'|$. Assume that $\epsilon\leq cs$, with $c>0$ sufficiently small as in the next lemma. This restriction keeps the nearby center $x'$ inside $E_s(x)$ and allows a positive inner comparison height. Increasing the height about $x'$ by a suitable multiple of $\epsilon$ produces a diamond containing $E_s(x)$, while decreasing it produces a diamond contained in $E_s(x)$. The region between these two comparison diamonds has volume at most $C\epsilon s^{d-1}$. Geometrically, this is a thickness of order $\epsilon$ over a boundary whose area, measured with $h$, is of order $s^{d-1}$. This estimate will let us pass from counting bounds on a finite mesh of test diamonds to bounds for every diamond in the family.

The proof below establishes the size and containment statements directly. The technical estimate controlling how a light cone moves with its tip, including at the cone vertex, is proved separately in Appendix~\ref{app:light-cone-stability}.

\begin{lemma}[Size and stability of small diamonds]
\label{lem:small-diamond-stability}
There exist $s_0,c,C,A_{\mathrm{br}}>0$ such that, for every center $x\in\overline\Omega$ and height $0<s\leq s_0$,
	\begin{equation}\label{eq:regular-diamond-size}
B_h(x,cs) \subset E_s(x) \subset B_h(x,Cs), \qquad cs^d \leq \overline \mu(E_s(x)) \leq Cs^d.
\end{equation}
	If also $x'\in\overline\Omega$, $0<s'\leq s_0$, and $d_h(x,x')+|s-s'| \leq \epsilon \leq cs$, then
	\[ E_{s'-A_{\mathrm{br}}\epsilon}(x')\subset E_s(x) \subset E_{s'+A_{\mathrm{br}}\epsilon}(x'), \]
with $s'-A_{\mathrm{br}}\epsilon>0$, and
	\[ \overline \mu(E_{s'+A_{\mathrm{br}}\epsilon}(x') \setminus E_{s'-A_{\mathrm{br}}\epsilon}(x')) \leq C\epsilon s^{d-1}. \]
\end{lemma}

\begin{proof}
By strong causality, choose finitely many nested open neighborhoods $U_k\Subset V_k\Subset W_k$ such that the $U_k$ cover $\overline\Omega$, each $V_k$ is causally convex in $M$, and each $W_k$ is convex normal~\cite[Lemma~3.21 and Remark~3.23]{MinguzziSanchez2008}. Here $U\Subset V$ means that $\overline U$ is compact and contained in $V$. Causal convexity means that every causal curve with endpoints in $V_k$ remains in $V_k$. The sets $U_k$ contain the centers, the larger sets $V_k$ will contain the diamonds, and the sets $W_k$ provide normal coordinates. Choose these coordinate neighborhoods small enough that $\nabla t$ is past timelike and $\partial_t$ is future timelike. Put $m_k=\operatorname{dist}_h(\overline U_k,\partial V_k)>0$.

Smoothness of the tip maps and compactness of $\overline\Omega$ give a preliminary height $s_*>0$ and a constant $C_0>0$ such that, writing $x_\pm=\exp_x^g(\pm sT_x/2)$,
\[
d_h(x,x_\pm)\leq C_0s,
\qquad x\in\overline\Omega,\quad 0<s\leq2s_*.
\]
Choose $0<s_0\leq s_*$ so that $(2C_0+1)s_0<\min_k m_k$. For $x\in U_k$, both tips lie in $V_k$, hence $E_s(x)\subset V_k\subset W_k$ by causal convexity. The same choice covers other centers in $\overline\Omega$ within $s_0$ of $x$ and heights up to $2s_0$.

On compact subsets of these temporal charts, future causal vectors satisfy $|v|_{\mathrm{Eucl}}\leq C\,dt(v)$: normalize their coordinate norms to one and take the positive minimum of $dt(v)$. Integration along causal curves, together with the $O(s)$ tip separation, gives $E_s(x)\subset B_h(x,Cs)$. In coordinates, $x_\pm=x\pm(s/2)T_x+O(s^2)$. Since $T_x$ stays uniformly inside the timelike cones, every point of $B_h(x,cs)$, for sufficiently small $c$, lies on two future timelike coordinate segments from $x_-$ to $x_+$. This proves the reverse inclusion. The volume bounds follow from the uniform upper and lower bounds on the coordinate volume density. Decrease $s_0$ so these estimates also hold up to height $2s_0$.

The tip maps have bounded derivatives on a neighborhood of $\overline\Omega\times[0,2s_0]$. Thus the nearby tips $x'_\pm=\exp_{x'}^g(\pm s'T_{x'}/2)$ satisfy $d_h(x_\pm,x'_\pm)\leq C_1\epsilon$. Uniform local cone comparison gives $\kappa>0$ such that moving a tip backward by proper time $u$ along its central timelike geodesic puts $B_h(x'_\pm,\kappa u)$ to the chronological future of the new tip; moving forward puts this ball to its chronological past. Choose $\kappa A_{\mathrm{br}}/2>C_1$ and $c\leq1/[2(A_{\mathrm{br}}+1)]$, and define
\[
q_\pm=\exp_{x'}^g\bigl(\pm(s'+A_{\mathrm{br}}\epsilon)T_{x'}/2\bigr),
\qquad
r_\pm=\exp_{x'}^g\bigl(\pm(s'-A_{\mathrm{br}}\epsilon)T_{x'}/2\bigr).
\]
Then
\[
q_-\preceq x_-\preceq r_-,
\qquad
r_+\preceq x_+\preceq q_+,
\]
where $\preceq$ denotes causal precedence, including equality. Causal transitivity proves the required diamond inclusions. Moreover, $s'-A_{\mathrm{br}}\epsilon\geq s/2>0$ and $s'+A_{\mathrm{br}}\epsilon\leq2s_0$, so all comparison diamonds lie in the same normal neighborhood by the first paragraph.

In the coordinates $(t,z)$ of $W_k$, the spatial coordinates of every point in either comparison diamond lie in one Euclidean ball of radius $Cs$. We compare their time intervals along the lines $z=\mathrm{constant}$, as illustrated in Figure~\ref{fig:diamond-slices}. Corresponding tips of the two diamonds are $O(\epsilon)$ apart. Refine the finite cover, if needed, to use fixed coordinate cylinders in each $W_k$ satisfying the common-domain hypotheses of Lemma~\ref{lem:arrival-time-stability}: local timelike cone comparison allows a small tip region with a larger time interval. Restrict all four arrival graphs to the same spatial ball of radius $Cs$.

By Lemma~\ref{lem:arrival-time-stability}, the future-cone arrival time $\alpha(z)$ from the past tip and the past-cone arrival time $\beta(z)$ from the future tip each change by at most $C\epsilon$. The slice length is $[\beta(z)-\alpha(z)]_+$, where $[u]_+=\max\{u,0\}$; this also gives zero when the slice is empty. Since the positive-part function is $1$-Lipschitz, the difference in slice lengths is at most $C\epsilon$. The diamonds are nested, and the common spatial ball has coordinate volume at most $Cs^{d-1}$. Integrating the slice differences and using the uniform upper bound on the spacetime volume density proves the required $C\epsilon s^{d-1}$ bound.
\end{proof}

\begin{figure}[htbp]
\centering
\includegraphics[width=.42\linewidth]{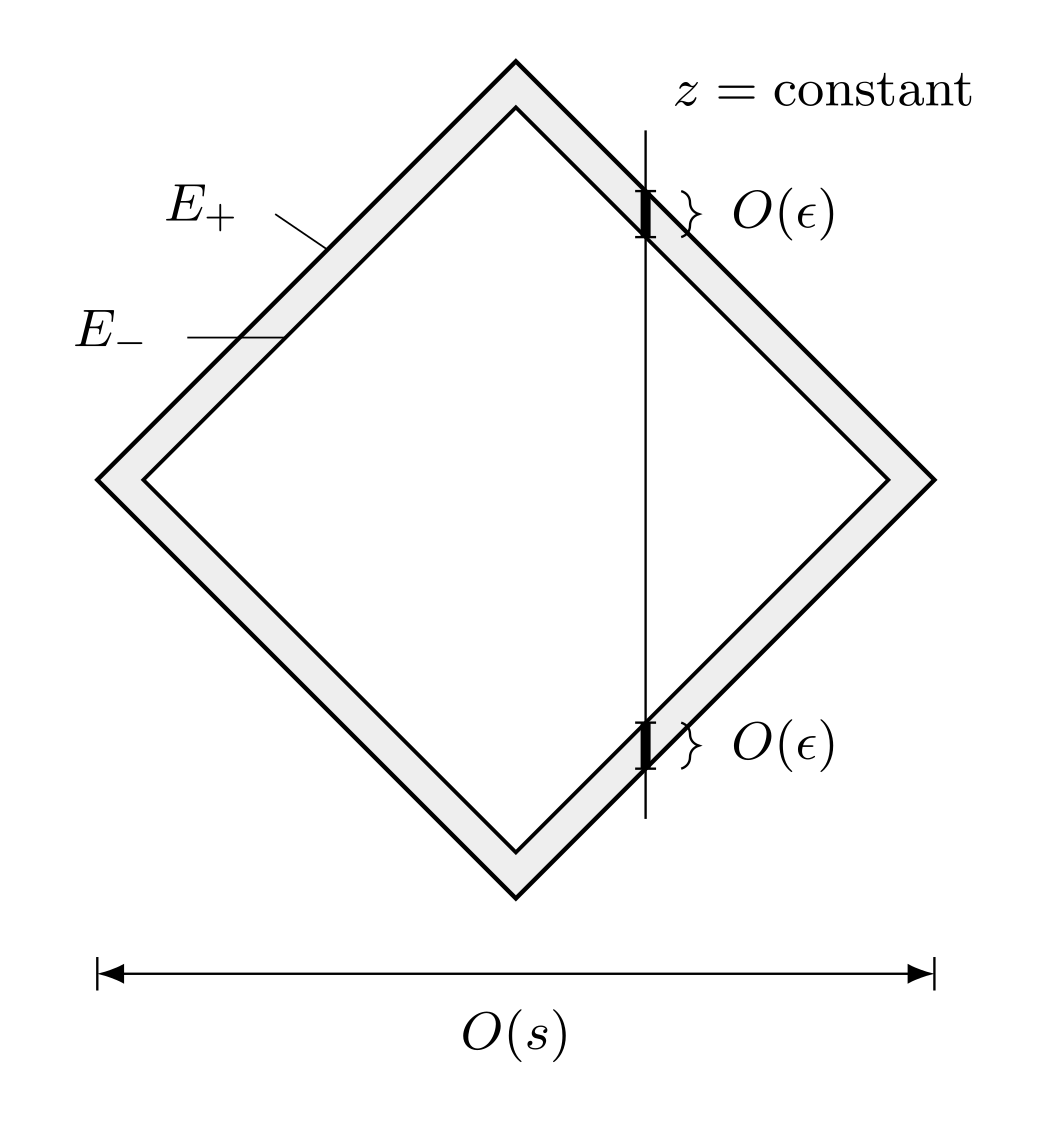}
\caption{The comparison diamonds $E_-=E_{s'-A_{\mathrm{br}}\epsilon}(x')$ and $E_+=E_{s'+A_{\mathrm{br}}\epsilon}(x')$ intersect a coordinate line in nested time intervals. Each extra segment has length $O(\epsilon)$, while the common spatial base has extent $O(s)$. Integrating the slice differences gives the volume bound $C\epsilon s^{d-1}$. The straight boundaries are schematic.}
\label{fig:diamond-slices}
\end{figure}

A set $A\subset M$ is achronal if no two of its points can be joined by a timelike curve. It may contain null segments and need not be a smooth hypersurface. Locally, achronality bounds how steeply its time coordinate can change with spatial position: otherwise two points would be timelike related. This is why a thin neighborhood of an achronal set has controlled volume, even when the set has corners, as illustrated in Figure~\ref{fig:achronal-tube}.

For $w>0$, write
\[
N_w^h(A)=\{p\in M:d_h(p,A)<w\}.
\]
This is the tube of width $w$ around $A$, with distance measured by the positive metric $h$. The tube is allowed to extend outside $D$; its volume is measured using $\overline\mu$, and no additional sample points are assumed. The neighborhood $\Omega$ provides room for these tubes and their covering diamonds.

\begin{figure}[htbp]
\centering
\includegraphics[width=.60\linewidth]{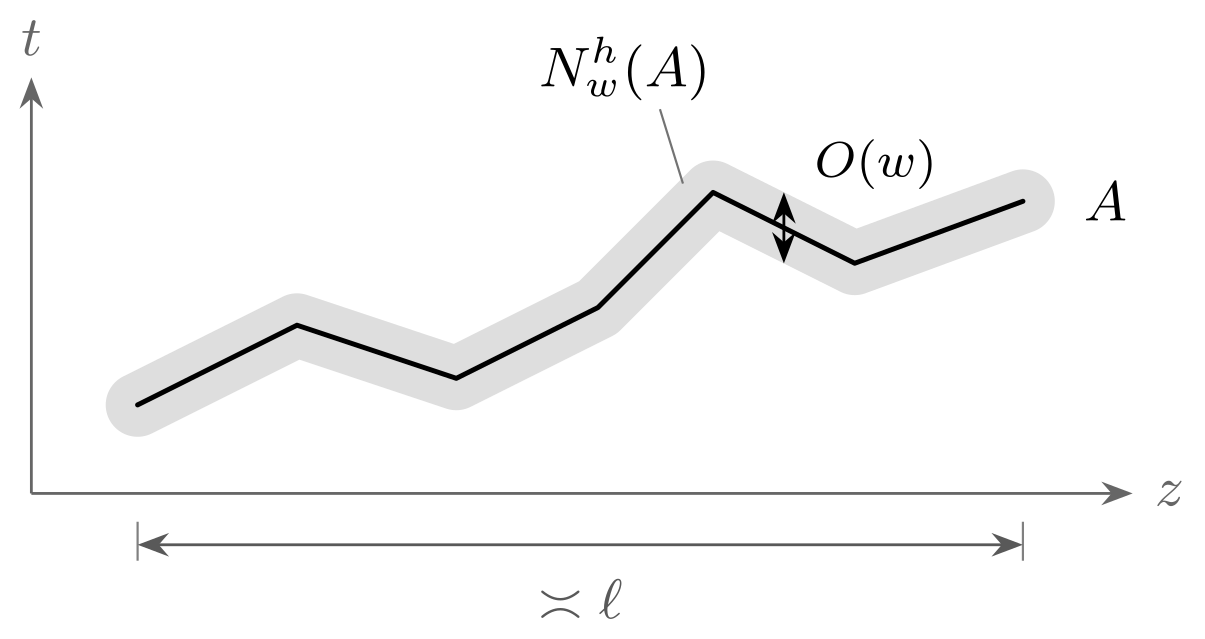}
\caption{An achronal set $A$ and its $h$-neighborhood in a local spacetime chart. Corners and null segments are allowed. Over spatial extent of order $\ell$, the shaded tube has thickness of order $w$, giving the volume bound $Cw\ell^{d-1}$.}
\label{fig:achronal-tube}
\end{figure}

\begin{samepage}
\begin{lemma}[Achronal tube bounds]
\label{lem:achronal-tube-bounds}
There exist $C,\ell_0>0$ such that, for every $x\in D$, $0<w\leq\ell\leq\ell_0$, and achronal set $A\subset B_h(x,\ell)$,
\[
\overline\mu(N_w^h(A))\leq Cw\ell^{d-1}.
\]
The tube can be covered by at most $C(\ell/w)^{d-1}$ diamonds $E_{Cw}(y)$ contained in $\Omega$.
\end{lemma}
\end{samepage}

\begin{proof}
Choose finitely many coordinate boxes $V_k$ whose closures lie in $\Omega$, together with smaller coordinate boxes $U_k$ whose closures lie in $V_k$, so that the $U_k$ cover $D$. A coordinate box is the inverse image of an open rectangular box under a coordinate chart. Choose the larger boxes so that the coordinate time direction is timelike and a uniform narrow cone consists of timelike vectors. Choose $\ell_0$ so that, whenever $x\in D\cap U_k$, the ball $B_h(x,2\ell_0)$ lies in $V_k$. Coordinate distances, $h$-distances, and volume densities have uniform comparison bounds on these boxes.

In this box, achronality implies that $A$ is the graph of a function over part of the spatial coordinates, with a uniform Lipschitz constant: two points violating that bound would be joined by a timelike straight coordinate segment. Extend the function to the full spatial base with the same Lipschitz constant. The tube lies in a vertical strip of thickness $Cw$ about this graph, over a spatial region of volume at most $C\ell^{d-1}$. Integrating the bounded spacetime volume density gives the claimed estimate.

Take coordinate grid cells of side comparable to $w$ meeting the tube. Their disjoint interiors lie in a tube of width $Cw$, which has volume at most $Cw\ell^{d-1}$ by the same graph argument. Each cell has volume at least $cw^d$, so there are at most $C(\ell/w)^{d-1}$ cells. Each is contained in an $h$-ball of radius $Cw$, and hence in a diamond of height $Cw$ by Lemma~\ref{lem:small-diamond-stability}. Choose $\ell_0$ small enough that these heights are at most $s_0$ and all the covering diamonds remain in $\Omega$.
\end{proof}

Covering the fixed diamond $D$ by finitely many balls of fixed radius gives the global bound $\overline\mu(N_w^h(A))\leq Cw$ for any achronal $A\subset D$ and all sufficiently small $w$. If a boundary is contained in a fixed finite union of achronal sets, the bounds follow by summing their tube volumes and covers. The constants then also depend on the number of sets. We will use these consequences for the boundaries of $D$ and of restricted causal pasts, futures, and intervals.

\subsection{Poisson realization of V1 and V2}\label{sec:poisson-volume}

\begin{theorem}[Poisson volume bounds]
\label{thm:poisson-volume-bounds}
For every $N>0$, there is $K_N<\infty$ such that a Poisson sample on $D$ of intensity $\rho\mu$ satisfies (V1) and (V2), with $K=K_N$, with probability at least $1-\rho^{-N}$ for all sufficiently large $\rho$.
\end{theorem}

\begin{proof}
For $X\sim\mathrm{Poisson}(m)$, the centered exponential moment is $\mathbb E e^{u(X-m)}=\exp(m(e^u-1-u))$. The exponential Markov inequality, applied with positive and negative $u$ and optimized, gives
\[
\mathbb P\{|X-m|>2\sqrt{ma}+2a\}\leq2e^{-a},
\qquad a>0.
\]
For the upper tail use $e^u-1-u\leq u^2/[2(1-u/3)]$ for $0<u<3$; for the lower tail use $e^{-u}-1+u\leq u^2/2$. The additive term also covers small means.

Put $\epsilon_\rho=s_{\min}(\rho)\rho^{-2}$. Choose finitely many auxiliary centers $x_j$ in a fixed compact neighborhood of $D$ contained in $\Omega$, so that every center of a test diamond is within $h$-distance $\epsilon_\rho$ of some $x_j$. Such a collection is called an $\epsilon_\rho$-net. Choose also finitely many heights $s_k\in[s_{\min},s_0]$ so that every height in this interval differs from some $s_k$ by at most $\epsilon_\rho$. These heights form the height mesh. The centers and heights are chosen before the Poisson sample is drawn. Coordinate grids in the fixed finite chart cover and a subdivision of the height interval give at most
\[
C\epsilon_\rho^{-(d+1)}\leq C\rho^q, \qquad q=(d+1)(2+1/d).
\]
At each pair $(x_j,s_k)$ use the deterministic comparison sets
\[
B^-_{jk}=E_{s_k-A'\epsilon_\rho}(x_j)\cap D,
\qquad
B^+_{jk}=E_{s_k+A'\epsilon_\rho}(x_j)\cap D.
\]
Figure~\ref{fig:finite-diamond-mesh} illustrates the choice of center and height and the resulting comparison diamonds.

\begin{figure}[htbp]
\centering
\includegraphics[width=.94\linewidth]{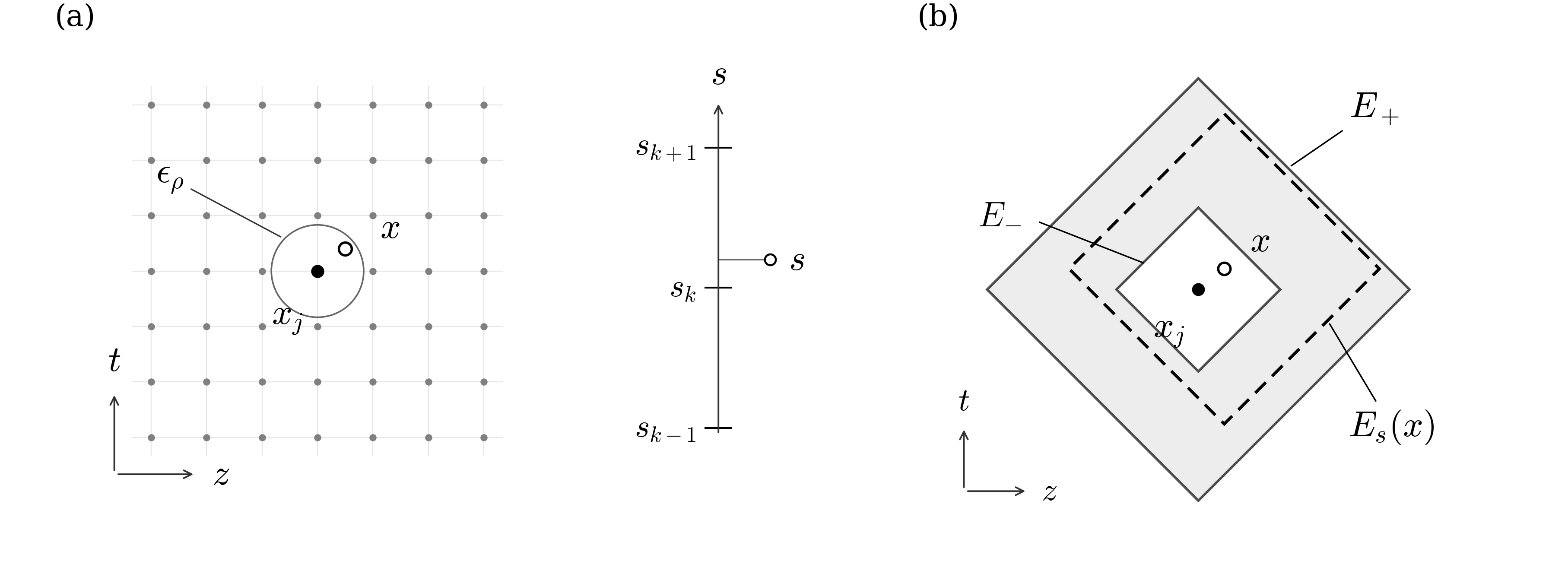}
\caption{Finite approximations to the center and height of a test diamond. (a) Auxiliary grid centers are shown in a flat spacetime chart. The hollow point $x$ lies within $h$-distance $\epsilon_\rho$ of the selected center $x_j$. The separate height scale shows an arbitrary height $s$ and a nearby mesh height $s_k$. These auxiliary centers are chosen before the Poisson sample. (b) The dashed diamond $E_s(x)$ lies between $E_-=E_{s_k-A'\epsilon_\rho}(x_j)$ and $E_+=E_{s_k+A'\epsilon_\rho}(x_j)$. Shading marks $E_+\setminus E_-$. Intersecting the diamonds with $D$ preserves the inclusions. Spacings and displacements are exaggerated.}
\label{fig:finite-diamond-mesh}
\end{figure}
For every test pair $(x,s)$, the chosen approximations satisfy $d_h(x,x_j)+|s-s_k|\leq2\epsilon_\rho$. Choose $A'$ to include this factor of two in the comparison heights supplied by Lemma~\ref{lem:small-diamond-stability}. For large $\rho$, the inner height is positive and the outer height is below $2s_0$. Lemma~\ref{lem:small-diamond-stability} brackets every contained test diamond between these sets and bounds the volume difference by $C\epsilon_\rho s^{d-1}$. Intersection with $D$ preserves both statements and allows the outer comparison diamond to extend beyond the sampled region.

Apply the tail bound to all comparison sets and to $D$ with $a=(N+q+3)L_\rho$. A union bound gives simultaneous success with probability at least $1-\rho^{-N}$ for large $\rho$. Monotonicity between the comparison sets then yields, uniformly over contained test diamonds,
\[
|\nu(E_s(x))-\overline\mu(E_s(x))|
\leq C_N\left(
\sqrt{\frac{s^dL_\rho}{\rho}}
+\frac{L_\rho}{\rho}
+\epsilon_\rho s^{d-1}
\right).
\]
Dividing by $\overline\mu(E_s(x))$, which is bounded above and below by positive multiples of $s^d$, gives
\[ 
\left|\frac{\nu(E_s(x))}{\overline\mu(E_s(x))}-1\right|
\leq C_N\left(
\sqrt{\frac{L_\rho}{\rho s^d}}
+\frac{L_\rho}{\rho s^d}
+\frac{\epsilon_\rho}{s}
\right). 
\]
Since $\rho s_{\min}^d=L_\rho^2$, the second term is absorbed by the first throughout $s_{\min}\leq s\leq s_0$. The third term is at most $\rho^{-2}$ and is also absorbed for $s\leq s_0$ at large $\rho$. Choosing $K_N$ to include these constants proves (V1). The same tail bound on $D$ proves (V2), with its additive $L_\rho/\rho$ term absorbed by $\sqrt{L_\rho/\rho}$.
\end{proof}

To compare samples at different densities, we construct a nested family on one probability space. Let $\mathcal P$ be a Poisson process on $D\times[0,\infty)$ with intensity $\mu\otimes du$. Each point $(x,u)$ carries a mark $u$, an additional nonnegative number recording the density at which the point $x$ enters the sample. At density $\rho$, retain the points with $u\leq\rho$ and project them to $D$. This gives a Poisson sample $P_\rho$ of intensity $\rho\mu$. Increasing $\rho$ adds points while keeping all previously retained points. This joint construction is called a marked-Poisson coupling; it makes statements about eventual behavior as $\rho$ grows refer to a single family of samples. At each density, use $P=P_\rho$ in the definition of $\nu$.

\begin{corollary}[Almost-sure eventual volume bounds]
\label{cor:poisson-eventual-volume-bounds}
Under this coupling, there is a deterministic $K<\infty$ such that (V1) and (V2) hold almost surely for all sufficiently large real $\rho$.
\end{corollary}

\begin{proof}
Fix $N>2$ in the probability estimate of Theorem~\ref{thm:poisson-volume-bounds}. The exponent $N$ controls the decay of the probability of failure as the density increases. The proof also applies at integer density $k$ with lower height $s_{\min}(k+1)$. At this cutoff, $k\,s_{\min}(k+1)^d$ lies between fixed positive multiples of $L_k^2$. Together with the estimate at density $k+1$, a union bound gives probability at most $Ck^{-N}$ that either integer estimate fails, with one deterministic allowance constant. Since $\sum_{k\geq1}k^{-N}<\infty$, the Borel--Cantelli lemma shows that, with probability one, there is an integer $k_0<\infty$ such that both estimates hold for every $k\geq k_0$. The index $k_0$ may depend on the realized nested sample family; the allowance constant remains fixed.

For $\rho\in[k,k+1]$, eventually $s_{\min}(\rho)\geq s_{\min}(k+1)$. Nestedness places the count at $\rho$ between the counts at $k$ and $k+1$. Replacing either integer mean by $\rho\overline\mu(E_s(x))$ costs at most $\overline\mu(E_s(x))$ in count units. After division by $\rho$, this is absorbed by the (V1) tolerance uniformly in $s$, since $s\leq s_0$. The factors $L_k/k$ and $L_\rho/\rho$ are uniformly comparable. The same argument on $D$ proves (V2). Increasing the fixed allowance constant gives both conditions for all sufficiently large real densities.
\end{proof}

\subsection{Ambient averaging and discrepancy}

Condition (V1) controls the counts of regular test diamonds. To estimate counts in pasts, futures, and causal intervals, we average normalized diamond counts over centers in the region, keeping the test diamonds inside $D$. Fubini's theorem lets us evaluate the same average by summing the contributions from individual sample points. Each point inside the region and away from its boundary and $\partial D$ contributes approximately its counting weight $1/\rho$. The boundary estimates control the remaining contributions.

For a fixed height $t$, define
\[
K_t(x,p)=\frac{1_{E_t(x)}(p)}{\overline\mu(E_t(x))}.
\]
The indicator $1_{E_t(x)}(p)$ is one when $p\in E_t(x)$ and zero otherwise. Thus $K_t(x,p)$ assigns the same weight $1/\overline\mu(E_t(x))$ to every point of the diamond and is zero outside it. Integrating against the sample measure means summing over sample points, each with weight $1/\rho$. In particular,
\[
\nu_t(x):=\int_D K_t(x,p)\,d\nu(p)
=\frac1\rho\sum_{p\in P}K_t(x,p)
=\frac{\#(P\cap E_t(x))}{\rho\,\overline\mu(E_t(x))}.
\]
Here $d\nu(p)$ specifies the measure used in the integral. The quantity $\nu_t(x)$ is the sample count in the diamond divided by $\rho$ times its continuum volume. For a Poisson sample, the denominator is the expected count. Condition (V1) therefore says that this ratio is close to one.

Write $\delta_t=K\sqrt{L_\rho/(\rho t^d)}$ for the relative tolerance in (V1), and $e_{tot}=|\nu(D)-\mu(D)|$ for the total-volume error, bounded by (V2).

The terms row and column refer to the two variables of $K_t(x,p)$, as they would in a matrix with centers indexing rows and sample points indexing columns. A row calculation fixes the center $x$ and sums over the points in $E_t(x)$, giving $\nu_t(x)$. A column calculation fixes a point $p$ and varies the center $x$. It integrates over all centers whose diamonds contain $p$, with the volume weights
\[
\int_\Omega K_t(x,p)\,d\overline\mu(x)
=\int_{\{x\in\Omega:p\in E_t(x)\}}
\frac{d\overline\mu(x)}{\overline\mu(E_t(x))}.
\]
The centers in this integral range over the continuum neighborhood $\Omega$. The integral is a weighted volume of centers, rather than a count of finitely many balls. Lemma~\ref{lem:ambient-column-normalization} shows that it is close to one. Fubini's theorem relates the two calculations. For a measurable set of centers $A\subset\Omega$,
\[
\int_A\nu_t(x)\,d\overline\mu(x)
=\frac1\rho\sum_{p\in P}\int_A K_t(x,p)\,d\overline\mu(x).
\]

Under (V1), the row estimate is $|\nu_t(x)-1|\leq\delta_t$ whenever $E_t(x)\subset D$ and $t\in[s_{\min},s_0]$. The following geometric estimate controls the corresponding column integral.

The coordinate comparison below uses two set operations. For a closed set $Q\subset\R^d$, let $B_r$ be the closed Euclidean ball of radius $r$ centered at the origin, and define
\[
Q\ominus B_r=\{z:z+B_r\subset Q\},\qquad Q+B_r=\{q+u:q\in Q,\ u\in B_r\}.
\]
The erosion $Q\ominus B_r$ retains the centers whose radius-$r$ ball fits inside $Q$; the thickening $Q+B_r$ adds points within distance $r$ of $Q$. These balls and distances are Euclidean in the coordinate model, as illustrated in Figure~\ref{fig:erosion-dilation}.

\begin{figure}[htbp]
\centering
\includegraphics[width=.75\linewidth]{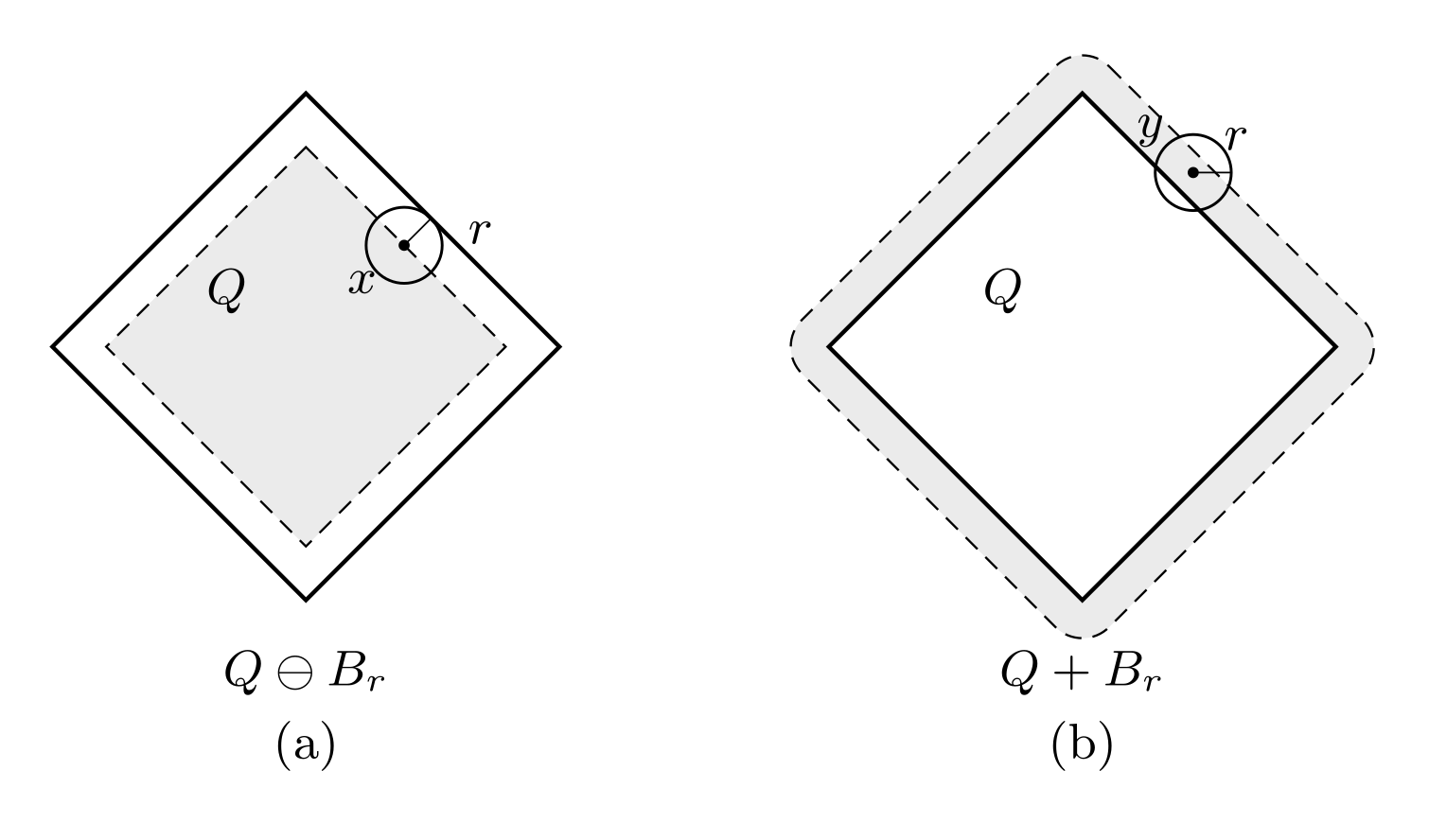}
\caption{(a) Erosion retains the centers $x$ of radius-$r$ balls contained in $Q$. (b) Thickening adds points $y$ whose radius-$r$ ball meets $Q$. The solid boundary belongs to $Q$ and the dashed boundary to the eroded or thickened set. Gray marks the retained region in (a) and the added region in (b). Balls and distances are Euclidean in the coordinate model.}
\label{fig:erosion-dilation}
\end{figure}

\begin{lemma}[Ambient column normalization]
\label{lem:ambient-column-normalization}
	Let $\Omega_0$ be an open neighborhood of $D$ with $\overline\Omega_0\subset\Omega$. For sufficiently small $t$, integrating over centers in $\Omega$ gives, uniformly for $p\in\overline\Omega_0$,
	\begin{equation}\label{eq:ambient-column-normalization}
w_t(p) := \int_\Omega K_t(x,p)d\overline \mu(x) = 1+O(t)
\end{equation}
\end{lemma}
\begin{proof}
	Use a fixed smooth coordinate chart near $p$, write $d\overline \mu(x) = a(x)dx$, and put $x=p+tz$. The support estimate in \eqref{eq:regular-diamond-size} confines contributing $z$ to a fixed bounded set. Let $Q_p$ be the unit-height tangent diamond for $g(p),T(p)$ in these coordinates. It has uniform positive inradius, bounded diameter, and a uniformly Lipschitz boundary. The rescaled metrics and tips for $E_t(p+tz)$ differ from the translated flat data $z+Q_p$ by $O(t)$, uniformly for bounded $z$. 
	The arrival-graph argument in Lemma~\ref{lem:arrival-time-stability}, now also using the $O(t)$ upper and lower constant-cone comparison for the rescaled metric, places the rescaled diamond between the inner and outer parallel sets of $z+Q_p$ at distance $Ct$. In particular, the set of contributing centers $z$ lies between the erosion $-Q_p\ominus B_{Ct}$ and the thickening $-Q_p+B_{Ct}$. Their volume difference is $O(t)$ by the graph tube bound. Also, uniformly in these centers,
	\[ \overline \mu(E_t(p+tz))=a(p)t^d|Q_p|(1+O(t)), \qquad d\overline \mu(p+tz)=a(p)t^d(1+O(t))dz. \]
	Integrating the quotient over the contributing-center set proves \eqref{eq:ambient-column-normalization}. Every integral here is ambient.
\end{proof}

The mass assigned to a set by a measure is its total weight. For the sample measure $\nu$, each point has weight $1/\rho$, so the sample mass of a set $A$ is $\nu(A)=\#(P\cap A)/\rho$.

\begin{lemma}[Boundary mass and global set discrepancy]
\label{lem:global-set-discrepancy}
	Assume (V1) and (V2), with $\delta_t\leq 1/2$. Suppose $t$ and every fixed multiple of $t$ used as a test-diamond height in the counting covers lie in $[s_{\min},s_0]$. For a sufficiently large fixed $b$,
	\begin{equation}\label{eq:boundary-mass}
\nu\{p\in D: d_h(p,\partial D) \leq bt\} \leq C(t+\delta_t +e_{tot}).
\end{equation}
	Uniformly over restricted pasts, futures, their symmetric differences, and sets formed by a fixed finite number of unions, intersections, and differences of these sets,
	\begin{equation}\label{eq:global-set-discrepancy}
|\nu(A)-\mu(A)| \leq C(t+\delta_t + e_{tot}).
\end{equation}
	The constant for these set operations depends on the fixed number of constituent sets, not on their endpoints. 
\end{lemma}
\begin{proof}
Put $D_b(t)=\{p\in D:d_h(p,\partial D)>bt\}$. Choose a fixed $C_0>0$ such that $E_t(x)\subset B_h(x,C_0t)$, as supplied by Lemma~\ref{lem:small-diamond-stability}. Thus $K_t(x,p)$ can be nonzero only when $d_h(x,p)\leq C_0t$. For centers in $D_{b+C_0}(t)$, (V1) applies and all contributing samples lie in $D_b(t)$. Fubini and Lemma~\ref{lem:ambient-column-normalization} give
\[
(1-\delta_t)\mu(D_{b+C_0}(t))\leq\int_{D_{b+C_0}(t)}\nu_t\,d\overline\mu\leq(1+Ct)\nu(D_b(t)).
\]
Since $\mu(D\setminus D_{b+C_0}(t))\leq Ct$ by Lemma~\ref{lem:achronal-tube-bounds}, this yields $\nu(D_b(t))\geq\mu(D)-C(t+\delta_t)$. Subtracting from $\nu(D)\leq\mu(D)+e_{tot}$ proves \eqref{eq:boundary-mass}.

Away from $\partial D$, Lemma~\ref{lem:achronal-tube-bounds} covers the $Ct$-neighborhood of $\partial A$ by $O(t^{1-d})$ regular diamonds of height $C't$. Choose $b$ large enough that these diamonds lie in $D$. Their total sample mass is $O(t)$ by (V1); the portion near $\partial D$ is bounded by \eqref{eq:boundary-mass}.

Compare $\nu(A)$ with $\int_{A\cap D_{C_0}(t)}\nu_t\,d\overline\mu$. By Fubini, a sample farther than $Ct$ from both boundaries contributes $1+O(t)$ if it lies in $A$, and zero otherwise. Only samples within distance $Ct$ of $\partial A$ or $\partial D$ can contribute differently, as illustrated in Figure~\ref{fig:averaging-boundary}; their mass was bounded above. The accumulated column error is $O(t)$ by (V2). The row estimate compares the integral with $\mu(A\cap D_{C_0}(t))$ with error $C\delta_t$, and removing the continuum boundary strip costs $Ct$. Combining these bounds proves \eqref{eq:global-set-discrepancy}.
\end{proof}

\begin{figure}[htbp]
\centering
\includegraphics[width=.45\linewidth]{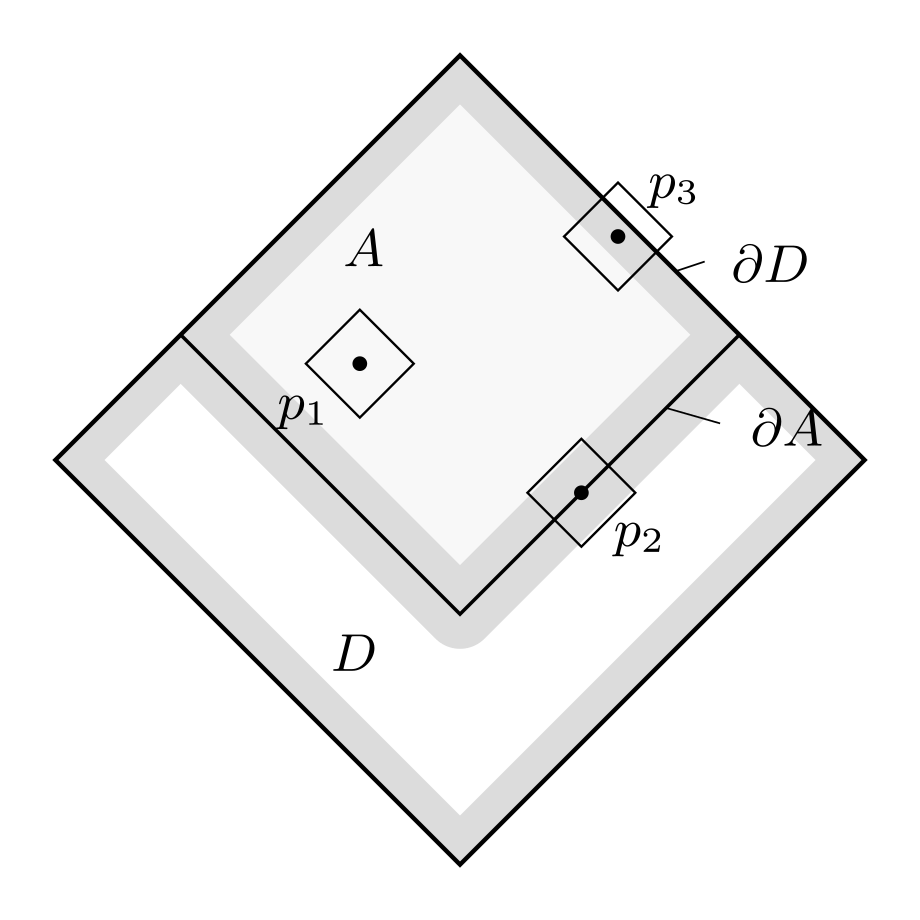}
\caption{Shading marks neighborhoods of the boundaries of $A$ and $D$. In this flat schematic, each thin diamond shows centers whose averaging diamond contains the marked sample. All contributing centers for $p_1$ lie in $A$ and away from $\partial D$. For $p_2$ or $p_3$, some lie across a boundary, so the contribution to the integral over $A\cap D_{C_0}(t)$ may differ. The proof bounds the sample mass in these boundary neighborhoods; exterior centers are not additional sample points.}
\label{fig:averaging-boundary}
\end{figure}

Let $K\subset\mathrm{int}D$ be a fixed compact set. Its points have a uniform positive $h$-distance from $\partial D$. For every fixed $L>0$, the balls $B_h(x,Ls)$ centered at $x\in K$ therefore remain in $D$ for sufficiently small $s$.

\begin{lemma}[Local discrepancy for sets and Lipschitz functions]
\label{lem:local-discrepancy}
	Assume (V1). Fix a compact set $K\subset\mathrm{int}D$. Let $U\subset B_h(x_0,L_0\ell)$ be a coordinate region, with $x_0\in K$. Assume the enlarged ball $B_h(x_0,L_1\ell)$ lies in the same coordinate chart and in $D$, where $L_1>L_0>0$ are fixed and the enlarged ball contains the averaging and covering diamonds used below. Suppose $s_{\min}\leq t\leq c\ell$, $\delta_t\leq1/2$, and every test-diamond height used in the averaging and counting covers lies in $[s_{\min},s_0]$. If $A\subset U$ and its local boundary tubes have volume at most $Ct\ell^{d-1}$, then
	\begin{equation}\label{eq:local-set-discrepancy}
| \nu(A)-\mu(A)| \leq C\ell^d \left(\delta_t + \frac{t}{\ell} \right).
\end{equation}
	For a Lipschitz test function supported in $U$ and extended by zero, one also has
	\begin{equation}\label{eq:smooth-test-discrepancy}
\left| \int \phi d\nu - \int \phi d\mu \right| \leq C \ell^d \left( \delta_t + \frac{t}{\ell}\right) (\|\phi\|_\infty + \ell \,\mathrm{Lip}\,\phi ).
\end{equation}
\end{lemma}
\begin{proof}
	Since $K_t(x,p)\ne0$ implies $d_h(x,p)\leq C_0t$ and $t\leq c\ell$, the contributing samples lie in $B_h(x_0,L_2\ell)$ for a fixed $L_2>L_0$. A finite regular-diamond cover of this ball, with heights proportional to $\ell$, gives $\nu(B_h(x_0,L_2\ell))\leq C\ell^d$. The enlarged ball $B_h(x_0,L_1\ell)$ in the hypotheses contains these covering diamonds. The count tolerances for this cover are bounded by a constant multiple of $\delta_t$. Cover the boundary tube by $C(\ell/t)^{d-1}$ contained regular diamonds at scale $C't$. Their total empirical mass is at most $Ct\ell^{d-1}$. Apply Fubini to $\int_A \nu_t d\overline \mu$. Interior samples contribute their column weight, and the exceptional samples lie in that tube. The column error is at most $Ct\ell^d$ and the row error at most $C\delta_t \ell^d$. In fixed units with $\ell \leq 1$, the column term is absorbed by $Ct\ell^{d-1}$. This proves \eqref{eq:local-set-discrepancy}.
	
	For \eqref{eq:smooth-test-discrepancy}, apply Fubini to $\int \phi(x)\nu_t(x)d\overline\mu(x)$. Replacing $\phi(x)$ by $\phi(p)$ in each column costs at most $Ct\mathrm{Lip}(\phi)$ per unit column mass. Only samples in the covered ball $B_h(x_0,L_2\ell)$ contribute. Replacing the remaining column weight by one costs $Ct\|\phi\|_\infty$ per sample mass. Together with the row error, this bounds the discrepancy by $C\ell^d [(\delta_t + t) \| \phi \|_\infty + t\mathrm{ Lip} \phi ]$, which implies the claimed estimate.
\end{proof}

Local causal intervals, even arbitrarily thin ones, satisfy the tube hypothesis of Lemma~\ref{lem:local-discrepancy}: their boundaries are contained in two achronal boundaries. 

\begin{corollary}[A deterministic sample net]
\label{cor:sample-net}
	Assume (V1). For every fixed compact set $K\subset\mathrm{int}D$ and sufficiently large $\rho$, every $x\in K$ lies within $u_\rho=Cs_{\min}(\rho)$ of a sample point, with distance measured by $h$. 
\end{corollary}
\begin{proof}
	For each $x\in K$, the regular diamond $E_{s_{\min}}(x)$ is contained in $D$ at sufficiently large density, lies in $B_h(x,Cs_{\min})$, and has positive count by (V1) once $\delta_{s_{\min}}<1$. Thus $B_h(x,Cs_{\min})$ contains a sample point. All constants are uniform on $K$. 
\end{proof}

Figure~\ref{fig:sample-net} illustrates the covering property supplied by Corollary~\ref{cor:sample-net}.
\begin{center}
\begin{minipage}{\linewidth}
    \centering
    \captionsetup{hypcap=false}
    \includegraphics[width=.64\linewidth]{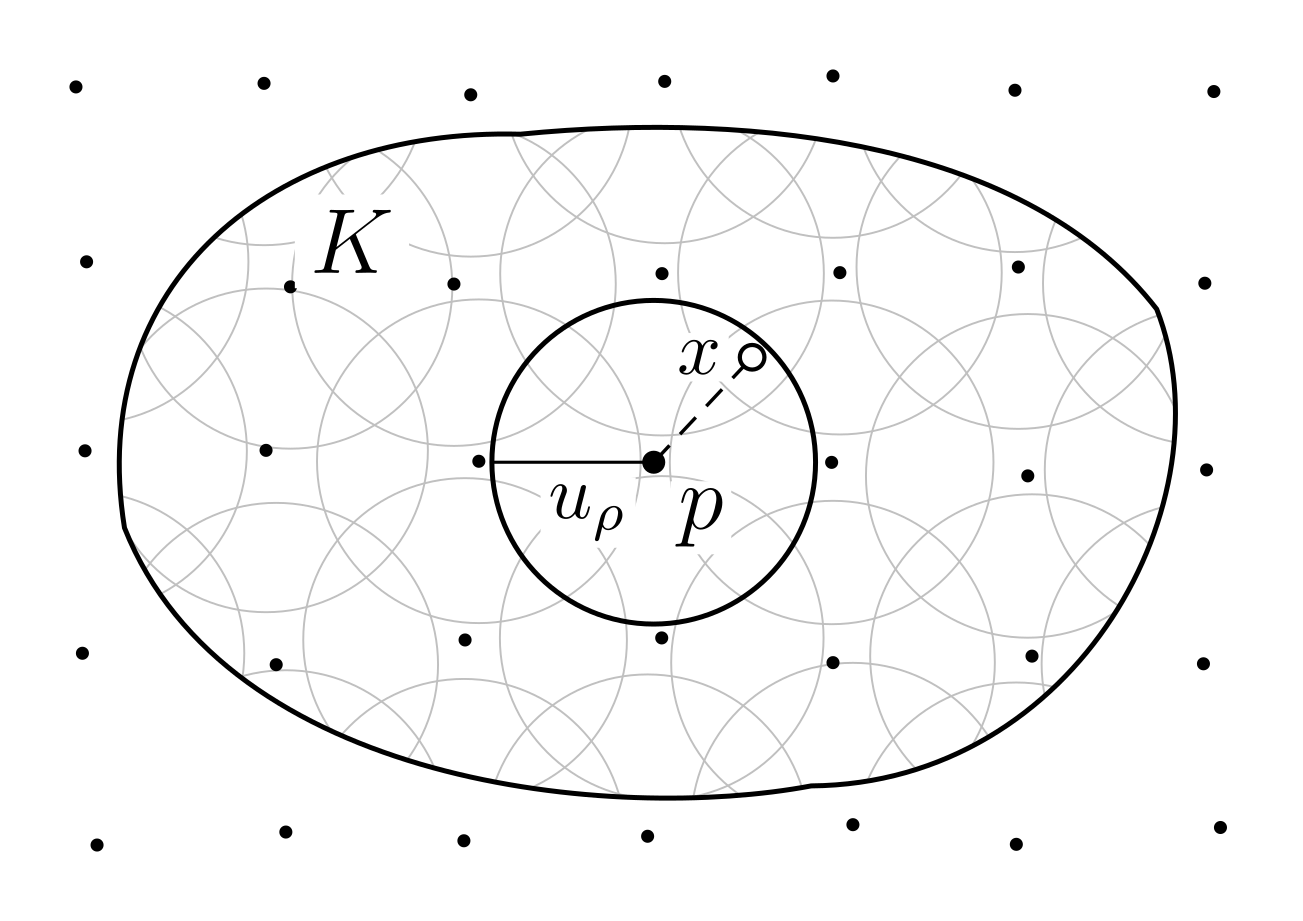}
    \captionof{figure}{A compact set $K\subset\mathrm{int}D$ covered by $h$-balls of radius $u_\rho$ centered at sample points. For every $x\in K$, some sample $p$ satisfies $d_h(x,p)\leq u_\rho$. Filled points are samples and the hollow point is a continuum location. Only the portions of the light balls inside $K$ are drawn. Balls are circular in this flat illustration.}
    \label{fig:sample-net}
\end{minipage}
\end{center}

\subsection{Causal profiles control the observer distance}

For two sets $A,B$, their symmetric difference is $A\triangle B=(A\setminus B)\cup(B\setminus A)$, the set of points belonging to exactly one of them. Define
\[ V^\pm(p) = \mu(J^\pm(p) \cap D), \qquad v(p)=\min\{V^-(p), V^+(p)\}, \]
\[ d_P(p,q)=\mu((J^+(p)\triangle J^+(q))\cap D) + \mu((J^-(p)\triangle J^-(q))\cap D). \]
Thus $d_P(p,q)$ is the sum of the future- and past-disagreement volumes in $D$, as illustrated in Figure~\ref{fig:causal-profiles}. It compares membership, not merely the total past and future volumes. The intersection with $D$ in these formulas is retained for clarity, although $\mu$ is already restricted.

\begin{figure}[htbp]
\centering
\includegraphics[width=.80\linewidth]{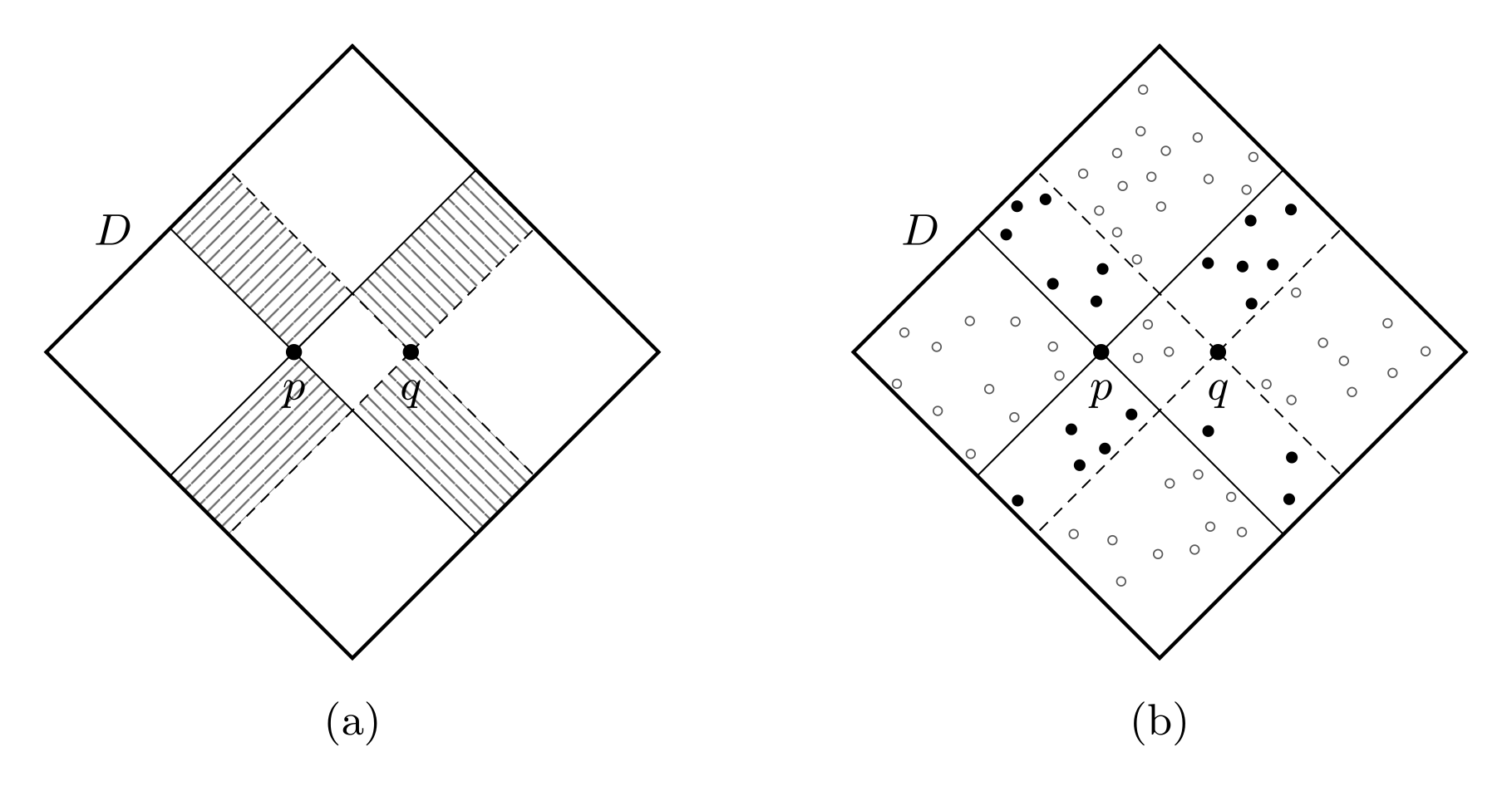}
\caption{(a) Hatching marks points whose future or past relation differs between $p$ and $q$. (b) Filled sample points witness this disagreement; open circles have the same relation to both endpoints. Solid and dashed null boundaries belong to $p$ and $q$, respectively. The larger dots mark the endpoints themselves. All relations are restricted to $D$.}
\label{fig:causal-profiles}
\end{figure}

The smooth weights used below are auxiliary tools in the proof. A weight equal to one near $D$ lets us apply the regularity result for weighted volumes, while local weights isolate small cone patches without introducing sharp edges. Neither changes the counting measure $\nu$.

\begin{proposition}[Differentiability of weighted causal volumes]
\label{prop:weighted-causal-volumes}
	In a globally hyperbolic spacetime with smooth metric, if $0\leq \phi \in C_c^\infty(M)$, then
	\[ F_\phi^{\pm}(x) = \int_{J^\pm(x)}\phi d\overline \mu \]
	is continuously differentiable. This is the compact-support case of Chru\'sciel--Grant--Minguzzi~\cite[Lemmas~3.1 and~3.3]{ChruscielGrantMinguzzi2016}. Their derivative formula is expressed by variation of null generators. We use only the stated $C^1$ conclusions for the global upper estimate below.
\end{proposition}

\begin{lemma}[Comparison of causal-profile and observer distances]
\label{lem:causal-profile-distance}
	The functions $V^\pm$ are Lipschitz on $D$. For every compact $K \Subset \mathrm{int} D$ there are $c_K, C_K > 0$ such that
	\begin{equation}\label{eq:profile-distance-comparison}
c_K d_h(p,q) \leq d_P(p,q) \leq C_K d_h(p,q), \qquad p\in K, \quad q\in D.
\end{equation}
	Moreover, $K_\theta = \{ p\in D : v(p) \geq \theta \}$ is a compact subset of $\mathrm{int} D$ for each $\theta > 0$.
\end{lemma}
\begin{proof}
	Choose a smooth compactly supported function $\phi$, with $0\leq\phi\leq1$, equal to one on a neighborhood of $D$. By Proposition~\ref{prop:weighted-causal-volumes}, its weighted past and future volumes have bounded derivatives on a compact neighborhood of $D$. If $d_h(p,q)=\epsilon$ is small, move a fixed multiple of $\epsilon$ along the $T$-flow to the past and future of $p$ to obtain points $p_-$ and $p_+$. Choose the multiplier uniformly large enough that both $p$ and $q$ lie chronologically after $p_-$ and before $p_+$, so that $p_-\ll p,q\ll p_+$. Since $h(T,T)=1$, the two flow segments give $d_h(p_-,p_+)\leq C\epsilon$. The future symmetric difference of $p,q$ restricted to $D$ is contained in $J^+(p_-) \setminus J^+(p_+)$; the past version is contained in $J^-(p_+) \setminus J^-(p_-)$. The weighted-volume derivative bound gives $d_P(p,q) \leq C\epsilon$. The same argument bounds $|V^\pm(p)-V^\pm(q)|$. 

For the remaining pairs, fix $\epsilon_0>0$ below which the local argument applies. If $d_h(p,q)\geq\epsilon_0$, the bound $d_P(p,q)\leq2\mu(D)$ gives
\[
d_P(p,q)\leq\frac{2\mu(D)}{\epsilon_0}\,d_h(p,q).
\]
Thus the upper bound and Lipschitz assertions hold globally on $D$. 
	
	We give the local lower argument explicitly. At an interior point $p_0$, choose $d$ independent past-directed null covectors. Realize their directions by $d$ short null segments ending at $p_0$, entirely inside one convex normal neighborhood in $\mathrm{int} D$. Choose small, disjoint neighborhoods about points on these segments, away from their tips, and smooth weights $0 \leq \phi_j \leq 1$ supported there, each equal to one on a smaller neighborhood meeting its null segment. 

Use coordinates $(t,z)$ in the chosen normal neighborhood, with $\partial_t$ future timelike, $t$ the time coordinate, and $z$ the spatial coordinates, and write $d\overline\mu=a(t,z)\,dt\,dz$. Each patch stays away from the cone vertex. For $x$ near $p_0$, its past cone is therefore a smooth graph $t=\beta_x(z)$. At a fixed spatial position $z$, the value $\beta_x(z)$ is the time at which the vertical coordinate line meets the past cone of $x$. Within the patch, points with $t\leq\beta_x(z)$ lie in $J^-(x)$. Figure~\ref{fig:cone-arrival-graph} shows this graph. Its weighted volume is
	\[ W_j(x) = \int dz \int_{-\infty}^{\beta_x(z)}\phi_j(t,z)a(t,z) dt .\]
	Hence
	\[ DW_j(x) = \int \phi_j(\beta_x(z),z)a(\beta_x(z),z)D_x \beta_x (z) dz. \]
	To compute how $\beta_x(z)$ changes with $x$, describe the cone by Synge's world function. Within a convex normal neighborhood,
\[
\sigma(x,y)=\tfrac12 g_x(\log_x^g y,\log_x^g y),
\]
where $\log_x^g y=(\exp_x^g)^{-1}(y)$ is the initial velocity of the unique geodesic from $x$ to $y$, parameterized on $[0,1]$. The world function is half the signed squared geodesic interval and vanishes for null-related points. On the selected past-cone patch, $y=(\beta_x(z),z)$ therefore satisfies
\[
\sigma\bigl(x,(\beta_x(z),z)\bigr)=0.
\]
Implicit differentiation, together with the endpoint-variation formula, gives
\[
D_x\beta_x(z)
=-\frac{D_x\sigma(x,y)}{\partial_t\sigma(x,y)}
=\frac{g_x(\log_x^g y,\cdot)}{\partial_t\sigma(x,y)}.
\]
Here $\partial_t$ differentiates the time coordinate of the second argument. The denominator is nonzero because a timelike coordinate line is transverse to the null cone away from its vertex. On this past-cone patch it is positive, so $D_x\beta_x(z)$ is a positive multiple of the initial null covector. The first-variation formula and this transversality argument are given in Appendix~\ref{app:light-cone-stability}. By shrinking the patches, the $d$ row covectors $DW_j(p_0)$ become arbitrarily close, after positive row rescaling, to the selected independent null covectors. They are therefore independent. Such $d$ directions exist because one may choose $d$ affinely independent points on the spatial unit sphere. 
	
	The map $W=(W_1, \dots , W_d)$ is a local coordinate system. After shrinking to a convex coordinate ball, its derivative stays close to one invertible matrix, so $|W(p)-W(q)| \geq cd_h(p,q)$ there. On the other hand, $|W_j(p)-W_j(q)| \leq d_P(p,q)$, since each weight is at most one and supported in $D$. A finite nested cover of $K$ gives the desired uniform lower bound near the diagonal. This local construction uses smooth cone patches before any cut point, not a global smoothness claim about null boundaries.
	
	For pairs away from the diagonal, first show that $d_P(p,q)>0$ whenever $p$ is interior and $q\ne p$. If $p\not\leq q$, a small neighborhood of $p$ avoids the closed set $J^-(q)$ and contains a nonempty open subset of $I^-(p)\cap D$. These points belong to the past of $p$ but not to the past of $q$. If $p\leq q$ and $p\ne q$, causality gives $p\notin J^+(q)$; a small neighborhood of $p$ avoids $J^+(q)$ and contains a nonempty open subset of $I^+(p)\cap D$. These points belong to the future of $p$ but not to the future of $q$. In either case the disagreement has positive volume. Continuity of $d_P$ and compactness give a positive minimum on pairs in $K\times D$ whose $h$-distance is at least any fixed $\epsilon_0>0$. Since $d_h$ is bounded on this product, this yields the lower estimate for the remaining pairs and completes \eqref{eq:profile-distance-comparison}.
	
	Finally, on $\partial D$ at least one of the restricted past and future lies in a null boundary and has zero volume by Lemma~\ref{lem:achronal-tube-bounds}. Thus $v=0$ there. Continuity of $v$ and compactness of $D$ imply that $K_\theta$ is a compact subset of $\mathrm{int} D$ for every $\theta>0$.
\end{proof}

\begin{figure}[htbp]
\centering
\includegraphics[width=.62\linewidth]{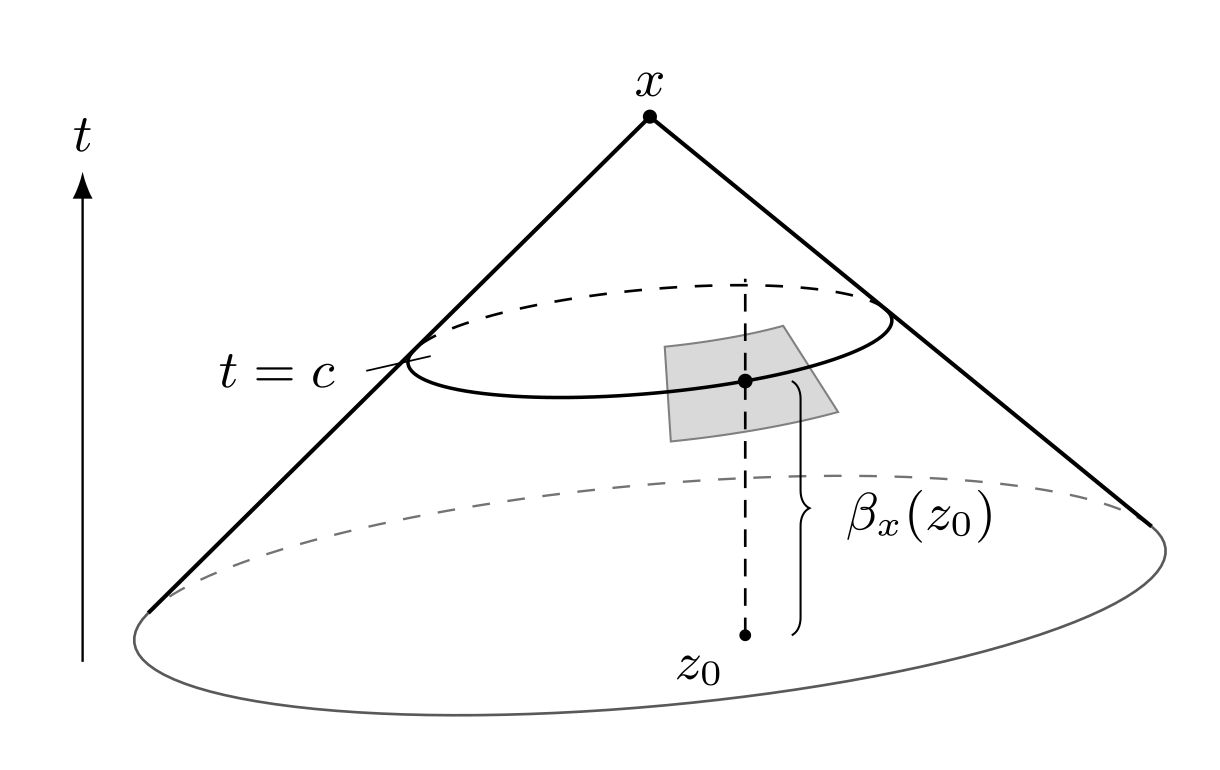}
\caption{A past null cone in flat $2+1$-dimensional coordinates. The vertical line over one spatial point $z_0$ meets the cone at time $\beta_x(z_0)$. The circular section $t=c$ instead consists of all spatial points $z$ with $\beta_x(z)=c$. Shading marks a cone patch within a neighborhood supporting a local smooth weight. In curved coordinates these sections need not be circles.}
\label{fig:cone-arrival-graph}
\end{figure}

Figure~\ref{fig:volume-superlevel} illustrates the compact sets $K_\theta\subset\mathrm{int}D$: requiring both past and future volumes to be at least $\theta$ keeps a point away from $\partial D$.

\begin{figure}[htbp]
\centering
\includegraphics[width=.78\linewidth]{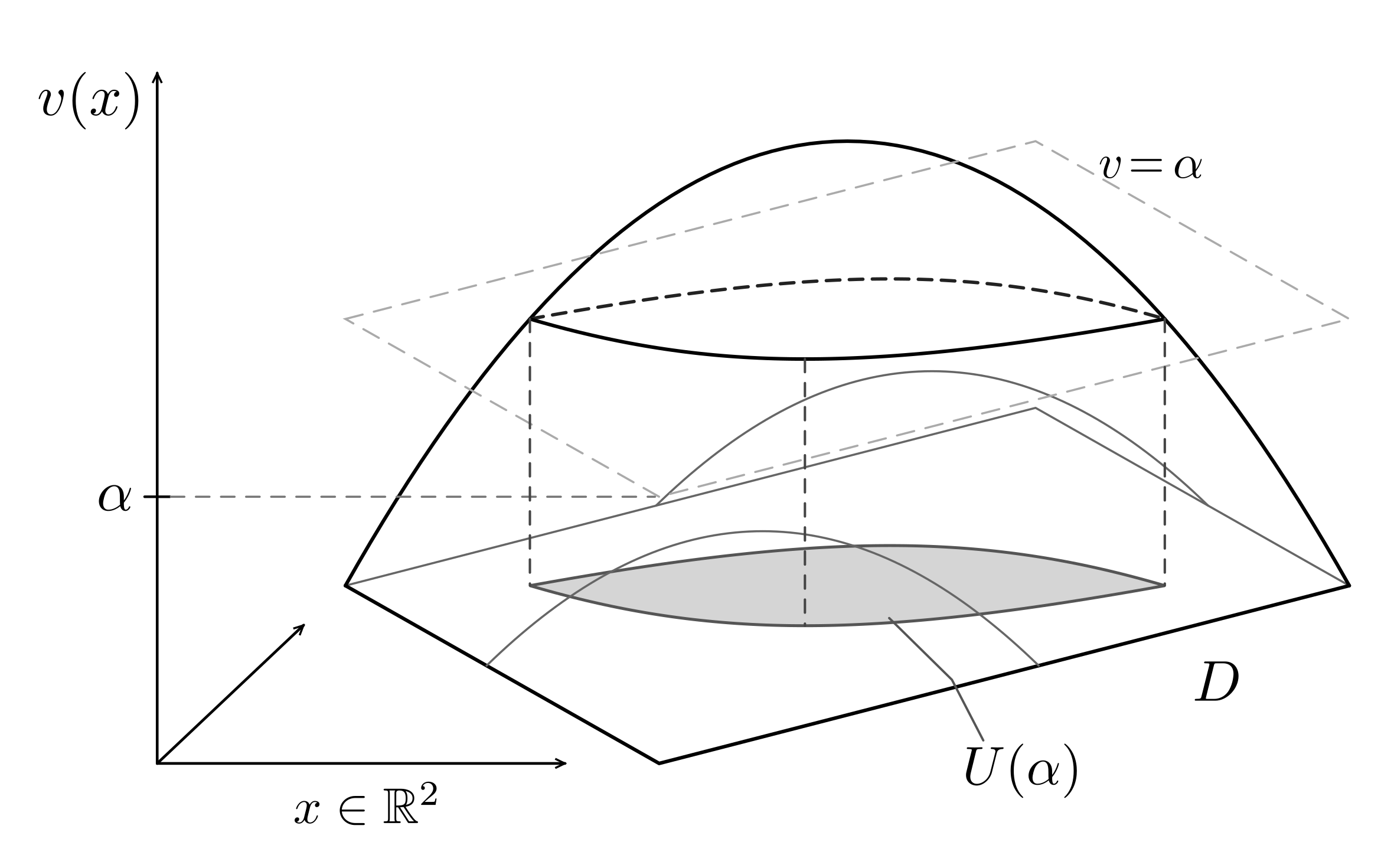}
\caption{A schematic graph of $v$ over $D$. Cutting it at height $\alpha>0$ and projecting the portion above this level gives $U(\alpha)=\{p\in D:v(p)\geq\alpha\}=K_\alpha$. The vertical coordinate is causal volume, not spacetime time. The drawing illustrates the superlevel-set construction; it does not assert that $v$ is smooth or that its level sets have this shape.}
\label{fig:volume-superlevel}
\end{figure}

Let $d_\nu$ be the count-based profile distance obtained by replacing $\mu$ with $\nu$ in the definition of $d_P$. For a sample point $p\in P$, define its profile ball by
\[
B_\nu(p,R)=\{q\in P:d_\nu(p,q)\leq R\}.
\]
The distance $d_\nu(p,q)$ adds the weighted counts of samples whose future membership differs between $p$ and $q$ and samples whose past membership differs. Past and future disagreements are counted separately. Thus $B_\nu(p,R)$ consists of the sample points whose total disagreement with $p$ is at most $R$. Its definition uses the causal order and $\rho$. Corollary~\ref{cor:profile-localization} shows that, for the stated points and scales, it contains the samples in a smaller $h$-ball and is contained in a larger one.

\begin{corollary}[Localization from causal counts]\label{cor:profile-localization}
	Assume V1 and V2 and set 
	{\interlinepenalty=10000 \postdisplaypenalty=10000
	\[ s_\rho = (L_\rho /\rho)^{1/(d+2)}, \qquad \eta_\rho = C(s_\rho + \sqrt{L_\rho / \rho}). \]
	The individual past/future counts and empirical symmetric-difference profile satisfy, uniformly over all continuum endpoints,
	\[ |\nu(J^\pm(p))-V^\pm(p)| \leq \eta_\rho, \qquad |d_\nu(p,q)-d_P(p,q)| \leq \eta_\rho, \]
	}
	after changing $C$. Fix $\theta>0$. If a sample point has both empirical past and future volumes at least $2\theta$ and $\eta_\rho < \theta$, it belongs to $K_\theta$. For such a point $p$ and every $q\in P$,
	\begin{equation}\label{eq:empirical-profile-localization}
c d_h(p,q)-\eta_\rho \leq d_\nu(p,q) \leq C d_h(p,q) + \eta_\rho .
\end{equation}
	Fix $A>0$, independently of $\rho$ and $\ell$. If $\eta_\rho=o(\ell)$, then, for sufficiently large $\rho$,
	\[
P\cap B_h(p,c_A\ell)\subseteq B_\nu(p,A\ell)\subseteq P\cap B_h(p,C_A\ell).
\]
	The positive constants $c_A,C_A$ depend only on $A$, $\theta$, and the fixed geometry. These estimates hold simultaneously for every sample pair whose first endpoint satisfies the stated past- and future-count bounds.
\end{corollary}

\begin{proof}
	Apply Lemma~\ref{lem:global-set-discrepancy} with $t=s_\rho$; $\delta_t = K s_\rho$. The relevant constant multiples of this scale exceed $s_{\min}$ and are below $s_0$ eventually. The two individual count estimates place the selected point in $K_\theta$, and \eqref{eq:empirical-profile-localization} follows from Lemma~\ref{lem:causal-profile-distance}. For $R>\eta_\rho$, this comparison gives
	\[
P\cap B_h\!\left(p,\frac{R-\eta_\rho}{C}\right)\subseteq B_\nu(p,R)\subseteq P\cap B_h\!\left(p,\frac{R+\eta_\rho}{c}\right).
\]
	Take $R=A\ell$. Since $\eta_\rho=o(\ell)$, eventually $\eta_\rho\leq A\ell/2$, so $c_A=A/(2C)$ and $C_A=2A/c$ suffice.
\end{proof}

Figure~\ref{fig:profile-ball-localization} illustrates these inclusions.

\begin{figure}[htbp]
\centering
\includegraphics[width=.62\linewidth]{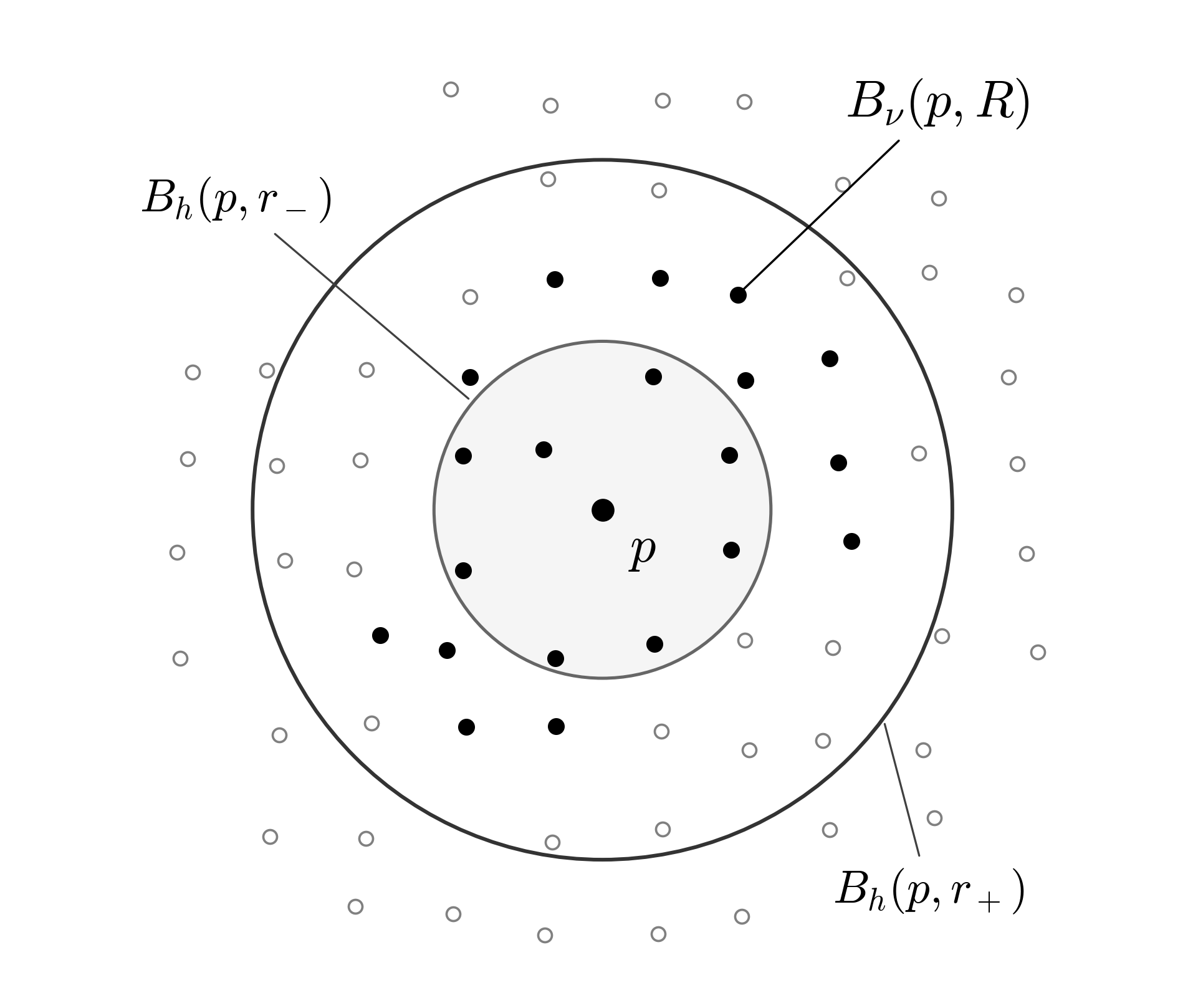}
\caption{The count-defined profile ball $B_\nu(p,R)$ under the hypotheses of Corollary~\ref{cor:profile-localization}. Filled points belong to it and open points are other samples. For $R>\eta_\rho$, put $r_-=(R-\eta_\rho)/C$ and $r_+=(R+\eta_\rho)/c$. Every sample in the smaller $h$-ball belongs to $B_\nu(p,R)$, and every member of $B_\nu(p,R)$ lies in the larger $h$-ball. Membership in the intervening annulus can vary. The circles depict the observer balls in a flat schematic; the profile ball itself is the discrete set of filled points.}
\label{fig:profile-ball-localization}
\end{figure}

For an exact order embedding, $d_\nu$ at sample endpoints and the individual past/future counts are computed from the finite order and $\rho$. 

\subsection{Whole-interval localization and proper-time estimates}

In this subsection $I(p,q) = J^+(p) \cap J^-(q)$ denotes the closed causal interval. For comparable sample points $p,q\in P$, define
\[
N(p,q)=\#\bigl(P\cap I(p,q)\bigr)=\rho\nu(I(p,q)).
\]
This count includes both endpoints. Since $f$ is an exact order embedding, it is also the number of causal-set elements between the corresponding labels, including those labels.

\begin{lemma}[Containment of small causal intervals]
\label{lem:whole-interval-localization}
	For a fixed compact set $K\subset\mathrm{int}D$ and $C_0 < \infty$, there are $C_1, \ell_0 > 0$ such that if $p\in K$, $p\leq q$, and $d_h(p,q) \leq C_0 \ell$ with $0<\ell < \ell_0$, then the entire global interval $J^+(p) \cap J^-(q)$ lies in a controlled $g$-normal chart and in an observer ball of radius $C_1 \ell$ about $p$. A fixed multiple of this ball remains inside $D$. 
\end{lemma}
\begin{proof}
	Use the finite cover $U_k\Subset V_k\Subset W_k$ from Lemma~\ref{lem:small-diamond-stability}, with the sets $U_k$ covering $K$. For sufficiently small $\ell$, if $p\in U_k$, then $q\in V_k$. Since $V_k$ is causally convex, $I(p,q)\subset V_k\subset W_k$.

	In temporal coordinates on $W_k$, every future causal tangent vector satisfies $|v|_{\mathrm{Eucl}}\leq C\,dt(v)$. Integrating along causal curves from $p$ to $q$, whose coordinate time separation is $O(\ell)$, shows that every point of $I(p,q)$ lies in $B_h(p,C_1\ell)$.

	Normal-coordinate neighborhoods centered at points of $K$ have a uniformly available radius. Thus, for any fixed $L>C_1$, the ball $B_h(p,L\ell)$ lies in such a neighborhood when $\ell$ is sufficiently small. It also lies inside $D$, because $K$ has positive $h$-distance from $\partial D$.
\end{proof}

\begin{lemma}[Cone and proper-time comparison]
\label{lem:normal-cone-comparison}
	Let $K\subset\mathrm{int}D$ be compact, and let $z$ be normal coordinates centered at $o\in K$, defined by a Lorentz-orthonormal frame adapted to $T(o)$. Fix $a>0$. For sufficiently small $\ell$, any $p\leq q$ with $|z_p|,|z_q|\leq a\ell$ admits constant wider and narrower comparison metrics $g_w,g_n=\eta+O(\ell^2)$ whose flat endpoint proper times satisfy
	\begin{equation}\label{eq:absolute-cone-comparison}
0\leq \tau_n \leq \tau_g(p,q) \leq \tau_w, \qquad \tau_w \leq \tau_g(p,q) + C\ell^2, \qquad \mu(I(p,q)) \leq C\tau_w^d.
\end{equation}
	Set $\tau_n=0$ if the endpoints are not causally related for $g_n$. In the same coordinates,
	\begin{equation}\label{eq:common-chart-squared-interval}
|\eta(z_q-z_p,z_q-z_p)+\tau_g(p,q)^2|\leq C\ell^4.
\end{equation}
	All estimates are uniform for $o\in K$ and such pairs, including pairs close to the null cone.
\end{lemma}
\begin{proof}
	Choose a compact $K'\subset\mathrm{int}D$ containing a neighborhood of $K$. Uniform equivalence of coordinate distance and $d_h$ puts $p,q$ in $K'$ and gives $d_h(p,q)\leq C\ell$ for small $\ell$. Lemma~\ref{lem:whole-interval-localization} then places their entire global interval inside the coordinate ball $|z|\leq L\ell$ about $o$, for some fixed $L$. Increase $L$ to include the flat comparison intervals as well: metrics sufficiently close to $\eta$ have uniformly bounded causal speeds, and the endpoint coordinate time difference is $O(\ell)$. For small $\ell$, this larger ball remains inside both the normal coordinate domain and $D$.

	On this ball, $g=\eta+O(\ell^2)$ and its volume density is $1+O(\ell^2)$, uniformly over $K$. Take $g_w=\eta-\epsilon I$ and $g_n=\eta+\epsilon I$, where $I$ is the Euclidean identity matrix and $\epsilon=C\ell^2$ bounds the operator norm of the metric error there. Write $I_w,I_n$ for their flat causal intervals, identified with subsets of the coordinate chart; take $I_n=\varnothing$ if its endpoints are not causally related. The quadratic-form inequalities $g_w\leq g\leq g_n$ give
	\[ I_n\subset I(p,q)\subset I_w, \qquad \tau_n\leq\tau_g(p,q)\leq\tau_w. \]
	Indeed, every actual causal curve is wider-causal and has no greater proper length for $g$ than for $g_w$. Conversely, the straight segments defining $I_n$ are $g$-causal, and the direct narrow-causal segment gives the lower proper-time bound. All these curves stay in the ball on which the metric comparison holds.

	Put $v=z_q-z_p$, $b=-\eta(v,v)$, and $u=\epsilon|v|^2=O(\ell^4)$. The flat squared times are $\tau_w^2=b+u$ and $\tau_n^2=\max\{b-u,0\}$, so
	\[ \max\{b-u,0\}\leq\tau_g(p,q)^2\leq b+u. \]
	This gives $|\tau_g(p,q)^2-b|\leq u$ and $0\leq\tau_w^2-\tau_g(p,q)^2\leq 2u$. Taking square roots proves the additive bound in \eqref{eq:absolute-cone-comparison} and the squared-interval estimate \eqref{eq:common-chart-squared-interval}, without dividing by a possibly small proper time. Finally, the coordinate volume of $I_w$ is $\kappa_d\tau_w^d/\sqrt{|\det g_w|}$, where $\kappa_d$ is the volume of the flat unit-height Alexandrov interval. The interval inclusion and the actual volume density give $\mu(I(p,q))\leq C\tau_w^d$.
\end{proof}

Curvature corrections to small causal-diamond volumes were studied by Myrheim~\cite{Myrheim1978} and Gibbons--Solodukhin~\cite{GibbonsSolodukhin2007}. The following consequence gives the relative volume estimate used in the count-based proper-time reconstruction.

\begin{corollary}[Volume of a timelike interval]
\label{cor:relative-interval-volume}
	Under the hypotheses of Lemma~\ref{lem:normal-cone-comparison}, if $\tau_g(p,q)\geq c_0\ell$ for a fixed $c_0>0$, then
	\begin{equation}\label{eq:relative-interval-volume}
\mu(I(p,q))=\kappa_d\tau_g(p,q)^d[1+O(\ell^2)],
\end{equation}
	uniformly over such pairs, where $\kappa_d$ is the volume of the flat unit-height Alexandrov interval.
\end{corollary}
\begin{proof}
	The squared-time bounds in the lemma give $\tau_n,\tau_w=\tau_g(p,q)[1+O(\ell^2)]$ when $\tau_g(p,q)\geq c_0\ell$; in particular, $\tau_n>0$ for small $\ell$. The coordinate volumes of $I_n$ and $I_w$ are $\kappa_d\tau_n^d/\sqrt{|\det g_n|}$ and $\kappa_d\tau_w^d/\sqrt{|\det g_w|}$, respectively. Both determinant magnitudes and the actual volume density are $1+O(\ell^2)$. The inclusions $I_n\subset I(p,q)\subset I_w$ therefore give \eqref{eq:relative-interval-volume}.
\end{proof}

\begin{theorem}[Proper time from interval counts]
\label{thm:local-count-certificate}
	Assume V1. Fix a compact set $K\subset\mathrm{int}D$ and constants $C_0,n_0>0$. There exist $\ell_0,c_*,\epsilon_*>0$ such that the following holds. Let $p\in P\cap K$, $q\in P$, and $0<\ell\leq\ell_0$ satisfy $p\leq q$, $d_h(p,q)\leq C_0\ell$, and
	\begin{equation}\label{eq:count-certificate}
N(p,q)\geq n_0\rho\ell^d.
\end{equation}
	Let the test height $t$ satisfy
	\[
s_{\min}\leq t\leq c_*\ell,\qquad \delta_t+t/\ell\leq\epsilon_*.
\]
	Then uniformly over all such pairs
	\begin{equation}\label{eq:certified-interval-size}
c \ell \leq \tau_g(p,q) \leq C\ell, \qquad \mathrm{diam}_h I(p,q) \leq C\ell,
\end{equation}
	and the data estimator
	\[ \hat{\tau}(p,q) = \left( \frac{N(p,q)}{\rho \kappa_d} \right)^{1/d} \]
	satisfies
	\begin{equation}\label{eq:certified-time-error}
\left| \frac{\hat \tau(p,q)}{\tau_g(p,q)} -1 \right| \leq C\left( \ell^2 + \delta_t + \frac{t}{\ell} \right).
\end{equation}
	In any common-base normal chart satisfying Lemma~\ref{lem:normal-cone-comparison}, its negative square differs from $\eta(z_q-z_p,z_q-z_p)$ by at most $C\ell^2(\ell^2+\delta_t + t/\ell)$. 
\end{theorem}
\begin{proof}
	Lemma~\ref{lem:whole-interval-localization} supplies containment and the diameter bound before any proper time estimate. Fix the covering constants in Lemma~\ref{lem:local-discrepancy} so that all test diamonds have heights between $t$ and $B\ell$ for a fixed $B\geq1$. Choose $c_*$ small enough for that lemma and $\ell_0$ small enough that $B\ell_0\leq s_0$ and all these diamonds remain inside $D$. Their heights then lie in $[s_{\min},s_0]$, so V1 applies.

	The two causal boundaries satisfy Lemma~\ref{lem:achronal-tube-bounds}, so the local discrepancy estimate \eqref{eq:local-set-discrepancy} applies even if the interval is thin. Choose $\epsilon_*$ small enough that the discrepancy is at most $(n_0/2)\ell^d$ and $\delta_t\leq1/2$. The lower count in \eqref{eq:count-certificate} gives $\mu(I(p,q))\geq(n_0/2)\ell^d$. The absolute volume upper bound in \eqref{eq:absolute-cone-comparison} therefore forces $\tau_w\geq c\ell$. Its additive comparison gives $\tau_g\geq c'\ell$ for small $\ell$. The upper proper time bound follows from the wider cone and endpoint localization, without an upper-count assumption. This proves \eqref{eq:certified-interval-size}.
	
	Now apply \eqref{eq:relative-interval-volume}. Local count discrepancy and the volume expansion show $N(p,q)/(\rho\kappa_d)=\tau_g^d[1+O(\ell^2+\delta_t+t/\ell)]$. Taking the $d$th root gives \eqref{eq:certified-time-error}. The final assertion follows by squaring this estimate and using \eqref{eq:common-chart-squared-interval}. 
\end{proof}

\subsection{Choice of reconstruction scales}\label{sec:reconstruction-scales}

To estimate interval volumes within a region of radius $\ell$, we use the counting bounds on test diamonds of height $t$. Making $t$ smaller reduces the boundary error $t/\ell$, but increases the relative counting tolerance $\delta_t$. We first balance these errors, then choose reconstruction scales for which the counting error is no larger than the geometric error $O(\ell^2)$.

Balancing $\sqrt{L_\rho/(\rho t^d)}$ and $t/\ell$ gives the test height
\[ t_*=(L_\rho\ell^2/\rho)^{1/(d+2)}, \qquad \delta_{t_*}+t_*/\ell\leq C\left(\frac{L_\rho}{\rho\ell^d}\right)^{1/(d+2)}. \]
This height lies well above $s_{\min}$ when $\rho\ell^d\gg L_\rho^{(d+4)/2}$, and well below $\ell$ when $\rho\ell^d/L_\rho\to\infty$. If this choice fails the height or smallness requirements of Theorem~\ref{thm:local-count-certificate}, one must choose a height satisfying those requirements, and the displayed rate need not hold. Under these conditions, the relative proper-time error in Theorem~\ref{thm:local-count-certificate} is bounded by a constant times
\[ E_\ell=\ell^2+\left(\frac{L_\rho}{\rho\ell^d}\right)^{1/(d+2)}. \]

We use a fine scale $r$ for local coordinates and short-edge transport, and a larger scale $H$ for curvature loops. Equating the counting term in $E_r$ with $r^2$ determines $r$; the relation $H=\sqrt r$ will be justified by the curvature error budget. Thus put
\[ r=(L_\rho/\rho)^{1/(3d+4)}, \qquad H=\sqrt r. \]
Then $t_*(r)=r^3$ and
\[ \sqrt{\frac{L_\rho}{\rho r^{3d}}}=r^2, \qquad \frac{1}{\rho r^d}=\frac{r^{2d+4}}{L_\rho}=o(r^2). \]
Moreover
\[ s_\rho=r^{3-2/(d+2)}, \quad s_{\min}=r^{3+4/d}L_\rho^{1/d}, \quad s_{\min}\ll r^3\ll s_\rho\ll r^2\ll r. \]
Thus the net radius is $o(r^3)$ and the profile error is $o(r^2)$. At both reconstruction scales $r$ and $H$, the corresponding test heights satisfy the height and smallness requirements above, and
\[ E_r=O(r^2), \qquad E_H=O(H^2). \]
Indeed the counting term at $H=\sqrt r$ is $r^{(5d+8)/(2d+4)}=o(r)$.

The length scales used in the construction are summarized below. The variables $t$ and $\ell$ enter the local counting estimates, while $r$ and $H$ are the chosen reconstruction scales.
\begin{center}
\small
\renewcommand{\arraystretch}{1.25}
\begin{tabular}{@{}lp{.28\linewidth}p{.53\linewidth}@{}}
\textbf{Symbol} & \textbf{Value or choice} & \textbf{Role} \\
\hline
$s_{\min}$ & $\rho^{-1/d}L_\rho^{2/d}$ & \raggedright Minimum test-diamond height in (V1). \tabularnewline
$s_\rho$ & $(L_\rho/\rho)^{1/(d+2)}$ & \raggedright Averaging height used to estimate past, future, and profile counts. \tabularnewline
$t$ & Variable & \raggedright Test-diamond height in the local counting estimate. \tabularnewline
$\ell$ & Variable & \raggedright Geometric scale of the region containing the interval being estimated. \tabularnewline
$r$ & $(L_\rho/\rho)^{1/(3d+4)}$ & \raggedright Fine-chart scale and scale of the short-edge subdivision. \tabularnewline
$H$ & $\sqrt r$ & \raggedright Side length of the target coordinate rectangles used for curvature reconstruction. \tabularnewline
\hline
\end{tabular}
\end{center}
When its height requirements hold, the balanced choice is $t=t_*(\ell)$, and in particular $t_*(r)=r^3$.

The preceding results now provide the ingredients for reconstruction: counts select sample points in a compact subset of $\mathrm{int}D$, profile distances locate nearby points, and the sample forms a sufficiently fine net there. Local count estimates and estimates comparing sample sums with volume integrals control the regions used in reconstruction, while the interval-count test gives proper-time and squared-interval estimates for the pairs that pass it.

\section{Normal coordinates from causal order and interval counts}
\label{sec:sampled-charts}

By sampled normal coordinates we mean coordinate maps defined on sample labels that approximate continuum normal coordinates at the same base, up to one bounded Lorentz transformation for each chart. Proposition~\ref{prop:accepted-chart-soundness} gives the precise error bound. The inner active set is fixed before choosing anchors, and every label in that set is retained. This supplies overlaps containing sufficiently many points spread across independent coordinate directions. Lemma~\ref{lem:empirical-conditioning} uses this spread to bound the covariance inverse and keep the overlap fits stable.

\subsection{Eligible chart and loop bases}\label{sec:eligible-bases}
Choose a fixed $\theta>0$ such that $\{x\in D:v(x)>6\theta\}$ is nonempty. For $k=2,4$, define the computable sets
\[
C_\rho[k\theta]=\left\{p\in C_\rho:\rho^{-1}\min\bigl(|J_C^-(p)|,|J_C^+(p)|\bigr)\geq k\theta\right\}.
\]
The set $C_\rho[k\theta]$ consists of labels whose past and future within the sampled diamond each have estimated volume at least $k\theta$. These sets are computable from the finite order and the specified density $\rho$ by counting the labels in each past and future and dividing by $\rho$. Under the volume conditions, for sufficiently large $\rho$, these thresholds keep the selected bases in a fixed interior region of $D$.

We construct charts at bases in $C_\rho[2\theta]$ and start curvature loops at bases in $C_\rho[4\theta]$. The larger threshold leaves room to move from a loop base to nearby fine-chart bases while preserving the smaller past- and future-count requirement. By the uniform individual count estimates in Corollary~\ref{cor:profile-localization}, every chart base lies in $K_\theta$ for sufficiently large $\rho$. The corresponding neighborhood statement for loop bases is established in Section~\ref{sec:raw-chart-interfaces}.

For sufficiently large $\rho$, the loop-base set is nonempty. Indeed, a small fixed ball contained in $\{v>6\theta\}$ contains a sample point by Corollary~\ref{cor:sample-net}, and the individual count estimates of Corollary~\ref{cor:profile-localization} place that label in $C_\rho[4\theta]$.

\subsection{Nested profile balls and the finite tests}
Compute from the poset
\[ d_C(p,c) = \rho^{-1}(|J_C^+(p)\triangle J_C^+(c)| + |J_C^-(p)\triangle J_C^-(c)|). \]
The past and future sets in this profile are taken in the entire sampled diamond $D$, so neighborhood selection uses causal relations throughout $D$. Under the counting and scale assumptions below, the profile comparison places the selected labels in an $h$-ball of radius $O(s)$. The anchor and active points lie in a fixed larger ball of the same order, and the interval counts used for coordinates are localized there by Theorem~\ref{thm:local-count-certificate}.
Under the deterministic volume conditions, \eqref{eq:empirical-profile-localization} gives, for every base $p\in C_\rho[2\theta]$ and label $c\in C_\rho$, 
\begin{equation}\label{eq:sample-profile-localization}
c_P d_h(f(p),f(c))-C\eta_\rho \leq d_C(p,c) \leq C_P d_h(f(p),f(c))+C\eta_\rho .
\end{equation}
The lower bound controls the $h$-distance from $f(p)$ to every label selected by a profile-distance threshold. Choose $a>0$, then a sufficiently large fixed $A_{\mathrm{cand}} > a$, and put
\[ \mathcal A_{p,s} =\{c : d_C(p,c) \leq as \}, \quad \mathcal{B}_{p,s} = \{c : d_C(p,c) \leq A_{\mathrm{cand}}s\}. \]
For $\eta_\rho = o(s)$ there are uniform $c_{in}, C_{act}, C_{out} > 0$ such that 
\[ \{ c: d_h(f(p),f(c)) \leq c_{in} s\} \subset \mathcal A_{p,s} \subset \{ c: d_h(f(p),f(c)) \leq C_{act} s\}, \]
and $\mathcal B_{p,s}$ lies in the physical $C_{out}s$ ball. For example $c_{in} = a/(2C_P)$ works for large density, after a harmless further decrease if necessary. 

The inner set $\mathcal A_{p,s}$ contains the sample points to which we will assign coordinates. The larger set $\mathcal B_{p,s}$ supplies candidates for reference points, called anchors. Figure~\ref{fig:anchor-construction} illustrates this arrangement; the construction itself uses only profile distances, causal order, and interval counts, not the displayed continuum coordinates.
\begin{center}
\begin{minipage}{\linewidth}
    \centering
    \captionsetup{hypcap=false}
    \includegraphics[width=.52\linewidth]{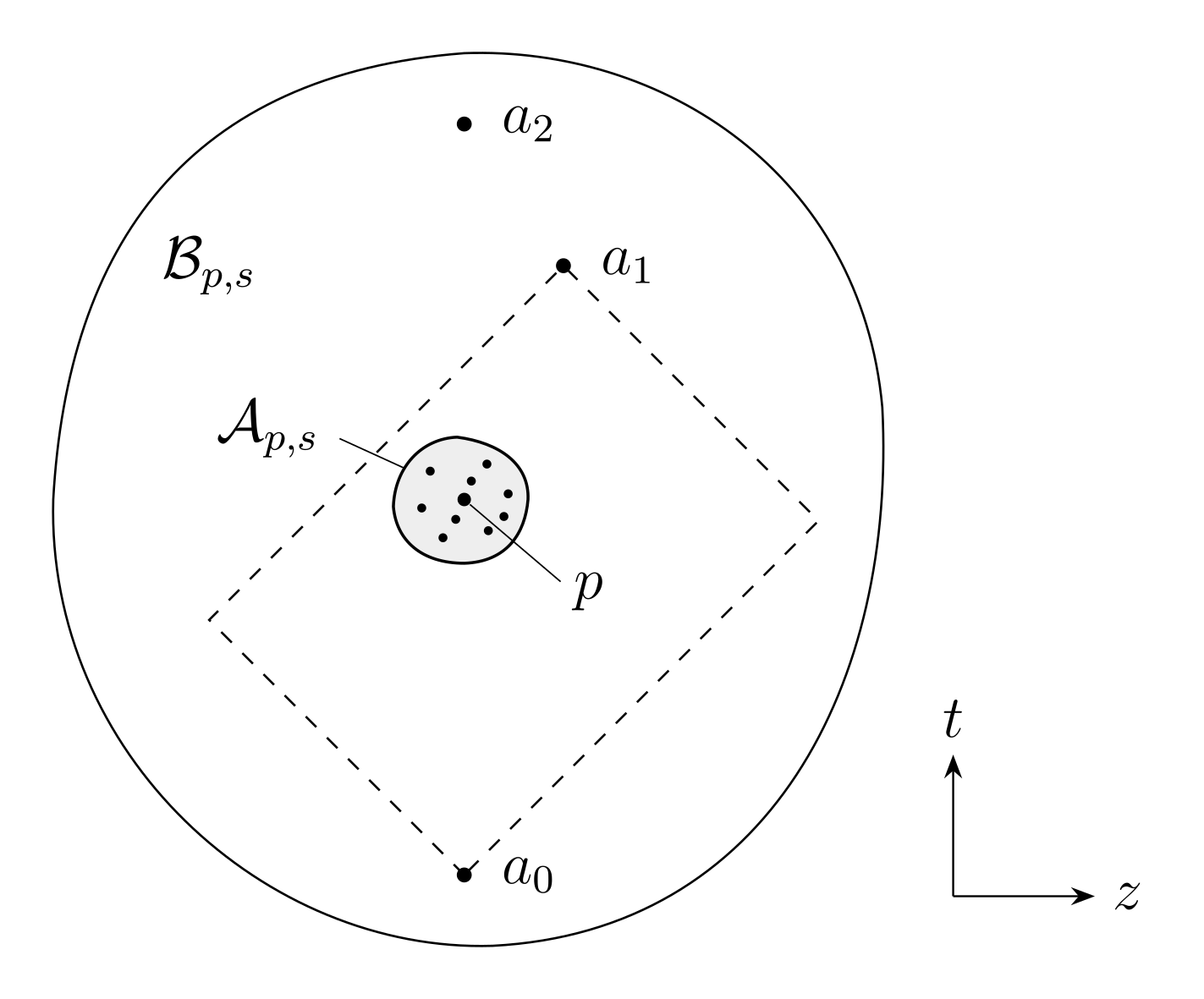}
    \captionof{figure}{An anchor tuple in a two-dimensional spacetime illustration. The sample points in $\mathcal A_{p,s}$ lie between $a_0$ and $a_1$, while $a_0\prec a_1\prec a_2$ are chosen from $\mathcal B_{p,s}\setminus\mathcal A_{p,s}$. Dashed lines show the null boundaries of $I(a_0,a_1)$. The profile-set boundaries are schematic; their shapes are not assumed to be Euclidean balls. This picture illustrates causal comparability, not the interval-count or Gram-matrix tests.}
    \label{fig:anchor-construction}
\end{minipage}
\end{center}

For a distinct comparable pair $c,c'$, let $N(c,c')$ count the labels in the closed order interval from the earlier label to the later one. Thus $N(c,c')=N(c',c)$.

Enumerate distinct order chains $\mathbf a = (a_0 \prec a_1 \prec \dots \prec a_d)$ in $\mathcal B_{p,s} \setminus \mathcal A_{p,s}$. This enumeration uses only the poset order. Theorem~\ref{thm:local-count-certificate} shows that pairs passing the lower count test are timelike in the realization. Require
\begin{itemize}
	\item[(C1)] Every anchor-anchor and anchor-active pair is comparable, including every pair with $c=p$.
	\item[(C2)] Every such distinct pair obeys $n_0 \rho s^d \leq N(c,c') \leq n_1 \rho s^d$, where $0 < n_0 < n_1 < \infty$ are fixed. Counts use the full order interval.
	\item[(C3)] With \[\hat D(c,c') = -\left( \frac{N(c,c')}{\rho \kappa_d} \right)^{2/d}, \quad \hat D(c,c)=0, \]
	and $\kappa_d$ the volume of the flat unit-height Alexandrov interval, the matrix
	\[ \hat G_{ij} = \frac{\hat D(a_0,a_i) + \hat D(a_0,a_j)-\hat D(a_i,a_j)}{2}, \qquad 1\leq i,j\leq d, \]
	has one negative and $d-1$ positive eigenvalues, and every eigenvalue of $\hat G/s^2$ has absolute value in $[g_0,g_1]$, for fixed $0<g_0 < g_1 < \infty$.
\end{itemize}

Condition (C1) makes every required interval available from the order. By Theorem~\ref{thm:local-count-certificate}, the lower count in (C2), together with localization, makes its proper time comparable to $s$. This excludes intervals that are too close to null at that scale. Condition (C3) gives the Gram matrix Lorentzian signature and a uniformly bounded inverse after rescaling by $s^2$, allowing stable coordinate reconstruction.

These tests are finite. Under the deterministic volume conditions, Proposition~\ref{prop:accepted-scaffolds} proves that a passing tuple exists when the constants are chosen with a fixed positive margin in the test inequalities. The accuracy bound below holds for every passing tuple, including tuples selected using the sample. 

The first diagonal element has the following lower bound, which keeps the factorization below away from a zero time pivot:
\begin{equation}\label{eq:time-pivot-bound}
-\hat G_{11} \geq (n_0/\kappa_d)^{2/d}s^2
\end{equation}

\subsection{Gram-matrix factorization and coordinate reconstruction}
The coordinate formula comes from expressing inner products in terms of squared intervals. Put $\eta=\operatorname{diag}(-1,1,\ldots,1)$. First consider flat Minkowski coordinates $z_c$, with $z_p=0$, and let $S$ have the independent anchor displacements $z_{a_i}-z_{a_0}$ as columns. Their Gram matrix is
\[
G=S^T\eta S,
\]
the pullback of the Minkowski metric $\eta$ under the linear map $S$. Thus $G$ expresses Minkowski inner products in terms of coefficients in the anchor basis. Polarization determines the inner products
\[
w_i(c)=\eta(z_{a_i}-z_{a_0},z_c-z_{a_0})
\]
from the squared intervals between the anchors and $c$. Subtracting the value at $p$ gives
\[
w(c)-w(p)=S^T\eta z_c,
\qquad
G^{-1}[w(c)-w(p)]=S^{-1}z_c.
\]
Thus the inverse Gram matrix recovers the coefficients of $z_c$ in the anchor basis.

To express these coefficients in Minkowski coordinates, we seek a matrix $B(G)$ satisfying
\[
B(G)^T\eta B(G)=G.
\]
For any such matrix,
\[
B(G)G^{-1}[w(c)-w(p)]=B(G)S^{-1}z_c.
\]
The matrix $\Lambda=B(G)S^{-1}$ satisfies $\Lambda^T\eta\Lambda=\eta$. The recovered coordinates therefore differ from $z_c$ by a Lorentz transformation. We now specify $B(G)$ by an explicit factorization.

For a symmetric matrix $G$ with one negative and $d-1$ positive eigenvalues and with $G_{11}<0$, define
\[ a_* = \sqrt{-G_{11}}, \quad t_j=-G_{1j}/a_* \quad (2\leq j \leq d), \quad S_{sp} = G_{2:d,2:d}-G_{2:d,1}G_{1,2:d}/G_{11} \]
To see why $S_{sp}$ is positive definite, write $G=\begin{pmatrix}\alpha&b^T\\ b&Q\end{pmatrix}$, where $\alpha=G_{11}<0$. Completing the square gives
\[ \begin{pmatrix}u\\v\end{pmatrix}^{T}G\begin{pmatrix}u\\v\end{pmatrix} = \alpha\left(u+\frac{b^Tv}{\alpha}\right)^2+v^T\left(Q-\frac{bb^T}{\alpha}\right)v. \]
An invertible linear change of variables preserves the numbers of positive and negative directions of a quadratic form. The first term accounts for the unique negative direction of $G$, so the remaining form, whose matrix is $S_{sp}=Q-bb^T/\alpha$, is positive definite. This is the Schur complement inertia identity in the case used here.

The positive-definite matrix $S_{sp}$ has a unique Cholesky factorization $R^TR=S_{sp}$ with $R$ upper triangular and its diagonal entries positive. Thus $v^TS_{sp}v=|Rv|^2$. The matrix $R$ converts the remaining positive quadratic form into Euclidean spatial coordinates. The triangular and positive-diagonal requirements specify a unique factor, which varies smoothly with $S_{sp}$. Combining it with the negative direction, set
\[ B(G) = \begin{pmatrix}a_* & t^T\\ 0 & R \end{pmatrix}, \quad B(G)^T \eta B(G) = G \]
The identity follows by block multiplication: the off-diagonal block is $-a_* t^T = G_{1,2:d}$ and the lower block is $-tt^T + S_{sp} = G_{2:d,2:d}$. Restrict $G/s^2$ to matrices with the eigenvalue bounds in (C3) and the corresponding bound on $-G_{11}/s^2$ from \eqref{eq:time-pivot-bound}. These normalized matrices form a compact set on which $B$ is smooth and both $B$ and its inverse are uniformly bounded. Consequently
\[ \|B(G)\| \leq Cs, \quad \|B(G)^{-1}\| \leq C/s, \quad \|B(G')-B(G)\| \leq (C/s)\|G'-G\| \]
for sufficiently close normalized matrices. The final assertion follows from bounded derivatives on a fixed slightly enlarged compact class; the segment remains there for sufficiently small perturbations.

We now replace the exact squared intervals by their estimates. For $c \in \mathcal A_{p,s}$, define
\[ \hat w_i(c) = \frac12 (\hat G_{ii}-\hat D(a_i,c) + \hat D(a_0,c)) \]
and
\begin{equation}\label{eq:sample-coordinate-formula}
X_{p,s,\mathbf a}(c) = B(\hat G)\hat G^{-1} [\hat w(c)-\hat w(p)]
\end{equation}
The subtraction centers the chart exactly at its base, so $X_{p,s,\mathbf a}(p)=0$.

The chart construction uses $s=r$ for fine charts and $s=C_HH$ for coarse charts. Here $C_H>0$ is a fixed enlargement constant chosen so that the coarse chart contains the loop and the neighborhoods needed along it in Section~\ref{sec:transport-curvature}. We use the superscript $H$ for these coarse coordinates.

For comparison, at each base $p$ let $e_p:\R^d\to T_{f(p)}M$ be a Lorentz-orthonormal frame whose time vector is $T_{f(p)}$, and write $\chi_p=e_p^{-1}\log_{f(p)}^g$ for its normal chart. This reference chart is used to prove accuracy.

\begin{proposition}[Uniform accuracy of reconstructed charts]
\label{prop:accepted-chart-soundness}
	Assume the deterministic volume conditions of Definition~\ref{def:deterministic-volume-event}, and let $s=r$ or $s=C_H H$. At every base $p\in C_\rho[2\theta]$, every passing tuple defines coordinates satisfying
	\[ X_{p,s,\mathbf a}(c) =\Lambda_{p,s,\mathbf a} \chi_p(f(c))+e(c), \qquad \sup_{c\in \mathcal A_{p,s}} |e(c)| \leq CE_s s, \]
	where $\Lambda^T \eta \Lambda = \eta$ and $\|\Lambda\| + \| \Lambda^{-1} \| \leq C$, uniformly over all choices.  
\end{proposition}
\begin{proof}
	Take normal coordinates at $f(p)$ from a Lorentz-orthonormal frame whose time vector is $T_{f(p)}$, and write $z_c = \chi_p(f(c))$. By \eqref{eq:sample-profile-localization} all anchors and active labels are $O(s)$ from $f(p)$. The whole intervals, not just their endpoints, are confined to a controlled $O(s)$ chart by Lemma~\ref{lem:whole-interval-localization}. The thin-interval count discrepancy and the lower count test force each used interval to have volume at least $cs^d$. The wider-cone argument then forces its proper time to be at least $c's$, before applying the relative small-diamond volume expansion. Thus Theorem~\ref{thm:local-count-certificate} gives, uniformly for every tested pair,
	\[ \hat D(c,c') = \eta(z_c -z_{c'}, z_c -z_{c'}) + O(E_s s^2). \]
	The $O(s^4)$ normal-coordinate squared-interval error is included in $E_s s^2$. This is the point at which long thin intervals are excluded, without having assumed their coordinate diameter from their proper time.
	
	Let $S$ have columns $z_{a_i}-z_{a_0}$, and let $G=S^T \eta S$. Then $\|\hat G - G \| \leq C E_s s^2$ and $\| S \| \leq Cs$. The certificate and smallness of $E_s$ imply $\|G^{-1}\| \leq C/s^2$, so $S$ is invertible and 
	\[ S^{-1} = G^{-1} S^T \eta, \qquad \| S^{-1} \| \leq C/s \]
	The exact Lorentz matrix
	\[ \Lambda = B(G)S^{-1} \]
	satisfies the asserted bounds and $\Lambda^T \eta \Lambda = \eta$. For exact flat squared differences, polarization gives
	\[ w_i(c) = \eta(z_{a_i}-z_{a_0}, z_c - z_{a_0}), \qquad w(c) - w(p) = S^T \eta z_c. \]
	The scalar error yields $\| \hat w(c) -w(c) \| \leq CE_s s^2$, uniformly in $c$. Moreover
	\[ \|\hat G^{-1} - G^{-1}\| \leq CE_s /s^2, \qquad \|B(\hat G)-B(G)\| \leq CE_s s. \]
	Insert these into \eqref{eq:sample-coordinate-formula}; each error is bounded by $CE_s s$. The exact term is $B(G)G^{-1} S^T \eta z_c = B(G) S^{-1} z_c = \Lambda z_c$.
\end{proof}

\subsection{Existence of accepted anchor tuples}

\begin{proposition}[Existence of an accepted anchor tuple]
\label{prop:accepted-scaffolds}
	Assume the deterministic volume conditions of Definition~\ref{def:deterministic-volume-event}. The constants in the certification tests can be chosen in advance, with a fixed positive margin in the test inequalities, so that for all sufficiently large $\rho$ every base $p\in C_\rho[2\theta]$ has a passing tuple at both prescribed scales.
\end{proposition}
\begin{proof}
	By \eqref{eq:sample-profile-localization} the active set lies inside a fixed coordinate ball $|z| \leq C_1 s$. Choose $\sigma=\gamma s$, with $\gamma$ sufficiently large compared with $C_1$ and the profile comparison constants. The constant $\gamma$ separates the ideal anchors from the active neighborhood. In the proof-only normal chart consider the ideal points
	\[ z_{a_0} = (-4d\sigma , 0), \qquad z_{a_i} = (4i\sigma , \sigma e_i) \quad (1\leq i < d), \qquad z_{a_d} = (4d\sigma, 0) \]
	They are chronologically ordered, separated from the active ball, and every anchor is timelike related to every point in that ball with a fixed timelike margin. The determinant magnitude of the column matrix of differences from $a_0$ is $8d\sigma^d$. All relevant squared intervals and Gram eigenvalues therefore satisfy the desired normalized upper and lower bounds with strict positive margins. Curved geometry changes these statements by $O(s^2)$ in normalized variables.
	
	The lower profile bound places every ideal anchor outside the active set if $\gamma$ was sufficiently large. After fixing $\gamma$, take $A_{\mathrm{cand}}$ sufficiently large that the upper profile bound places the anchors and their nearby sampled replacements in the candidate set $\mathcal B_{p,s}$. Both constants are fixed independently of density. The entire anchor configuration remains inside a fixed compact subset of $\mathrm{int}D$ for small $s$, since its radius is a fixed multiple of $s$.
	
	Apply Corollary~\ref{cor:sample-net} on a fixed compact subset of $\mathrm{int}D$ containing the ideal anchors. Its sample-net radius is $u_\rho=o(s)$ at the scales used here. Choose a sample within $u_\rho$ of each ideal anchor, as illustrated in Figure~\ref{fig:anchor-replacement}. Smallness of $u_\rho/s$ preserves chronological and geometric margins and distinctness. The local count discrepancy puts every actual anchor–active interval count in a fixed range $[n_0 \rho s^d, n_1 \rho s^d]$: choose $n_0$ below the uniform ideal minimum and $n_1$ above the maximum. The scalar approximation similarly puts the empirical Gram spectrum inside $[g_0, g_1]$ in absolute value, with the right inertia. All inequalities are uniform over the full active set, not just a selected finite subcollection. Hence this sampled tuple passes every test.
\end{proof}

\begin{center}
\begin{minipage}{\linewidth}
    \centering
    \captionsetup{hypcap=false}
    \includegraphics[width=.72\linewidth]{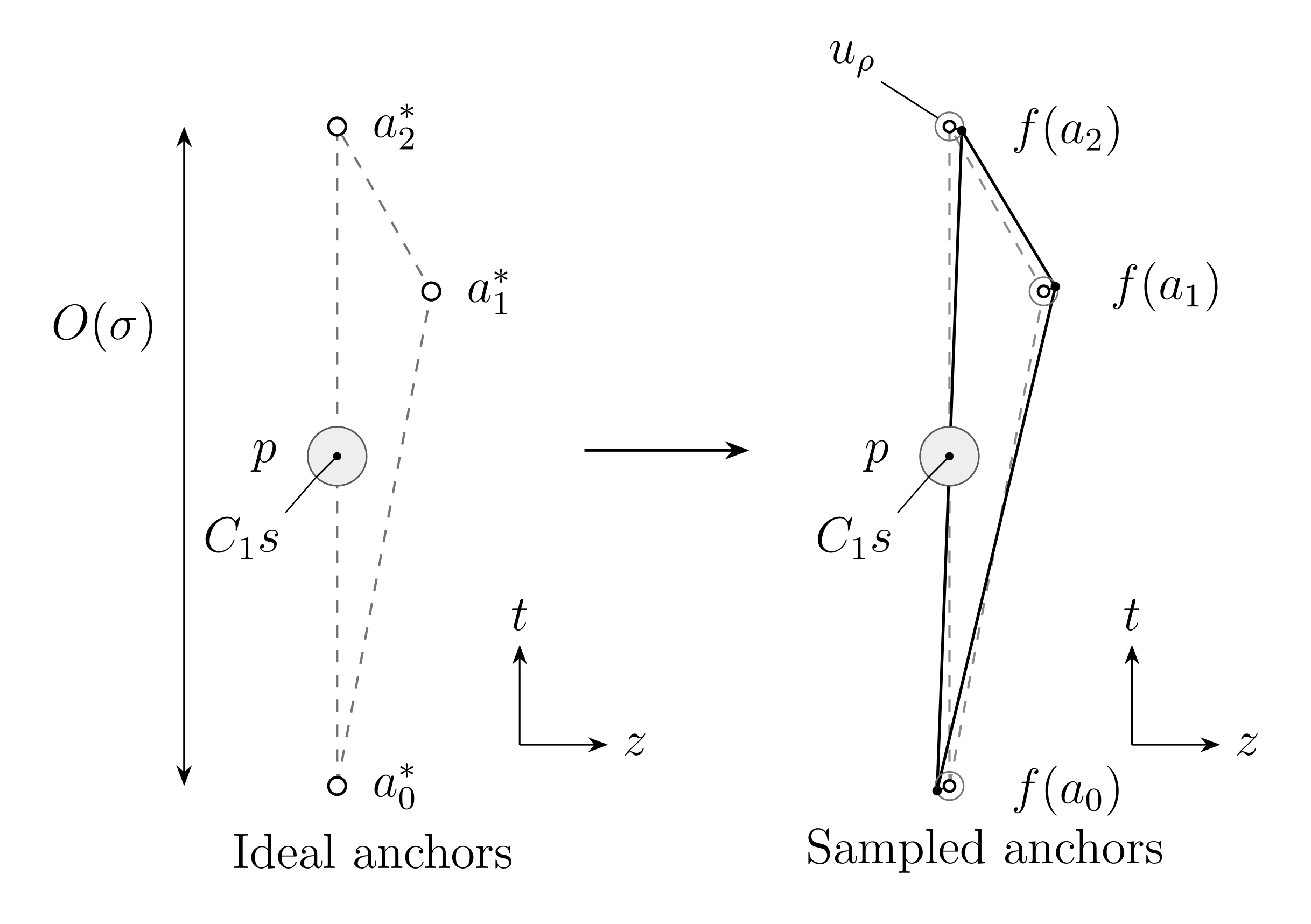}
    \captionof{figure}{Ideal anchors (hollow points) and their sampled replacements (filled points) in a schematic normal chart. The shaded coordinate ball has radius $C_1s$ and contains the active points. Anchor separations are of order $\sigma=\gamma s$, while each small circle has $h$-radius $u_\rho$ in the flat model. Dashed and solid segments join the two configurations, not null boundaries. Since $u_\rho=o(s)$, replacement preserves the strict causal separations and affine independence. The ideal locations are used only in the proof, not in the finite search.}
    \label{fig:anchor-replacement}
\end{minipage}
\end{center}

\subsection{Uniform conditioning on chart overlaps}

\begin{lemma}[Covariance bounds on chart overlaps]\label{lem:empirical-conditioning}
	Assume the deterministic volume conditions of Definition~\ref{def:deterministic-volume-event}, let $s=r$ or $s=C_HH$, and use coordinates from passing tuples. If bases $p,q\in C_\rho[2\theta]$ obey $d_C(p,q)\leq as/2$, their full common active set $\mathcal O = \mathcal A_{p,s} \cap \mathcal A_{q,s}$ is nonempty. Uniformly weighted coordinates of either chart on $\mathcal O$ have centered covariance
	\[ M_X = \frac{1}{|\mathcal O|} \sum_{c \in \mathcal O} (X(c)-\overline X)(X(c)-\overline X)^T \geq \kappa s^2 I \]
	for a uniform $\kappa >0$. Also $|X(c)|\leq Cs$ after harmless translation and $c\rho s^d \leq | \mathcal O | \leq C\rho s^d$. For the same-center fine/coarse fit, either coordinate map on $\mathcal A_{p,r}$ has covariance at least $\kappa r^2I$, provided $\mathcal A_{p,r}\subset\mathcal A_{p,C_HH}$. 
\end{lemma}
\begin{proof}
	Symmetric-difference distance is a pseudometric, so the profile ball of radius $as/4$ about $p$ is contained in both active sets. It contains a physical observer ball of radius $c_s s$, by \eqref{eq:sample-profile-localization}. Choose $d+1$ tiny disjoint fixed-axis diamonds of size $\epsilon s$ near the vertices of a Euclidean simplex in this ball. The simplex has all directional widths bounded below by a fixed multiple of $s$. First fix sufficiently small $\epsilon > 0$, independent of density. V1 then gives at least $c_3 \rho s^d$ labels in each diamond. A bounded number of test diamonds of size $C_3 s$ covers the entire common active set, giving at most $C_4 \rho s^d$ labels in total. Thus each of the $d+1$ clusters has a fixed positive fraction of the mass. Figure~\ref{fig:cluster-geometry} illustrates the cluster construction and its projected separation.
	
	The single bounded and boundedly invertible $\Lambda$ of Proposition~\ref{prop:accepted-chart-soundness} preserves the simplex's directional widths up to constants. Cluster diameter $O(\epsilon s)$ and coordinate error $O(E_s s) = o(s)$ do not remove those widths. For any Euclidean unit vector $\xi$, two of the clusters consequently have projected separation at least $c_4 s$. The exact identity
	\[ \frac{1}{N} \sum_i (v_i - \overline v)^2 = \frac{1}{2N^2} \sum_{i,j} (v_i-v_j)^2 \]
	applied to $v_i = \xi \cdot X(c_i)$ gives $\xi^T M_X \xi \geq \kappa s^2$. This proves the matrix inequality. The same-center statement uses the identical cluster construction inside the fine active ball and the coarse inclusion. 
\end{proof}

\begin{center}
\begin{minipage}{\linewidth}
    \centering
    \captionsetup{hypcap=false}
    \includegraphics[width=.80\linewidth]{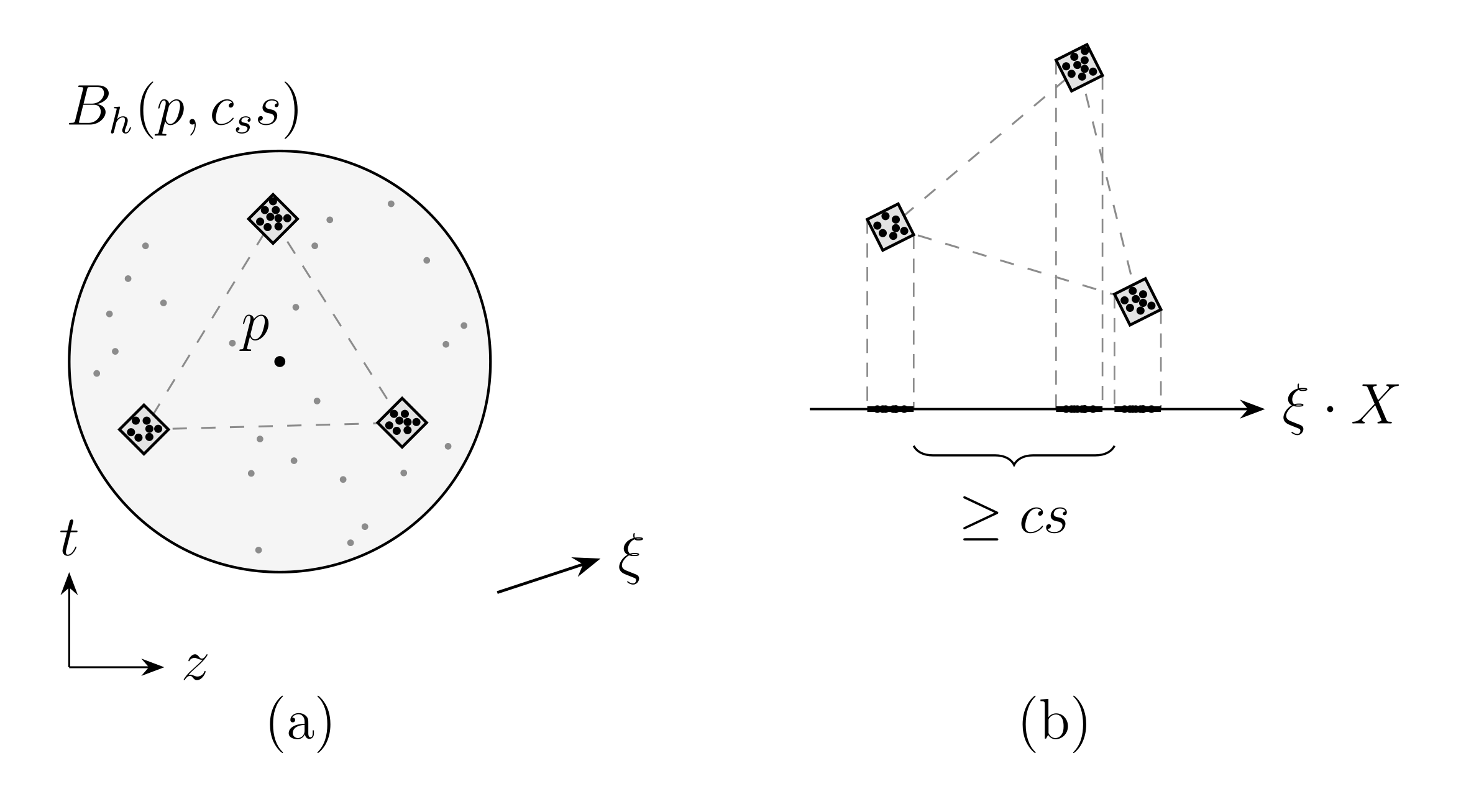}
    \captionof{figure}{Small diamonds around simplex vertices inside a ball contained in the overlap, shown in a flat two-dimensional model. Each cluster contains a fixed positive fraction of the full overlap sample, which may extend beyond the displayed ball. In (b), the clusters are projected onto a unit coordinate direction $\xi$. Two clusters have projected separation at least $cs$; the third may project between them. The separated pair may change with $\xi$, but some pair has this separation in every direction. Dashed triangle edges join centers, not null boundaries.}
    \label{fig:cluster-geometry}
\end{minipage}
\end{center}

For the same-center fit, $\mathcal A_{p,r} \subset \mathcal A_{p,C_HH}$ holds directly from the definitions when $r < C_H H$. For neighboring fine-chart bases, \eqref{eq:sample-profile-localization} shows that physical separation at most a sufficiently small fixed multiple of $r$ implies the profile overlap test, since $\eta_\rho =o(r)$. These domain and overlap estimates supply both overlap types needed below.

% ============================================================
% ============================================================
\section{From sampled normal coordinates to transport and curvature}
\label{sec:transport-curvature}
% ============================================================
The coordinates constructed in Section~\ref{sec:sampled-charts} allow us to compare tangent vectors at nearby sample points. The linear part of a transition between normal charts approximates parallel transport along the short geodesic joining their bases. An affine fit estimates this linear part from the coordinates of the shared sample labels. For the fit to be stable, those labels must be spread across independent coordinate directions, as expressed by a lower bound on their covariance.

We use the fine charts at scale $r$ to estimate transport along edges of length $O(r)$. The coordinate error $O(r^3)$ then contributes $O(r^2)$ to each fitted matrix. Combining these matrices around a loop gives a curvature estimate, but the accumulated fitting error must be small compared with the loop's area. A loop of side $H$ has $O(H/r)$ edges, so the accumulated error is $O(Hr)$, while curvature contributes at order $H^2$. The larger loop scale $H=\sqrt r$ makes their ratio tend to zero. The coarse chart locates the loop and expresses its curvature in one frame at the base.

We work throughout this section in the fixed geometric setting of Section~\ref{sec:geo}, under the uniform volume conditions~\eqref{eq:V1} and~\eqref{eq:V2}. We first establish the coverage, chart accuracy, and covariance bounds for the actual interval-count construction, then use them in the transport and loop estimates. All conclusions are uniform over accepted chart tuples and choices among equally close minimizing labels.

We use
\[
\eta=\mathrm{diag}(-1,1,\dots,1),\qquad R(U,V)=[\nabla_U,\nabla_V]-\nabla_{[U,V]}.
\]
Vectors are columns, and the earliest transport factor is on the right. These conventions give $I-H^2R_{ab}+O(H^3)$ for transport around a positively oriented small coordinate rectangle. We occasionally write a sample label $q$ for its realization $f(q)$ in geometric expressions.

\subsection{Geometric properties of the reconstructed charts}
\label{sec:reconstructed-chart-properties}
\label{sec:raw-chart-interfaces}
Use the scales $r=(L_\rho/\rho)^{1/(3d+4)}$ and $H=\sqrt r$ from Section~\ref{sec:reconstruction-scales}. A loop base is a label $p\in C_\rho[4\theta]$. At each chart base $q\in C_\rho[2\theta]$, Proposition~\ref{prop:accepted-scaffolds} gives passing tuples for the coordinate construction. Choose any passing tuple at each base and scale, and write
\[
I_q^r=\mathcal A_{q,r},\qquad I_p^H=\mathcal A_{p,C_HH}
\]
for the full fine and coarse active sets. The outer sets $\mathcal B_{q,s}$ contain the anchor candidates. Write $X_q^r$ and $X_p^H$ for the corresponding recentered coordinate maps. The fixed enlargement $C_H$ will be chosen below.

The compact sets
\[
K_0=\{x\in D:v(x)\geq7\theta/2\},\qquad
K_1=\{x\in D:v(x)\geq3\theta\},\qquad
K_2=\{x\in D:v(x)\geq\theta\}
\]
satisfy $K_0\subset\mathrm{int}K_1$ and $K_1\subset\mathrm{int}K_2$ by Lemma~\ref{lem:causal-profile-distance}. They are fixed independently of the sample and density. Smoothness of the fixed geometry on these compact neighborhoods supplies uniform bounds for normal-coordinate maps and their inverses, their derivatives used below, the connection and its first two coordinate derivatives, and short-geodesic endpoint maps. Choose the true normal charts on $h$-balls of a fixed sufficiently small radius about bases in $K_1$, with their closures in $\mathrm{int}K_2$. The finite maps $X_q^r$ and $X_p^H$ are evaluated only on their active label sets; the true charts also contain the smooth coordinate segments used in the proofs.

\begin{proposition}[Uniform neighborhoods and reconstructed-chart bounds]
\label{prop:count-chart-guarantees}
Fix $R_0>3$. The constants $C_H,C_*,R_1,K_{\mathrm{fr}},\delta_0,\kappa>0$ can be chosen independently of the density so that, for all sufficiently large $\rho$, the following statements hold at every loop base $p\in C_\rho[4\theta]$.

The point $f(p)$ lies in $K_0$. Every sample label in $B_h(f(p),C_*H)$ belongs to $C_\rho[2\theta]$ and therefore has a fine chart. This ball contains the coarse active domain and the fine-chart neighborhoods needed below. The shrinking active neighborhoods lie in $K_1$, while their true normal charts and the short geodesic constructions lie in $K_2$. Every continuum point in these fixed compact neighborhoods has a sample within $h$-distance
\[
u=O(\rho^{-1/d}L_\rho^{2/d}),\qquad
\frac{u}{r^3}=O(r^{4/d}L_\rho^{1/d})=o(1).
\]

Each reconstructed chart has one Lorentz-orthonormal comparison frame, constant across that chart, whose norm and inverse norm are at most $K_{\mathrm{fr}}$. At every fine-chart base $q$ in the above neighborhood,
\begin{equation}\label{eq:fine-chart-accuracy}
X_q^r(q)=0,\quad x_q^r(z)=(E_q^r)^{-1}\log_{f(q)}^g z,\quad |X_q^r(c)-x_q^r(f(c))|\leq Cr^3,
\end{equation}
for all $c\in I_q^r$, and
\begin{equation}\label{eq:fine-chart-domain}
\{c:f(c)\in B_h(f(q),c_0r)\}\subset I_q^r\subset\{c:f(c)\in B_h(f(q),C_0r)\}.
\end{equation}
The coarse chart satisfies
\begin{equation}\label{eq:coarse-chart-accuracy}
X_p^H(p)=0,\quad x_p^H(z)=(E_p^H)^{-1}\log_{f(p)}^g z,\quad |X_p^H(c)-x_p^H(f(c))|\leq CH^3,
\end{equation}
for all $c\in I_p^H$, and
\begin{equation}\label{eq:coarse-chart-domain}
\{c:|x_p^H(f(c))|\leq R_0H\}\subset I_p^H\subset\{c:|x_p^H(f(c))|\leq R_1H\}.
\end{equation}
These inclusions are taken within the fixed normal neighborhood. In particular, $I_p^r\subset I_p^H$.

If fine-chart bases $q,q'$ in the above neighborhood satisfy $d_h(f(q),f(q'))\leq\delta_0r$, their full common active set $J_{q,q'}=I_q^r\cap I_{q'}^r$ is nonempty and
\begin{equation}\label{eq:overlap-covariance}
\frac{1}{|J_{q,q'}|}\sum_{c\in J_{q,q'}}(X_q^r(c)-\overline X)(X_q^r(c)-\overline X)^T\succeq\kappa r^2I.
\end{equation}
The same lower bound holds for the covariance of $X_p^r$ on all of $I_p^r$, the fitting set for alignment with $X_p^H$. The covariance matrices are computable from the reconstructed coordinates.
\end{proposition}

\begin{proof}
The individual past- and future-count estimates in Corollary~\ref{cor:profile-localization} give $v(f(p))\geq4\theta-\eta_\rho\geq7\theta/2$ for large $\rho$. For every fixed $C_*$, the Lipschitz bound for $v$ places $B_h(f(p),C_*H)$ in $K_1$ when $H$ is small. The same count estimates then give empirical past and future volumes at least $2\theta$ at every sample label in this ball. This proves chart availability on any such fixed-multiple neighborhood; we choose $C_*$ after fixing the coarse domain below.

For a passing tuple $\mathbf a$ at scale $s=r$ or $s=C_HH$, Proposition~\ref{prop:accepted-chart-soundness} gives
\[
X_{q,s,\mathbf a}(c)=\Lambda_{q,s,\mathbf a}\chi_q(f(c))+O(E_ss),\qquad c\in\mathcal A_{q,s},
\]
where the single matrix $\Lambda=\Lambda_{q,s,\mathbf a}$ satisfies
\[
\Lambda^T\eta\Lambda=\eta,\qquad \|\Lambda\|+\|\Lambda^{-1}\|\leq C.
\]
The bound follows from the fixed normalized anchor and Gram-matrix thresholds and is uniform over accepted tuples at sufficiently small chart heights. With $\chi_q=e_q^{-1}\log_{f(q)}^g$, set $E_{q,s,\mathbf a}=e_q\Lambda_{q,s,\mathbf a}^{-1}$. These are the frames denoted by $E_q^r$ and $E_p^H$. Since the reference frames $e_q$ are adapted to $T$, the displayed bound gives a fixed $K_{\mathrm{fr}}$. Exact centering follows from the coordinate formula. The scales of Section~\ref{sec:reconstruction-scales} give $E_r=O(r^2)$ and $E_{C_HH}=O(H^2)$ for each fixed $C_H$, proving the two coordinate-error bounds.

The profile-ball inclusions of Section~\ref{sec:sampled-charts} give~\eqref{eq:fine-chart-domain}. The frame bounds compare the true coarse coordinate norm with $h$-distance uniformly. Choose $C_H$ large enough that the physical inner ball of radius $c_{in}C_HH$ contains the true coordinate ball of radius $R_0H$. The outer profile-ball bound then gives a finite $R_1$, proving~\eqref{eq:coarse-chart-domain}. These choices use the uniform frame bound above, which is independent of the fixed enlargement once the density is sufficiently large. Choose $C_*$ larger than the physical outer radius of this coarse domain divided by $H$, with room for its fine-chart neighborhoods. Large density puts all shrinking domains in $K_1$ and all the fixed-radius smooth comparison neighborhoods in $K_2$. Finally, $r<C_HH$ gives $I_p^r\subset I_p^H$ directly from their profile-distance definitions.

The profile upper bound gives $d_C(q,q')\leq C_Pd_h(f(q),f(q'))+C\eta_\rho$. Since $\eta_\rho=o(r)$, choose $\delta_0$ sufficiently small that $d_h(f(q),f(q'))\leq\delta_0r$ implies $d_C(q,q')\leq ar/2$ for large $\rho$. Lemma~\ref{lem:empirical-conditioning} then proves~\eqref{eq:overlap-covariance}. Its same-center conclusion gives the fine-source covariance for alignment on the entire set $I_p^r$. The coarse coordinate error also satisfies $H^3=o(r)$, so the same lemma applies to the coarse coordinates on this set. Every fit retains all labels of its stated fitting set.

Corollary~\ref{cor:sample-net}, applied on the fixed compact neighborhoods, gives the stated $u=O(s_{\min})$. The scale identity for $u/r^3$ above shows that $u=o(r^3)$. The profile error $\eta_\rho=o(r)$ controls neighborhood and overlap selection, while the sharper coordinate errors control transport and alignment.
\end{proof}

A short edge below is the unique local geodesic joining two sample points at observer distance at most $\delta_0r$. The preceding proposition makes their charts and stable fitting sets available. We now estimate transport on these edges and combine the estimates around the larger loops.

\subsection{Affine fits estimate short-edge parallel transport}
Let $N$ be the number of labels used in the fit and let bars denote averages over these labels. Fit an affine map $Y_c\approx BX_c+b$, allowing both the matrix $B$ and the translation vector $b$ to vary. For a fixed $B$, the least-squares translation is $b=\overline Y-B\overline X$, which gives the centered formula
\begin{equation}\label{eq:affine-fit}
\mathcal B(X,Y) = \left[ \frac{1}{N} \sum_c (Y_c - \overline Y)(X_c - \overline X)^T \right] \left[ \frac{1}{N} \sum_c (X_c - \overline X)(X_c -\overline X)^T \right]^{-1}.
\end{equation}

Figure~\ref{fig:overlap-affine-fit} illustrates this affine fit on the shared labels $J_{q,q'}$ of two neighboring fine charts.
\begin{center}
\begin{minipage}{\linewidth}
\centering
\captionsetup{hypcap=false}
\includegraphics[width=.94\linewidth]{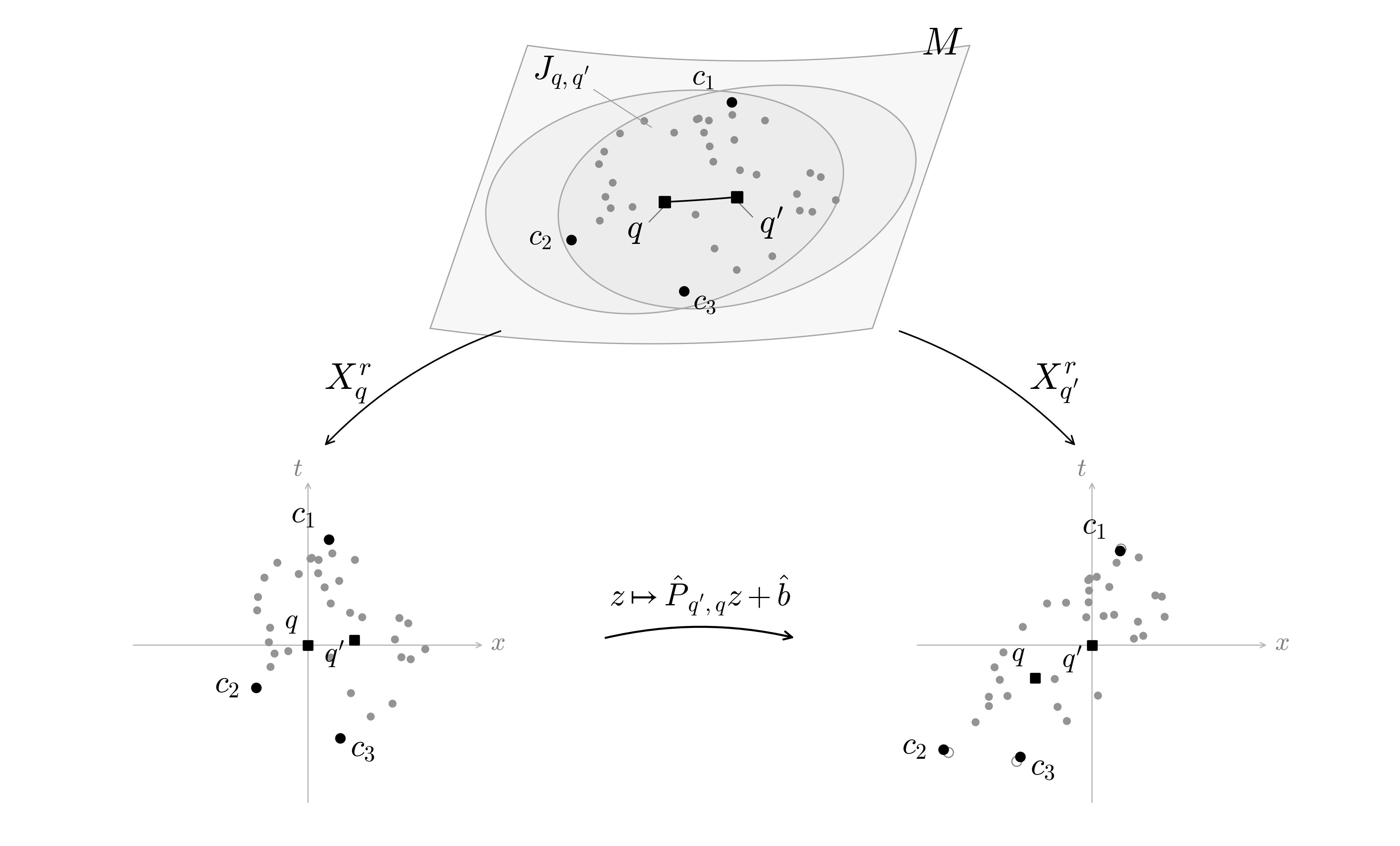}
\captionof{figure}{Top: a schematic neighborhood of $M$, with chart bases $q,q'$ and their shared sample labels $J_{q,q'}$. Squares identify the bases, and the segment joining them is the short geodesic whose transport is estimated. The neighborhood boundaries are schematic. The curved arrows assign coordinates to the same labels in the two fine charts and apply the fitted affine map $z\mapsto\hat P_{q',q}z+\hat b$, where $z=X_q^r(c)$. The matrix $\hat P_{q',q}$ is its linear part and $\hat b$ its translation; both are fitted using all shared labels. The three labels $c_1,c_2,c_3$ identify corresponding points in all three views. Open markers in the target view show their fitted locations. The linear part estimates transport in the reconstructed frames.}
\label{fig:overlap-affine-fit}
\end{minipage}
\end{center}

\begin{lemma}[Stability of affine least-squares fits]
\label{lem:affine-fit-stability}
	If $Y_c = AX_c + b +e_c$, $\sup_c |e_c| \leq \zeta r, |X_c -\overline X| \leq Cr$ and the source covariance is at least $\kappa r^2 I$, then $\| \mathcal B(X,Y)-A \| \leq C' \zeta$. 
\end{lemma}
\begin{proof}
	For source covariance $M_X$, the exact identity is
	\[ \mathcal B(X,Y)-A = \left[ \frac{1}{N} \sum_c (e_c-\overline e)(X_c-\overline X)^T \right] M_X^{-1}. \]
	The numerator is $O(\zeta r^2)$ and $\|M_X^{-1} \| \leq \kappa^{-1}r^{-2}$. 
\end{proof}

Figure~\ref{fig:parallel-transport} distinguishes geometric parallel transport from its matrix in the fine-chart frames.
\begin{center}
\begin{minipage}{\linewidth}
    \centering
    \captionsetup{hypcap=false}
    \includegraphics[width=.72\linewidth]{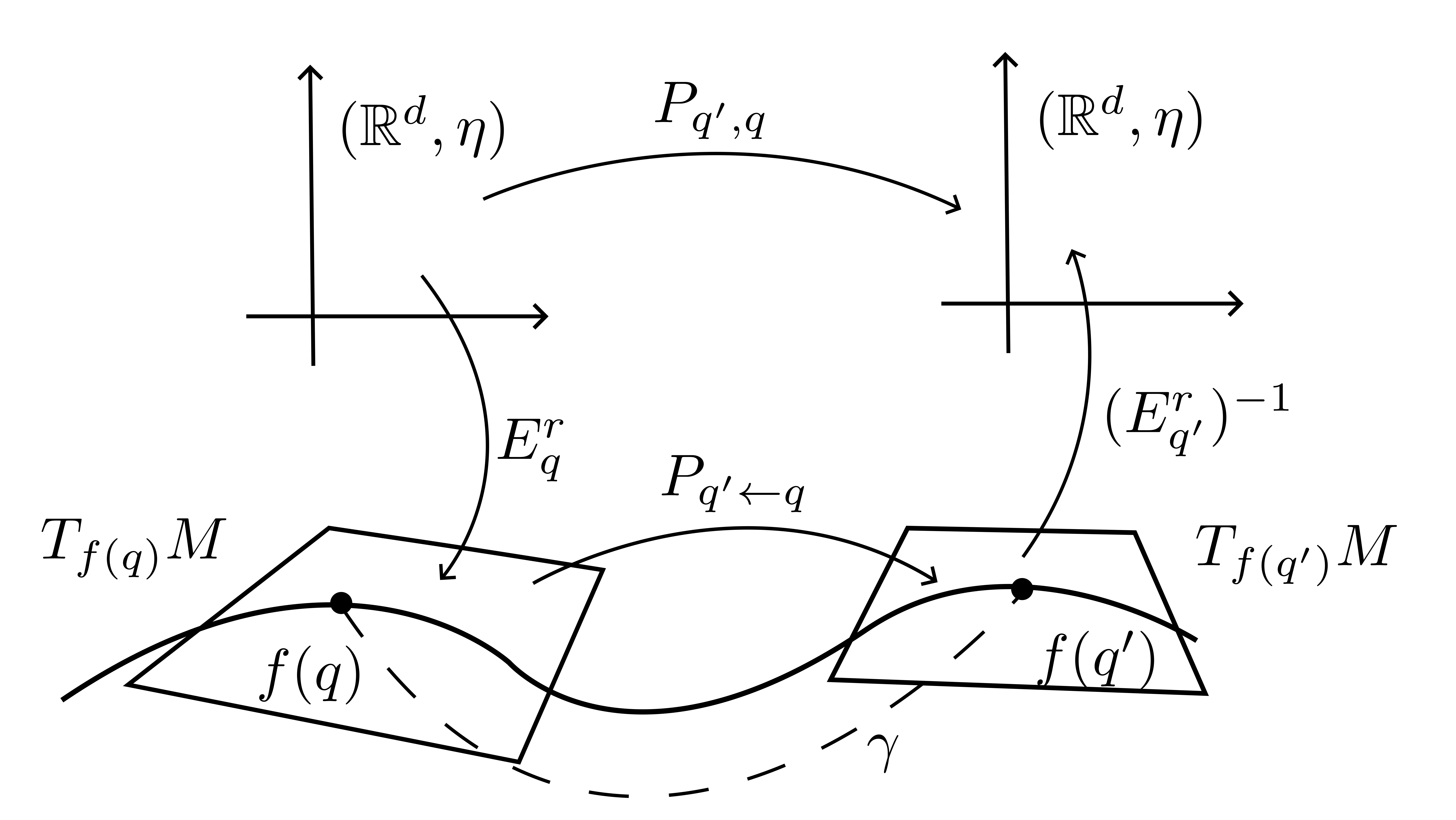}
    \captionof{figure}{Parallel transport $P_{q'\leftarrow q}$ along the short geodesic $\gamma$ maps $T_{f(q)}M$ to $T_{f(q')}M$. The fine-chart frames identify these tangent spaces with $(\mathbb R^d,\eta)$, giving the coordinate matrix $P_{q',q}=(E_{q'}^r)^{-1}P_{q'\leftarrow q}E_q^r$.}
    \label{fig:parallel-transport}
\end{minipage}
\end{center}

\begin{lemma}[Transport accuracy along one edge]
\label{lem:one-edge}
	Let eligible $q,q'$ satisfy $d_h(f(q),f(q'))\leq \delta_0 r$. Let $P_{q' \leftarrow q}$ be Levi-Civita parallel transport along their unique short geodesic and put
	\[ P_{q',q} = (E_{q'}^r)^{-1} P_{q' \leftarrow q} E_q^r. \]
	Fit $\hat P_{q',q} = \mathcal B(X_q^r, X_{q'}^r)$ using all labels in $J_{q,q'}$. Then
	\[ \| \hat P_{q',q}-P_{q',q} \| \leq Cr^2. \]
	The short geodesic may have any causal type. 
\end{lemma}
\begin{proof}
	Write $F=x_{q'}^r\circ(x_q^r)^{-1}$ and $v=\log^g_{f(q')}f(q)$, so $|v|_h=O(r)$. Expressing the Jacobi equation in a frame parallel along the radial geodesic from $f(q')$ to $f(q)$ gives
	\[ (d\exp^g_{f(q')})_v=P_{q\leftarrow q'}(I+O(r^2)). \]
	Its curvature coefficient contains two factors of the radial velocity, so this positive-norm estimate holds for any causal type of $v$. For the corresponding normal-coordinate Jacobian expansion and its relation to a parallel frame, see Poisson--Pound--Vega~\cite[Sec.~8, Eqs.~(8.6)--(8.7)]{Poisson2011}. Since $d((x_q^r)^{-1})_0=E_q^r$, inversion and the chain rule give
	\[ DF(0)=(E_{q'}^r)^{-1}(d\exp^g_{f(q')})_v^{-1}E_q^r=P_{q',q}+O(r^2). \]

	For $c\in J_{q,q'}$, put $x=x_q^r(f(c))$. The fine-domain bound \eqref{eq:fine-chart-domain} gives $|x|=O(r)$. The radial segment $(x_q^r)^{-1}(tx)$, $0\leq t\leq1$, stays $O(r)$ from both bases, hence lies in both true normal-chart neighborhoods for small $r$. Normal-coordinate Christoffel symbols vanish at their bases, so their uniform derivative bounds give $\Gamma_q=O(r)$ and $\Gamma_{q'}=O(r)$ along the segment. The connection transformation rule, with bounded $DF$, gives $D^2F=O(r)$. Taylor's formula therefore yields
	\[ F(x)=F(0)+P_{q',q}x+O(r^3). \]

	By \eqref{eq:fine-chart-accuracy} and boundedness of $P_{q',q}$, both coordinate errors are absorbed into a residual $e_c$ satisfying
	\[ X_{q'}^r(c)=P_{q',q}X_q^r(c)+F(0)+e_c,\qquad |e_c|\leq Cr^3. \]
	The fine-domain bound also gives $|X_q^r(c)-\overline X|\leq Cr$. Apply Lemma~\ref{lem:affine-fit-stability} with $\zeta=Cr^2$ and the covariance bound \eqref{eq:overlap-covariance}.
\end{proof}

The discrepancy between the coordinate-transition derivative and parallel transport contributes to the $O(r^2)$ edge error. Since curvature contributes at order $r^2$ around a loop of side $r$, we use the larger loop scale $H$ to separate the curvature contribution from the accumulated edge error. 

Figure~\ref{fig:coarse-fine-scales} shows how the two chart scales enter the construction: coarse coordinates select the loop, while overlapping fine charts supply the affine fits for its short edges.
\begin{center}
\begin{minipage}{\linewidth}
    \centering
    \captionsetup{hypcap=false}
    \includegraphics[width=.78\linewidth]{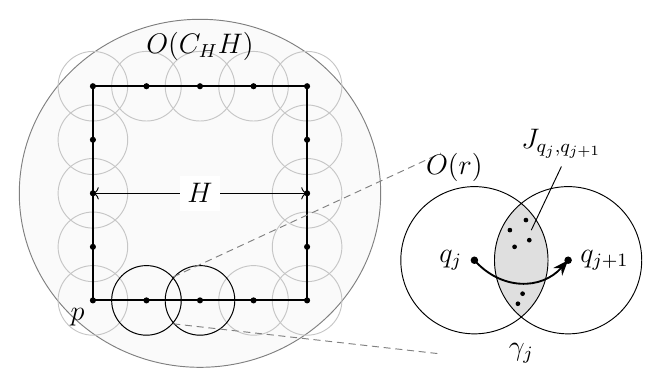}
    \captionof{figure}{A loop of side $H$ inside a coarse neighborhood of scale $C_HH$, with fine neighborhoods of scale $r$ along its edges. The enlargement shows the common sample labels $J_{q_j,q_{j+1}}$ used to fit transport between consecutive fine-chart bases. The curve $\gamma_j$ is the corresponding short geodesic. The neighborhoods and sample locations are schematic.}
    \label{fig:coarse-fine-scales}
\end{minipage}
\end{center}

\subsection{Selecting coarse rectangles from reconstructed coordinates}\label{sec:coarse-rectangle}
Choose a fixed $\ell_0 > 0$ sufficiently small in terms of the coordinate bounds and $\delta_0$, and set
\[ k = \lceil \frac{H}{\ell_0 r} \rceil, \qquad m=4k \]
For $0\leq a < b \leq d-1$, divide each side of the oriented rectangle
\begin{equation}\label{eq:coarse-rectangle}
0 \rightarrow He_a \rightarrow H(e_a + e_b) \rightarrow He_b \rightarrow 0
\end{equation}
into $k$ equal pieces; denote consecutive vertices by $a_0, \dots ,a_m = a_0$. At each nonclosing grid vertex choose a label $q_j \in I_p^H$ minimizing $|X_p^H(q_j)-a_j|$. If several labels attain the same minimum distance, choose one using a fixed ordering of the labels. Prescribe $q_0 = q_m = p$. No realization coordinates are used. 

Figure~\ref{fig:coordinate-plane-loops} illustrates the three coordinate-plane loops in $2+1$ dimensions. Each begins and ends at $p$, so their holonomies are expressed in the same coarse base frame. We write $\gamma_{ab}$ when distinguishing the loop for the coordinate pair $a,b$.
\begin{center}
\begin{minipage}{\linewidth}
    \centering
    \captionsetup{hypcap=false}
    \includegraphics[width=.52\linewidth]{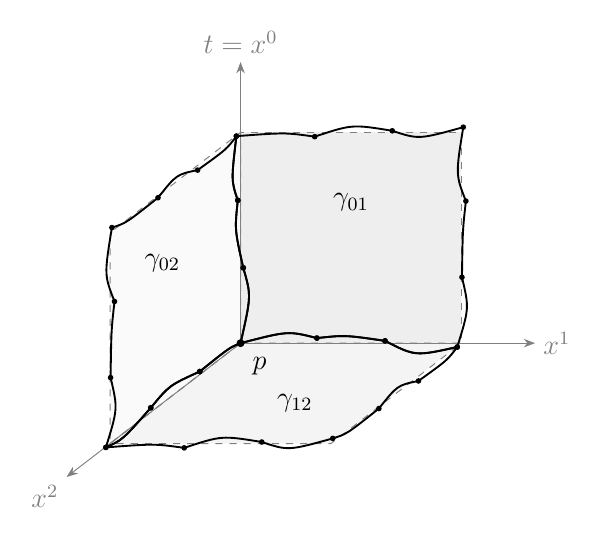}
    \captionof{figure}{The loops $\gamma_{01},\gamma_{02},\gamma_{12}$ based at $p$ in $2+1$ dimensions. The first two use time--space templates; the third uses a spatial template. Dashed rectangles indicate the ideal coordinate templates, and solid curves join the selected sample points by short geodesics. Displacements and bending are exaggerated.}
    \label{fig:coordinate-plane-loops}
\end{minipage}
\end{center}

\begin{lemma}[Loop-vertex accuracy and containment]
\label{lem:loop-vertex-domain}
	In true coarse coordinates, the selected vertices $b_j = x_p^H(f(q_j))$ satisfy
	\begin{equation}\label{eq:loop-vertex-error}
b_j = a_j + d_j, \qquad |d_j|\leq \epsilon := C(H^3+u).
\end{equation}
	For sufficiently small $r$, all vertices are eligible fine-chart centers, consecutive vertices have observer distance at most $\delta_0 r$, and their geodesic polygon has length $O(H)$ and stays in the buffered coarse normal neighborhood.
\end{lemma}
\begin{proof}
	The net supplies a candidate within true coordinate distance $Cu$ of every $a_j$. The margin in \eqref{eq:coarse-chart-domain} places it in $I_p^H$. Its reconstructed distance is at most $C(u+H^3)$, so a minimizing label does at least as well. A second application of \eqref{eq:coarse-chart-accuracy} proves \eqref{eq:loop-vertex-error}. Here $\epsilon =o(r)$, since $H=\sqrt{r}, u=o(r^3)$. 
	
	Ideal edge lengths are comparable to $\ell_0 r$. Uniform coordinate and observer-distance comparisons, with $\ell_0$ fixed sufficiently small, give the required actual edge distances for small $r$. Distinct grid vertices choose distinct labels, except for closure. The coordinate polygon has length at most $CH+C(H/r)\epsilon = O(H)$; short-geodesic estimates give the same bound for its geodesic replacement. All vertices lie $O(H)$ from $p$. For sufficiently small $r$, fine charts are therefore available at these vertices, and the short geodesic edges remain within the coarse normal neighborhood. Geodesics are only comparison objects, not computed by the estimator.
\end{proof}

Figure~\ref{fig:geodesic-polygon} shows the rectangular template and the selected geodesic loop. Write $\gamma_j$ for the short geodesic from $f(q_j)$ to $f(q_{j+1})$, so that $\gamma$ is their concatenation.
\begin{center}
\begin{minipage}{\linewidth}
    \centering
    \captionsetup{hypcap=false}
    \includegraphics[width=.52\linewidth]{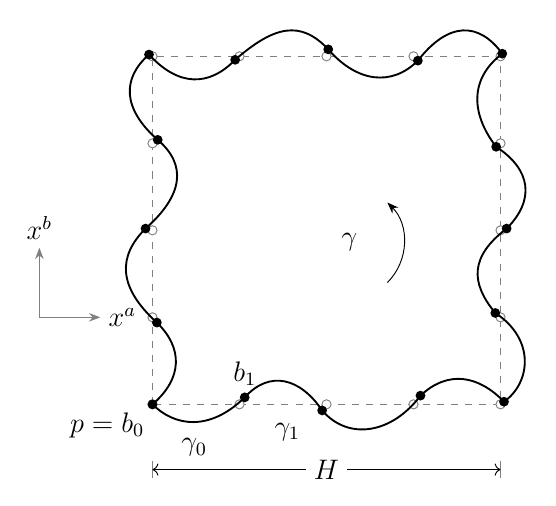}
    \captionof{figure}{A loop in the coarse $x^a,x^b$ coordinate plane. The dashed rectangle of side $H$ is the ideal template; open circles mark its grid vertices. Filled circles mark selected sample points, joined by the short geodesic segments $\gamma_j$. The interior arrow indicates the orientation of $\gamma$. Vertex displacements and geodesic bending are exaggerated. The drawing shows a planar schematic; in spacetime the selected points and geodesic segments may also leave the template plane.}
    \label{fig:geodesic-polygon}
\end{minipage}
\end{center}

\begin{lemma}[Signed area under bounded vertex perturbations]
\label{lem:signed-loop-area}
	For the selected geodesic polygon $\gamma$, let
	\[ A^{ij}(\gamma) = \frac12 \oint_{\gamma} (x^i dx^j -x^j dx^i) \]
	in the true coarse coordinates. Then
	\[ A(\gamma)=H^2 e_a \wedge e_b + O\left( H\epsilon + \frac{H}{r}\epsilon^2 + H^2r^2 \right). \]
	The estimate holds for every choice of vertex perturbations satisfying the stated bound, including polygons that intersect themselves.
\end{lemma}
\begin{proof}
	For straight-coordinate polygons through $a_j$ and $b_j$, their area difference is 
	\[ \frac12 \sum_j (a_j \wedge d_{j+1} + d_j \wedge a_{j+1} + d_j \wedge d_{j+1}). \]
	Since the polygon is closed, indices are understood modulo $m$. Shifting the index in the first sum and using antisymmetry of the wedge product gives
	\[ \sum_j (a_j\wedge d_{j+1} +d_j \wedge a_{j+1}) = \sum_j(a_{j-1} -a_{j+1})\wedge d_j .\]
	Each ideal difference is $O(r)$, including at corners. Summing over $O(H/r)$ edges gives $O(H\epsilon)$ for the linear terms and $O((H/r)\epsilon^2)$ for the quadratic terms. 
	
	Write a geodesic chord as $l_j +w_j$, where $l_j$ is its straight segment, parametrized on $[0,1]$. The coarse normal chart has $\Gamma =O(H)$; the chord speed is $O(r)$. Its equation and $w_j(0)=w_j(1)=0$ imply
	\[ \| w_j \|_\infty + \| w'_j \|_\infty \leq CHr^2. \]
	In the area difference, integration by parts changes $\int l_j \wedge w'_j$. Only the $O(r)$ increment remains, rather than the $O(H)$ absolute edge position. The area change per edge is $O(Hr^3)$. Summing gives $O(H^2r^2)$. 
\end{proof}

\subsection{Holonomy and curvature of the comparison loop}
Define curvature in the coarse base frame by
\[ \mathcal R_{ab}^H(p) = (E_p^H)^{-1} R_{f(p)}(E_p^H e_a, E_p^H e_b) E_p^H. \]
Put $P_j = (E_{q_{j+1}}^r)^{-1}P_{q_{j+1}\leftarrow q_j} E_{q_j}^r$, using the same fine chart when $q_m = q_0 =p$, and set
\[ P_{loop} = P_{m-1} \dots P_0, \qquad Q=(E_p^H)^{-1} E_p^r. \]
The matrix $Q$ converts fine base coordinates to coarse base coordinates in the continuum comparison.

\begin{lemma}[Curvature expansion of loop holonomy]
\label{lem:true-loop-holonomy}
	One has
	\begin{equation}\label{eq:true-loop-holonomy}
	Q P_{loop} Q^{-1} = I-H^2 \mathcal R^H_{ab}(p)+O\left( H^3 + H\epsilon + \frac{H}{r}\epsilon^2 + H^2r^2 \right).
	\end{equation}
	In particular $\| P_{loop} - I \| \leq CH^2$.
\end{lemma}
\begin{proof}
	For connection matrices in the coarse coordinate basis,
	\[ \Gamma_i(0) = 0,\qquad \Gamma_i(x)=\partial_j \Gamma_i(0) x^j + O(|x|^2). \]
	The loop radius and length are $O(H)$, so $\int_{\gamma} \| \Gamma_i dx^i \| \leq CH^2$. The equation $\dot V = -\Gamma_i \dot x^i V$, expanded in iterated integrals, gives
	\[ \mathrm{Hol}(\gamma) = I - \oint_\gamma \Gamma_i dx^i + O(H^4) \]
	For a closed loop $\oint x^j dx^i = A^{ji}$. Antisymmetrization and $\Gamma(0)=0$ yield
	\begin{align*}
		\oint_\gamma \Gamma_i dx^i &= \frac12 \sum_{j,i} (\partial_j \Gamma_i - \partial_i \Gamma_j)(0) A^{ji}(\gamma) + O(H^3) \\
		&= \frac12 \sum_{j,i} \mathcal R_{ji}^H (p) A^{ji}(\gamma) +O(H^3).
	\end{align*}
	The factor $1/2$ accounts for both antisymmetric index orders. For the rectangle \eqref{eq:coarse-rectangle}, the nonzero leading area components are $A^{ab}=H^2, A^{ba} = -H^2$. Lemma~\ref{lem:signed-loop-area} gives \eqref{eq:true-loop-holonomy}, with the displayed minus sign. This is coarse-frame holonomy, namely $QP_{loop} Q^{-1}$, not $P_{loop}$. The identity bound also follows directly from the transport equation. 
\end{proof}

\subsection{Accumulation of transport errors around a loop}
Keep the independent fine-overlap affine fits $\hat P_j$, and set $\hat P_{loop} = \hat P_{m-1} \dots \hat P_0$.
\begin{lemma}[Accumulation of edge errors]
\label{lem:loop-accumulation}
	The one-edge estimates imply
	\begin{equation}\label{eq:loop-accumulation}
	\| \hat P_{loop} - P_{loop} \| \leq CH r.
	\end{equation}
\end{lemma}
\begin{proof}
	Use the coarse normal coordinate basis as one smooth frame throughout the loop. Define
	\[ A_j = (dx_p^H)_{f(q_j)}E^r_{q_j}, \qquad T_j = A_{j+1} P_j A_j^{-1}, \qquad \hat T_j = A_{j+1} \hat P_j A_j^{-1}. \]
	The $A_j$ and their inverses are bounded even when neighboring fine charts use independently chosen frames. In this smooth frame the connection is $O(H)$ and edges have length $O(r)$. Consequently,
	\[ \| T_j-I \| \leq CHr, \qquad \| \hat T_j - T_j \| \leq Cr^2, \qquad \| \hat T_j - I \| \leq C(Hr+r^2). \]
	Every contiguous partial product is bounded by $\exp(Cm(Hr+r^2)) \leq \exp(C(H^2+Hr))$, hence uniformly bounded. The exact noncommutative identity is
	\[ \hat T_{m-1} \dots \hat T_0 - T_{m-1} \dots T_0 = \sum_{j=0}^{m-1} (\hat T_{m-1} \dots \hat T_{j+1})(\hat T_j - T_j)(T_{j-1} \dots T_0), \]
	with empty products $I$. Its norm is $O(mr^2)=O(Hr)$. The same fine chart is used on returning to $p$, so $A_m = A_0 = Q$. Endpoint conversions cancel exactly in each product; one bounded conjugation proves \eqref{eq:loop-accumulation}. 
\end{proof}

Expressing the independently fitted edge matrices in one smooth frame gives uniform bounds for their partial products. The changes of frame cancel at intermediate vertices, leaving only one bounded conjugation at the base.

\subsection{Fine-to-coarse alignment and the estimator}
Fit the coarse coordinates against the fine coordinates at the same base, allowing a translation and giving equal weight to every label in $I_p^r\subset I_p^H$:
\begin{equation}\label{eq:coarse-fine-fit}
\hat Q = \mathcal B(X^r_p , X^H_p).
\end{equation}
Define the data-derived estimator
\begin{equation}\label{eq:curvature-estimator}
\hat R_{ab}(p) = -\frac{\hat Q \hat P_{loop} \hat Q^{-1}-I}{H^2}.
\end{equation}

Figure~\ref{fig:sampled-holonomy-alignment} follows a vector around the sampled loop and shows how the fitted base-frame conversion enters the curvature estimate.
\begin{center}
\begin{minipage}{\linewidth}
\centering
\captionsetup{hypcap=false}
\includegraphics[width=.98\linewidth]{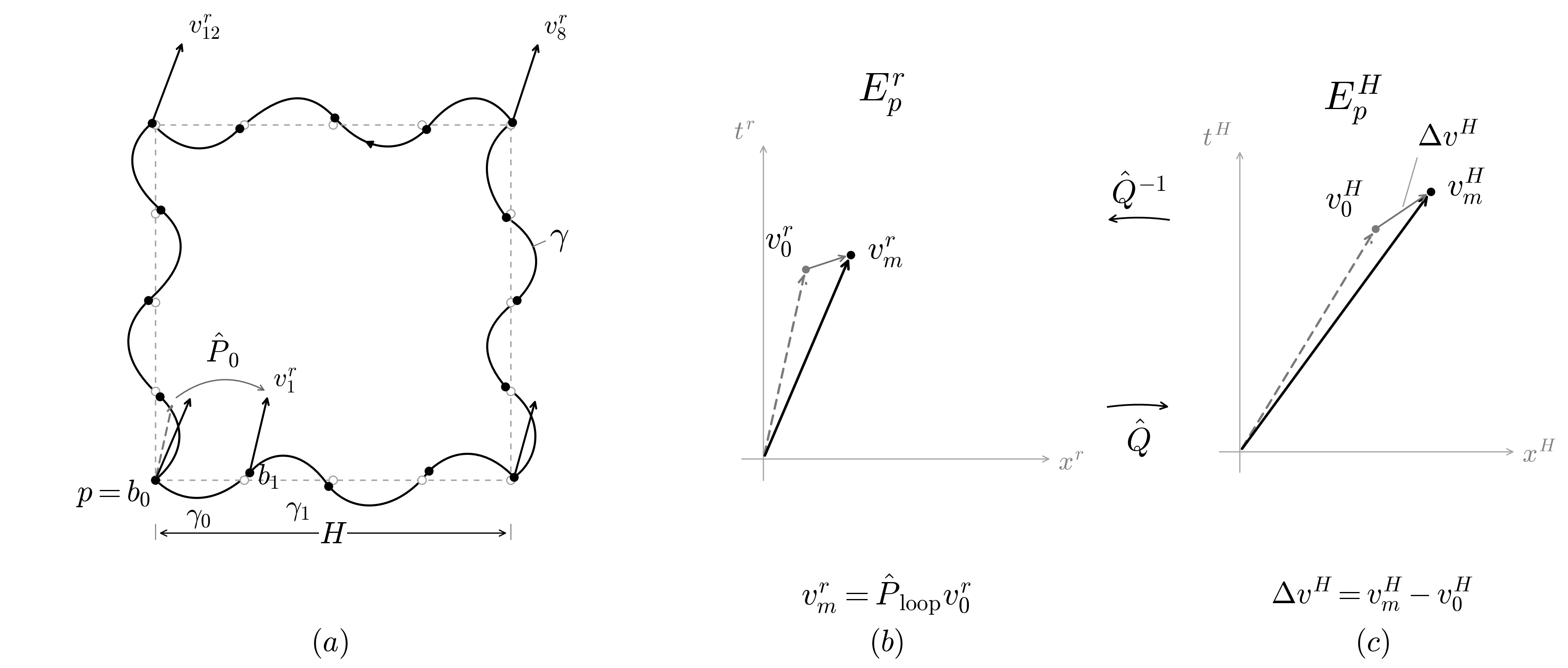}
\captionof{figure}{(a) The sampled loop from Figure~\ref{fig:geodesic-polygon}, with the same vertices and curves. The loop $\gamma$ is the concatenation of the short geodesics $\gamma_j$. The arrowhead on the curve indicates its orientation. The drawn vectors illustrate $v_{j+1}^r=\hat P_jv_j^r$; the curved arrow $\hat P_0$ maps the initial vector to the vector at $b_1$. After the ordered product, $v_m^r$ returns to the same base $p$, where the same fine chart is used. (b) An enlarged comparison of the initial and returned vectors at $p$: the initial vector $v_0^r$ is dashed and the returned vector $v_m^r$ is solid. (c) Conversion to the coarse base coordinates gives $v_0^H=\hat Qv_0^r$ and $v_m^H=\hat Qv_m^r$. Their difference is $\Delta v^H=(\hat Q\hat P_{\mathrm{loop}}\hat Q^{-1}-I)v_0^H$, so $-\Delta v^H/H^2=\hat R_{ab}(p)v_0^H$. Vector changes and geodesic bending are exaggerated; the plots are schematic projections.}
\label{fig:sampled-holonomy-alignment}
\end{minipage}
\end{center}

\begin{lemma}[Stability of the aligned loop estimate]
\label{lem:base-frame-alignment}
	For small $r$, $\hat Q$ has uniformly bounded inverse, and
	\begin{equation}\label{eq:base-frame-error}
	\| \hat Q - Q \| \leq C(H^3/r + r^2).
	\end{equation}
	Furthermore,
	\begin{equation}\label{eq:aligned-loop-error}
	\| \frac{\hat Q \hat P_{loop} \hat Q^{-1} -I}{H^2} - \frac{QP_{loop} Q^{-1} - I}{H^2} \| \leq C \left( \frac{r}{H} + \frac{H^3}{r} + r^2 \right).
	\end{equation}
\end{lemma}
\begin{proof}
	True normal coordinates with the same center differ exactly by $Q$. On $I_p^r$, therefore,
	\[ X_p^H(c) = Q X_p^r(c) + O(H^3+r^3). \]
	The base covariance and Lemma~\ref{lem:affine-fit-stability} give \eqref{eq:base-frame-error}. Its bound tends to zero at $H=\sqrt{r}$, so bounded invertibility of $Q$ gives that of $\hat Q$.
	
	Write $Z=P_{loop} - I$. Since conjugation fixes $I$, 
	\[ \hat Q P_{loop} \hat Q^{-1} - Q P_{loop} Q^{-1} = (\hat Q - Q)Z \hat Q^{-1} + QZ(\hat Q^{-1} - Q^{-1}). \]
	Use $\hat Q^{-1} - Q^{-1} = \hat Q^{-1}(Q-\hat Q)Q^{-1}$ and $\| Z \| \leq CH^2$. This costs at most $C \| \hat Q - Q \| H^2$. Replacing $P_{loop}$ by $\hat P_{loop}$ costs $CHr$ by Lemma~\ref{lem:loop-accumulation}. Divide by $H^2$. 
\end{proof}

\subsection{The curvature reconstruction theorem}
\label{sec:volume-reconstruction}
The established chart and loop estimates now give the count-based reconstruction theorem.

\begin{theorem}[Volume-based curvature reconstruction]\label{thm:volume-curvature}
In the fixed geometric setting of Section~\ref{sec:geo}, assume the deterministic volume conditions of Definition~\ref{def:deterministic-volume-event}. The chart and loop thresholds can be fixed so that, for all sufficiently large $\rho$, the constructions of Sections~\ref{sec:sampled-charts} and~\ref{sec:transport-curvature} succeed at every base $p\in C_\rho[4\theta]$. With
\[
r=(L_\rho/\rho)^{1/(3d+4)},\qquad H=\sqrt r,
\]
the estimator~\eqref{eq:curvature-estimator} satisfies
\begin{equation}\label{eq:volume-curvature-rate}
\sup_{p\in C_\rho[4\theta]}\max_{0\leq a<b\leq d-1}
\bigl\|\hat R_{ab}(p)-\mathcal R_{ab}^{H}(p)\bigr\|
\leq C\left(\frac{L_\rho}{\rho}\right)^{1/[2(3d+4)]}.
\end{equation}
Here $\mathcal R_{ab}^{H}(p)$ is the true curvature expressed in the single coarse comparison frame $E_p^H$. The constants depend on the fixed geometry, $\theta$, the count allowance $K$, and the fixed construction thresholds. The estimate holds simultaneously for every eligible base, every accepted chart tuple, and every choice among equally close minimizing labels.
\end{theorem}
\begin{proof}
Proposition~\ref{prop:count-chart-guarantees} establishes the required chart domains, coordinate errors, frame bounds, sample coverage, and covariance estimates for the chosen count-based construction. Lemma~\ref{lem:loop-vertex-domain} then supplies the short edges and loop domains. The covariance bounds make every edge fit and the fine-to-coarse fit well-defined; Lemma~\ref{lem:base-frame-alignment} gives the remaining inverse.

The true holonomy expansion~\eqref{eq:true-loop-holonomy} and aligned estimation error~\eqref{eq:aligned-loop-error}, divided by the loop area where required, give
\[
\bigl\|\hat R_{ab}(p)-\mathcal R_{ab}^{H}(p)\bigr\|
\leq C\left[H+\frac rH+\frac\epsilon H+\frac{\epsilon^2}{rH}+r^2+\frac{H^3}{r}\right],
\qquad \epsilon=C(H^3+u).
\]
At $H=\sqrt r$ and $u=o(r^3)$, one has $\epsilon=O(H^3)$. The six terms have orders $H,H,H^2,H^3,H^4,H$, respectively, so the error is at most $CH$. Substituting $H=(L_\rho/\rho)^{1/[2(3d+4)]}$ proves~\eqref{eq:volume-curvature-rate}. Every estimate used above is simultaneous over the eligible bases, accepted chart tuples, and minimizing-label choices in the statement.
\end{proof}

The $\mathcal R_{ab}^H$ determine the entire curvature endomorphism-valued two-form in the coarse Lorentz frame. Lowering the output index using $\eta$ gives the Riemann tensor. A fixed bounded linear projection onto the algebraic curvature symmetries preserves the rate. Comparison uses the frame $E_p^H$ associated with the chosen coarse chart. This frame may depend on the sample, and the theorem does not require the frames themselves to converge. 

\begin{corollary}[Poisson curvature reconstruction]\label{cor:poisson-curvature}
Let $C_\rho$ be the induced order on a Poisson sample in $D$ of intensity $\rho\mu$. For every $N>0$, the construction and estimate of Theorem~\ref{thm:volume-curvature} hold with probability at least $1-\rho^{-N}$ for sufficiently large $\rho$, with $C$ replaced by $C_N$. Under the marked-Poisson coupling defined in Section~\ref{sec:poisson-volume}, they hold almost surely for every sufficiently large real $\rho$, with deterministic construction thresholds and an error constant, and a finite random density threshold.
\end{corollary}
\begin{proof}
Apply Theorem~\ref{thm:volume-curvature} on the event supplied by Theorem~\ref{thm:poisson-volume-bounds}. For the coupled conclusion use Corollary~\ref{cor:poisson-eventual-volume-bounds}. The deterministic result controls all eligible choices on each of these events.
\end{proof}

For $a<b$, extend the estimated matrices by $\hat R_{ba}(p)=-\hat R_{ab}(p)$ and $\hat R_{aa}(p)=0$. Writing $(\hat R_{ab})^c{}_e$ for their matrix entries, define
\[
\widehat{\operatorname{Scal}}(p)
 =\sum_{a,b=0}^{d-1}\eta^{bb}(\hat R_{ab}(p))^a{}_b.
\]
Here the diagonal entries of $\eta^{-1}$ are $\eta^{00}=-1$ and $\eta^{bb}=1$ for $b>0$.

\begin{corollary}[Scalar-curvature reconstruction]\label{cor:scalar-curvature}
Under the hypotheses of Theorem~\ref{thm:volume-curvature}, for all sufficiently large $\rho$,
\[
\sup_{p\in C_\rho[4\theta]}
\left|\widehat{\operatorname{Scal}}(p)-\operatorname{Scal}_g(f(p))\right|
\leq C_{\mathrm{sc}}\left(\frac{L_\rho}{\rho}\right)^{1/[2(3d+4)]}.
\]
The estimate holds uniformly over accepted chart tuples and every choice among equally close minimizing labels. It also has the high-probability and eventual almost-sure Poisson conclusions of Corollary~\ref{cor:poisson-curvature}. The scalar target is independent of the comparison frame.
\end{corollary}
\begin{proof}
With the curvature convention of Lemma~\ref{lem:true-loop-holonomy}, scalar curvature in any Lorentz-orthonormal frame is
\[
\operatorname{Scal}_g(f(p))
 =\sum_{a,b=0}^{d-1}\eta^{bb}(\mathcal R^H_{ab}(p))^a{}_b.
\]
The contraction is a fixed linear map on the finite collection of curvature matrices. Applying it to the error in Theorem~\ref{thm:volume-curvature} proves the estimate, with a dimension-dependent factor absorbed into $C_{\mathrm{sc}}$. The Poisson statements follow on the same events as in Corollary~\ref{cor:poisson-curvature}.
\end{proof}

This gives a scalar consequence of the full tensor reconstruction, complementing scalar-curvature constructions based on causal-set wave operators~\cite{BenincasaDowker2010,BelenchiaBenincasaDowker2016} and chain abundances~\cite{RoySinhaSurya2013}.

% ============================================================
\section{Alternative interval estimators}\label{sec:alternative-estimators}
\subsection{Accuracy required of an interval estimator}
The coordinate construction uses estimates of squared proper time. For timelike pairs in a convex normal neighborhood, write
\[
D_g(x,y)=2\sigma_g(x,y)=-\tau_g(x,y)^2,
\]
where $\sigma_g$ is Synge's world function, half the signed squared geodesic interval. Suppose a symmetric scalar estimator $\hat D$, with $\hat D(c,c)=0$, satisfies
\begin{equation}\label{eq:alternative-scalar-accuracy}
|\hat D(c,c')-D_g(f(c),f(c'))|\leq CE_\rho(s)s^2
\end{equation}
uniformly on every distinct pair used in the chart tests. Assume these distinct pairs are localized in a normal neighborhood of size $O(s)$ and have proper time at least a fixed positive multiple of $s$. Retain the Gram eigenvalue bounds, full active sets, covariance lower bounds, and domain assumptions used above. Take $E_\rho(s)\geq s^2$ to include the $O(s^4)$ normal-coordinate squared-interval error, and $E_\rho(s)\to0$ at the reconstruction scales.

Under these conditions the proof of Proposition~\ref{prop:accepted-chart-soundness} applies with this scalar estimator: its coordinate error is $O(E_\rho(s)s)$. If the fine and coarse errors satisfy the bounds of Section~\ref{sec:transport-assumptions}, the same transport and curvature theorem applies. Other accuracy budgets can be assessed through the edge, loop-area, and frame-alignment estimates proved above.

For comparable labels the volume and longest-chain candidates are
\[
\hat D_{\mathrm{vol}}(c,c')=-\left(\frac{N(c,c')}{\rho\kappa_d}\right)^{2/d},
\qquad
\hat D_{\mathrm{ch}}(c,c')=-\left(\frac{\ell_C(c,c')}{c_d^{\mathrm{chain}}\rho^{1/d}}\right)^2.
\]
Here $\ell_C(c,c')$ is the maximum number of links in a chain from the earlier label to the later one, and $c_d^{\mathrm{chain}}>0$ is a specified flat-spacetime calibration constant for chain length~\cite{BrightwellGregory1991}. It is distinct from the diamond-volume constant $\kappa_d$. Both estimators are extended symmetrically and assigned zero on the diagonal. Counts may still be used to select and localize the pairs on which the chain estimator is evaluated.

For the chain construction, one must establish a uniform two-sided estimate of the form~\eqref{eq:alternative-scalar-accuracy} on the tested pairs, together with the probability statement and scale bounds required for the desired curvature conclusion. The preceding volume theorem supplies this input for $\hat D_{\mathrm{vol}}$.

\subsection{A curvature estimate for general sampled charts}\label{sec:transport-assumptions}
The volume construction above establishes its chart properties from interval counts. We now state the deterministic estimate for general sampled charts, so that other interval estimators can use the same transport and loop analysis once they supply these properties. Here an eligible label is one at which the fine-chart data described below are available.

Let $(M,g)$ be a fixed smooth Lorentzian $d$-manifold, $d \geq 2$, with positive observer metric $h$. Fix compact sets $K_0 \subset \mathrm{int} K_1 \subset \mathrm{int} K_2$ in a smooth buffered region. Norms use $h$ and Euclidean coordinate norms. Fix uniform bounds for normal-coordinate maps and their inverses, their derivatives needed below, the connection and its first two coordinate derivatives, and short-geodesic endpoint maps. Smooth fixed geometry on a compact buffer supplies such bounds after the normal-neighborhood radius is reduced.

We choose the true normal charts on neighborhoods containing $h$-balls of one fixed sufficiently small radius about their bases. The sampled maps are required only on the active label sets. This distinction allows us to use smooth coordinate segments in the proof without assuming that those segments are part of the finite sample.

Let $\mathcal C$ be a finite labelled sample with injective realization $f: \mathcal C \to M$. The map $f$ is used only to analyze the construction. It is not an input to the estimator. We occasionally write $q$ for $f(q)$. Let $r \to 0, H = \sqrt{r}$, and let the sample-net radius satisfy $u=o(r^3)$. Assume the following conditions hold uniformly for every selected base $p$, with $f(p) \in K_0$.
\begin{enumerate}
	\item \textit{Net and available centers.} Every continuum point in the neighborhoods used below has a sample within $h$-distance $u$. Every label within $C_* H$ of $f(p)$ is eligible for a fine chart, where the fixed $C_*$ contains the entire coarse active domain and its fine-chart buffers. These neighborhoods lie in $K_1$; all normal-chart and geodesic constructions stay in $K_2$.
	\item \textit{Fine charts and their domains.} At every such $q$, a known active label set $I_q^r$ and data-derived map $X_q^r : I_q^r \to \mathbb R^d$ are supplied. There is one Lorentz-orthonormal frame $E_q^r : \mathbb R^d \to T_{f(q)}M$, constant for the entire chart, such that
	\begin{equation*}
	X_q^r(q)=0, \quad x_q^r(z)=(E_q^r)^{-1} \log_{f(q)}^g z, \quad |X_q^r(c)-x_q^r(f(c))| \leq Cr^3,
	\end{equation*}
	\begin{equation*}
	\{ c: f(c) \in B_h(f(q),c_0 r) \} \subset I_q^r \subset \{c:f(c) \in B_h(f(q),C_0 r) \}.
	\end{equation*}
	The frame and its inverse have norm at most a fixed $K_{\mathrm{fr}}$. The estimates below hold uniformly over these frame choices.
	\item \textit{A sufficiently large coarse chart.} At $p$, a known active set $I_p^H$ and recentered coordinate map obey
	\begin{equation*}
	X_p^H(p)=0, \quad x_p^H(z) = (E_p^H)^{-1}\log_{f(p)}^g z, \quad |X_p^H(c)-x_p^H(f(c))| \leq CH^3.
	\end{equation*}
	Here $E_p^H$ is one Lorentz-orthonormal frame with the same bounds. For fixed $R_0>3, R_1 < \infty$, the true chart is defined on a buffered coordinate ball containing the sets below, and
	\begin{equation*}
	\{ c: |x_p^H(f(c))| \leq R_0 H\} \subset I_p^H \subset \{c: |x_p^H(f(c))| \leq R_1 H \}.
	\end{equation*}
	Inclusions are understood within the fixed normal neighborhood. For small $r$, $I_p^r \subset I_p^H$, by \eqref{eq:fine-chart-domain} and uniform coordinate bounds.
	\item \textit{Covariance bounds for fitting sets.} Fix $\delta_0 , \kappa > 0$. If eligible $q,q'$ have $d_h(f(q),f(q')) \leq \delta_0 r$, the common set $J_{q,q'} = I_q^r \cap I_{q'}^r$ is nonempty and its uniform empirical covariance in the source coordinates satisfies
	\begin{equation*}
	\frac{1}{|J_{q,q'}|} \sum_{c \in J_{q,q'}} (X_q^r(c)-\overline X)(X_q^r(c)-\overline X)^T \succeq \kappa r^2 I.
	\end{equation*}
	\item The same bound holds for the uniform covariance of $X_p^r$ on $I_p^r$. The covariance acceptance tests are computable from the supplied data.
\end{enumerate}

Each reconstructed chart may differ from a reference normal chart by a constant Lorentz transformation. This transformation may vary between charts, but it and its inverse have uniformly bounded norms. If $X_q^r = \Lambda_q \chi_q + O(r^3)$ with $\chi_q = e_q^{-1} \log_{f(q)}^g$, set $E_q^r = e_q \Lambda_q^{-1}$. Recentering at the base removes any fixed translation and at most doubles the coordinate error.

\begin{lemma}[Covariance bound from separated sample clusters]\label{lem:sufficient-mass-design}
	Suppose a fitting set lies in a coordinate ball of radius $Cr$. Suppose $d+1$ clusters each carry at least a fixed fraction $\alpha > 0$ of its empirical mass. Their true coordinate centers form a simplex with difference matrix of smallest singular value at least $\sigma r > 0$. If each cluster has radius at most $\beta r$, with $\beta$ sufficiently small depending on $\sigma, d$, then uniform coordinate errors $o(r)$ preserve a covariance lower bound $\kappa r^2 I$. 
\end{lemma}
\begin{proof}
	For every Euclidean unit covector $\xi$, two simplex centers have projection difference at least $c\sigma r$. Small cluster radii and the $o(r)$ error preserve a fixed fraction of this separation for every pair from those clusters. The identity
	\[ \frac{1}{N} \sum_c [ \xi \cdot (X_c - \overline X)]^2 = \frac{1}{2N^2} \sum_{c,c'} [\xi \cdot (X_c - X_{c'})]^2 \]
	gives a lower bound $c\alpha^2 \sigma^2 r^2$, uniformly in $\xi$. 	
\end{proof}

For regular-diamond counting, place small disjoint diamonds around a simplex inside an overlap ball. Each must contain at least $c\rho r^d$ labels, and a finite regular-diamond cover must bound the whole fitting set by $C\rho r^d$ labels.

\begin{theorem}[Deterministic two-scale curvature consistency]
\label{thm:curvature-consistency}
	Assume the sampling and chart hypotheses of Section~\ref{sec:transport-assumptions} uniformly at the selected base labels $p$, with $f(p)\in K_0$. The algorithm \eqref{eq:affine-fit}, \eqref{eq:coarse-rectangle}, \eqref{eq:curvature-estimator} is well-defined for small $r$. For every $a < b$, 
	\begin{equation}\label{eq:curvature-error}
	\| \hat R_{ab}(p) - \mathcal R_{ab}^H(p) \| \leq C \left[ H + \frac{r}{H} + \frac{\epsilon}{H} + \frac{\epsilon^2}{rH} + r^2 + \frac{H^3}{r} \right], \qquad \epsilon=C(H^3+u).
	\end{equation}
	At $H=\sqrt{r}, u=o(r^3)$, this yields
	\begin{equation}\label{eq:curvature-rate}
	\sup_{p, a<b} \| \hat R_{ab}(p) - \mathcal R_{ab}^H(p) \| \leq C\sqrt{r}.
	\end{equation}
	Constants depend only on the fixed geometric, domain, frame, coordinate-error, and covariance constants in the sampling and chart hypotheses. The bound holds uniformly over chart choices satisfying these hypotheses and every choice among equally close minimizing labels.
\end{theorem}
\begin{proof}
The analytic estimates of Section~\ref{sec:transport-curvature} use the chart accuracy, domains, frame bounds, covariance, and sample coverage just listed. The one-edge proof uses the fine-coordinate error and overlap covariance; vertex selection uses the coarse-coordinate error, the net, and the domain margins; the area and holonomy arguments use the fixed geometry; and accumulation and alignment use the frame and base-covariance bounds. Thus those proofs apply to these general sampled charts.

	Lemma~\ref{lem:loop-vertex-domain} supplies the required edges and domains. Covariance gives each affine fit; Lemma~\ref{lem:base-frame-alignment} supplies the remaining inverse. Divide \eqref{eq:true-loop-holonomy} by $H^2$ and add \eqref{eq:aligned-loop-error}, proving \eqref{eq:curvature-error}. At $H=\sqrt{r}, \epsilon = O(H^3)$; the six bracketed terms have orders $H,H,H^2,H^3,H^4,H$, respectively.
\end{proof}

\subsection{Agreement of two interval estimators}
The following result compares two constructions once their respective accuracy and geometric bounds hold.

\begin{proposition}[Comparison of interval estimators]\label{prop:estimator-comparison}
Under the preceding localization, Gram-matrix, covariance, and domain assumptions, let two scalar estimators satisfy~\eqref{eq:alternative-scalar-accuracy} with budgets $E_A(s)$ and $E_B(s)$. At a common base $p$, suppose their accepted charts use a common full-dimensional active set. There is a bounded Lorentz matrix
\[
L_{p,s}=\Lambda_{B,p,s}\Lambda_{A,p,s}^{-1}
\]
such that
\[
\sup_c|X_B(c)-L_{p,s}X_A(c)|\leq C(E_A(s)+E_B(s))s.
\]
If the centered empirical covariance on this set is at least $\kappa s^2I$, fitting a matrix and translation estimates $L_{p,s}$ with matrix error at most $C(E_A(s)+E_B(s))$.

Suppose both curvature constructions have errors $b_A,b_B$ relative to the true tensors in their coarse frames. Write $L_p^H$ for the Lorentz transformation between these frames. Then their full $(1,3)$ tensor estimates satisfy
\[
\|\hat{\mathcal R}_B-L_p^H\cdot\hat{\mathcal R}_A\|\leq C(b_A+b_B),
\]
where frame conversion acts on every slot:
\[
(L\cdot\mathcal R)(u,v)w=L\mathcal R(L^{-1}u,L^{-1}v)L^{-1}w.
\]
If $L_p^H$ is replaced by a fitted matrix whose error is $\delta$, the comparison bound becomes $C(b_A+b_B+\delta)$ for sufficiently small $\delta$ and bounded curvature estimates.
\end{proposition}
\begin{proof}
Write $X_A=\Lambda_A\chi_p+e_A$ and $X_B=\Lambda_B\chi_p+e_B$ in the same reference normal chart. Subtraction gives $X_B-L_{p,s}X_A=e_B-L_{p,s}e_A$, proving the coordinate bound. For the affine fit, centered source coordinates have size $O(s)$ and the centered residual has size $O((E_A+E_B)s)$. Its cross-covariance therefore has norm $O((E_A+E_B)s^2)$. Multiplication by the covariance inverse, whose norm is at most $C/s^2$, proves the matrix bound. The true curvature tensors in the two coarse frames are related by the displayed all-slot transformation. Subtract them and apply the triangle inequality. Uniform bounds for $L_p^H$ and its inverse control the transformation norm. Perturbing the matrix and its inverse by $O(\delta)$ changes the transformed bounded tensor by $O(\delta)$.
\end{proof}

Thus scalar estimators with the required accuracy and geometric tests recover the same limiting curvature after frame conversion. Their finite-density errors and convergence rates depend on their respective scalar accuracy bounds.

\section{Conclusion}\label{sec:conclusion}
We have constructed approximate normal coordinates on sample points in Section~\ref{sec:sampled-charts}. Using these coordinates, Section~\ref{sec:transport-curvature} estimates parallel transport along geodesic segments joining nearby sample points and recovers estimates of the full Riemann tensor from loops. The inputs are causal order and element counts, with density, dimension, and fixed construction thresholds specified. The construction applies to samples in a fixed smooth globally hyperbolic spacetime satisfying the stated uniform volume bounds. Timelike interval counts supply estimates of squared proper time, overlapping normal charts supply transport, and transport around shrinking loops supplies curvature. The fine scale $r$ controls chart size and the separation of consecutive loop vertices, while the larger scale $H$ controls loop size. This separation of scales allows the accumulated transport error to remain smaller than the leading curvature contribution to holonomy.

Theorem~\ref{thm:volume-curvature} gives the uniform curvature error
\[
O\!\left(\left[\frac{\log(e+\rho)}{\rho}\right]^{1/[2(3d+4)]}\right)
\]
at sample points satisfying the loop-base past- and future-count thresholds of Section~\ref{sec:eligible-bases}, simultaneously over accepted chart choices. Corollary~\ref{cor:poisson-curvature} gives the Poisson realization: the estimate holds with high probability and eventually almost surely under the marked-Poisson coupling introduced in Section~\ref{sec:poisson-volume}. Curvature is compared in the reconstructed Lorentz frames, whose changes of frame are uniformly bounded. The result therefore controls the full local tensor while allowing each chart to choose its own frame.

Corollary~\ref{cor:scalar-curvature} gives a scalar-curvature estimate with the same uniform rate, whose continuum target is independent of the reconstructed frame.

The scalar formulation in Section~\ref{sec:alternative-estimators} also identifies how other interval estimators can enter the construction. A uniform squared-interval accuracy bound, together with the stated geometric hypotheses and matrix bounds that keep coordinate reconstruction and overlap fits stable, gives the same coordinate construction and supports the transport and curvature estimates. The longest-chain estimator is one candidate, conditional on its own uniform two-sided accuracy estimate. Proposition~\ref{prop:estimator-comparison} shows that when two such constructions converge, their curvature tensors agree after frame conversion, with errors determined by their respective accuracy bounds.

Further work can seek sharper convergence rates and practical methods for selecting anchors and computing the overlap fits. Numerical studies can examine the size of the errors at finite density and guide these algorithmic choices. The present construction gives a quantitative local realization of the order-and-number principle, extending interval measurements to the comparison of tangent spaces and their curvature.

\section{Acknowledgments}
% ============================================================
The author thanks David Morrison for valuable discussions and guidance throughout the development of this work.

\appendix
\section{Stability of light-cone arrival times}
\label{app:light-cone-stability}

Fixing a spatial coordinate $z$ gives a timelike coordinate line.
Its intersection with the future light cone of a tip $a$ defines an
arrival time $\alpha_a(z)$. We need to control how this time changes
when the tip moves. Away from the cone vertex, the world function gives
an implicit equation for the intersection. At the vertex, where that
equation cannot be differentiated to solve for time, a direct cone
comparison supplies the same bound.

\begin{lemma}[Stability of light-cone arrival times]
\label{lem:arrival-time-stability}
Let $W$ be a convex normal neighborhood with coordinates $(t,z)$
such that $\nabla t$ is past timelike and $\partial_t$ is future timelike.
Let $P\Subset Q\Subset W$ be coordinate cylinders, each a product
of a time interval and a Euclidean spatial ball, and let $B$ be a
closed spatial ball in the spatial base of $Q$.
Suppose that, for every $a\in\overline P$ and $z\in B$, the future
null cone of $a$ in $W$, including its vertex, meets the line with
spatial coordinate $z$ at a unique point
$(\alpha_a(z),z)$ in $Q$, and that these intersections lie in a fixed
compact subset of $Q$. Then there is a constant $C$ such that
\[
|\alpha_a(z)-\alpha_{\tilde a}(z)|
\leq C\,d_h(a,\tilde a)
\qquad
(a,\tilde a\in\overline P,\ z\in B).
\]
The analogous statement holds for the past-cone intersection times
$\beta_b(z)$.
\end{lemma}

The common-domain assumption makes explicit the room needed to compare
graphs: both cones must be defined over the same spatial ball.
For tips in a sufficiently small inner cylinder, uniform local cone
comparison ensures this by taking a larger time interval for $Q$.
In the small-diamond application all tips and relevant spatial points
are within $O(s)$ of the center, so an interval of length $Cs$ suffices.
All constants can be taken from the fixed compact coordinate
neighborhoods, independently of $s$.

\begin{proof}
For the affinely parametrized geodesic $\gamma:[0,1]\to W$ from $a$
to $y$, Synge's world function is
\[
\sigma(a,y)=\frac12\int_0^1
g_{\gamma(\lambda)}(\dot\gamma(\lambda),\dot\gamma(\lambda))\,d\lambda.
\]
Differentiating this energy and integrating by parts, the geodesic
equation eliminates the interior term and gives
\[
\delta\sigma
=g_y(\dot\gamma(1),\delta y)-g_a(\dot\gamma(0),\delta a);
\]
see \cite[Sec.~3.3, Eqs.~(3.3)--(3.4)]{Poisson2011}.
For $y=(\alpha_a(z),z)\ne a$, the geodesic is null and
$g_y(\dot\gamma(1),\partial_t)\ne0$. The implicit function theorem
therefore applies to $\sigma(a,(\alpha_a(z),z))=0$, giving
\[
D_a\alpha_a(z)[v]
=\frac{g_a(\dot\gamma(0),v)}
       {g_y(\dot\gamma(1),\partial_t)}.
\]
On the fixed compact coordinate region, uniform timelikeness of $\partial_t$ and boundedness of $g$ imply
\[
|g_y(\dot\gamma(1),\partial_t)|
\geq c\|\dot\gamma(1)\|_h,
\qquad
|g_a(\dot\gamma(0),v)|
\leq C\|\dot\gamma(0)\|_h\|v\|_h.
\]
The first inequality follows by minimizing over null vectors of $h$-norm one. Smooth dependence of geodesics and parallel transport on endpoints in these compact subsets gives $\|\dot\gamma(0)\|_h\leq C\|\dot\gamma(1)\|_h$. Hence $|D_a\alpha_a(z)[v]|\leq C\|v\|_h$ away from the vertex. The norms here are positive auxiliary norms, not null proper lengths.

For nearby tips $a,\tilde a$, let $\delta$ be their Euclidean coordinate distance and join them by the straight coordinate segment in $\overline P$. If its spatial coordinate never equals $z$, integrating the derivative bound gives an arrival-time difference at most $C\delta$. Otherwise, both tips are within spatial distance $\delta$ of $z$. Local cone comparison bounds each arrival time minus its tip's time by $C\delta$; the two tip times also differ by at most $\delta$. This proves the same estimate through the vertex. The compact containments ensure the coordinate segments used in the cone comparison remain in $W$ for small $\delta$. Larger separations are handled by boundedness of the arrival times, increasing $C$. Comparison of coordinate distance with $h$-distance on $\overline P$ proves the stated estimate. Reversing time proves the past-cone case.
\end{proof}

\clearpage
\section{Fixed parameters and constants}\label{app:constants}
\begin{center}
\begingroup
\small
\renewcommand{\arraystretch}{1.35}
\setlength{\tabcolsep}{5pt}
\begin{tabular}{@{}p{.19\linewidth}p{.43\linewidth}p{.32\linewidth}@{}}
\hline
\raggedright Symbol & \raggedright Role & \raggedright Where fixed or used \tabularnewline
\hline
\raggedright $d,\theta$ & \raggedright Dimension and past/future volume threshold. & \raggedright Sections~\ref{sec:geo}, \ref{sec:eligible-bases}. \tabularnewline
\raggedright $K$ or $K_N$ & \raggedright Bounds the counting error in (V1) and (V2). & \raggedright Definition~\ref{def:deterministic-volume-event}, Theorem~\ref{thm:poisson-volume-bounds}. \tabularnewline
\raggedright $s_0$ & \raggedright Largest test-diamond height. & \raggedright Lemma~\ref{lem:small-diamond-stability}. \tabularnewline
\raggedright $A_{\mathrm{br}}$ & \raggedright Height adjustment for comparing nearby diamonds. & \raggedright Lemma~\ref{lem:small-diamond-stability}. \tabularnewline
\raggedright $c_P,C_P$ & \raggedright Compare causal-profile distance with $h$-distance. & \raggedright Lemma~\ref{lem:causal-profile-distance}, \eqref{eq:sample-profile-localization}. \tabularnewline
\raggedright $a,A_{\mathrm{cand}}$ & \raggedright Profile-distance bounds for active and candidate sets. & \raggedright Section~\ref{sec:sampled-charts}, Proposition~\ref{prop:accepted-scaffolds}. \tabularnewline
\raggedright $n_0,n_1,g_0,g_1$ & \raggedright Bounds on interval counts and rescaled Gram eigenvalues. & \raggedright Tests (C2), (C3), Proposition~\ref{prop:accepted-chart-soundness}. \tabularnewline
\raggedright $C_H,R_0,R_1,C_*$ & \raggedright Set coarse-chart and neighborhood sizes. & \raggedright Proposition~\ref{prop:count-chart-guarantees}. \tabularnewline
\raggedright $K_{\mathrm{fr}}$ & \raggedright Bounds comparison-frame norms and inverse norms. & \raggedright Proposition~\ref{prop:count-chart-guarantees}. \tabularnewline
\raggedright $\delta_0,\kappa$ & \raggedright Allowed base spacing and lower covariance bound. & \raggedright Proposition~\ref{prop:count-chart-guarantees}. \tabularnewline
\raggedright $\ell_0$ (loop subdivision) & \raggedright Controls the spacing of ideal loop vertices. & \raggedright Section~\ref{sec:coarse-rectangle}. \tabularnewline
\raggedright $C,C_N$ & \raggedright Error-bound constants; $C_N$ may also depend on $N$. & \raggedright Section~\ref{sec:geo}, Theorem~\ref{thm:volume-curvature}. \tabularnewline
\hline
\end{tabular}
\endgroup
\end{center}

\bibliographystyle{plain}
\bibliography{refs}

\end{document}